\documentclass{article}
\usepackage{fullpage}
\usepackage{url,color,cite}
\usepackage{caption}
\usepackage{amsmath,amsthm,amsfonts,amssymb}
\usepackage{mathtools}
\DeclarePairedDelimiter{\ceil}{\lceil}{\rceil}
\DeclarePairedDelimiter{\floor}{\lfloor}{\rfloor}
\usepackage{float}
\usepackage{epstopdf}
\usepackage{nicefrac}
\usepackage{pifont}
\usepackage{subcaption}
\usepackage{sidecap}
\usepackage{multirow}
\usepackage{bbm}

\newtheorem{lemma}{Lemma}
\newtheorem{theorem}{Theorem}
\theoremstyle{definition}
\newtheorem{ddd}{Definition}
\theoremstyle{definition}
\newtheorem{Example}{Example}
\newtheorem{remark}{Remark}
\theoremstyle{definition}
\newtheorem{corollary}{Corollary}
\theoremstyle{definition}

\usepackage{microtype}
\usepackage{graphicx}
\usepackage{booktabs} 
\usepackage{hyperref}
\usepackage{cleveref}
\usepackage{lipsum}
\usepackage{authblk}

\title{Successive Refinement of Privacy}

\author{Antonious M. Girgis, Deepesh Data, Kamalika Chaudhuri, \\ Christina Fragouli, and Suhas Diggavi\thanks{Antonious M. Girgis, Deepesh Data, Christina Fragouli, and Suhas Diggavi are with the University of California, Los Angeles, USA. Kamalika Chaudhuri is with the University of California, San Diego, USA.
Email: amgirgis@g.ucla.edu, deepesh.data@gmail.com, kamalika@cs.ucsd.edu, christina.fragouli@ucla.edu, suhas@ee.ucla.edu.
This was supported by the NSF grant \#1740047 and by the UC-NL grant LFR-18-548554.}}

\date{}

\begin{document}
\maketitle

\begin{abstract}

This work examines a novel question: how much randomness is needed to achieve local differential privacy (LDP)? A motivating scenario is providing {\em multiple levels of privacy} to multiple analysts, either for distribution or for heavy hitter estimation, using the \emph{same} (randomized) output. We call this setting  \emph{successive refinement of privacy}, as it provides hierarchical access to the raw data with different privacy levels. For example, the same randomized output could enable one analyst to reconstruct the input, while another can only estimate the distribution subject to LDP requirements. This extends the classical Shannon (wiretap) security setting to local differential privacy. We provide (order-wise) tight characterizations of privacy-utility-randomness trade-offs in several cases for distribution estimation, including the standard LDP setting under a randomness constraint. We also provide a non-trivial privacy mechanism for multi-level privacy. Furthermore, we show that we cannot reuse random keys over time while preserving privacy of each user. 
\end{abstract}

\section{Introduction}\label{Intro}

%----------------------------------------------------------------------------------------

Differential privacy~\cite{dwork2006calibrating} -- a cryptographically motivated notion of privacy -- has recently emerged as the gold standard in privacy-preserving data analysis. Privacy is provided by guaranteeing that the participation of a single person in a dataset does not change the probability of any outcome by much; this is ensured by randomness -- either by adding noise to (or randomizing) the raw data itself or to a function or statistic computed directly on the data. If the randomization is large enough relative to the change caused by a single person's data, then their participation is indistinguishable, and privacy is attained. An underlying assumption in the body of work on differential privacy has long been that an unlimited amount of randomness is available for use by any privacy mechanism. Under this assumption, the vast majority of the literature has focused on achieving better privacy-utility trade-offs -- see, for example, \cite{DP-book_DworkRoth14,SarwateChaudhuri_DP-survey13} for surveys. In this paper, we ask: how much randomness do we need to achieve a desired level of privacy and utility, and study privacy-utility-randomness trade-offs instead. Answering this question both contributes to our theoretical understanding, and also could support specific emerging applications that we discuss later in the section.

We consider local differential privacy (LDP) -- a privacy model that has recently seen use in industrial applications, \cite[RAPPOR]{Rappor}, \cite{Apple_DP}. Here, an untrusted analyst acquires already-privatized pieces of information from a number of users, and aggregates them into a statistic or a machine learning model. Concretely, there are $n$ users who observe i.i.d.~inputs $X_1, X_2,\hdots, X_n$ (user $i$ observes $X_i$) from a finite alphabet $\mathcal{X}$ of size $k$, where each $X_i$ is distributed according to a probability distribution $\mathbf{p}$.  Each user has a certain amount of randomness, measured in Shannon entropy, to randomize her input, that she then publicly shares. Our general setup also includes $d$ analysts who would like to use the users' public outputs to estimate $\mathbf{p}$, each at a different level of privacy $\epsilon_1, \ldots,\epsilon_d$, where smaller $\epsilon$ means higher privacy. Each analyst may or may not share some common randomness with the users. We call this general setup {\em successive refinement of privacy}, in which each user shares a public output with highest privacy level.
Then, each analyst uses a shared random key to partially undo the randomization of the public output to get less privacy and higher utility. 

This general formulation includes several interesting special cases, for which we study the trade-offs between privacy, utility, and randomness. These are:

{\sf (i)} There is a single analyst ($d=1$), who shares no randomness with the users and estimates $\mathbf{p}$ with privacy level $\epsilon$. This setting directly generalizes the classical setup of LDP to the case of limited randomness.

{\sf (ii)} There are two analysts ($d=2$), who observe the same public outputs from the users; the first analyst who shares common randomness with the users has permission to perfectly recover the original inputs (i.e., privacy level $\epsilon_1\to \infty$), while the second analyst who shares no randomness with the users estimates $\mathbf{p}$ with privacy level $\epsilon_2$. This setting is an adaptation of the classical {\em perfect secrecy} setup of Shannon \cite{shannon1949} to the differential privacy world. In Shannon's setup, Alice (users) wants to send a secret to Bob (the first analyst), which must remain perfectly private from Eve (the second analyst); whereas, in our setting, instead of complete independence, we only want that the secret remains hidden from Eve in the sense of differential privacy. We call this setup {\em private-recoverability}.

{\sf (iii)} There are $d>1$ analysts, who share some common randomness with the users. Analyst $i$ would like to estimate $\mathbf{p}$ with privacy level $\epsilon_i$, where $\epsilon_1>\ldots>\epsilon_d$.\footnote{We can assume, without loss of generality, that $\epsilon_j>\epsilon_{j+1}, \forall j\in[d-1]$; otherwise, we can group the equal $\epsilon_j$'s together and the corresponding analysts can use the same privatized data that the users share with them.}

\begin{figure}[t]
  \centerline{\includegraphics[scale=0.3]{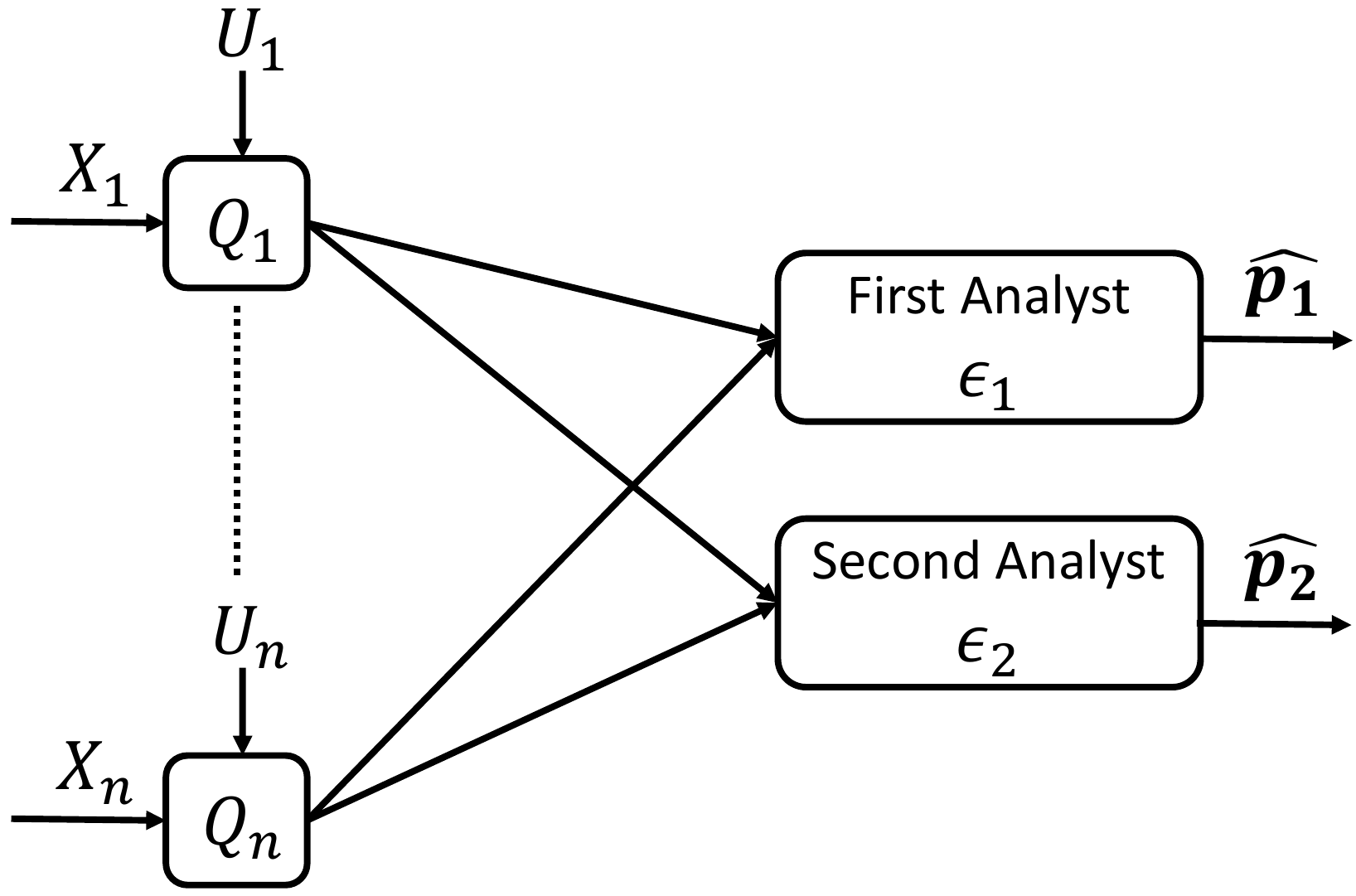}}
 \caption{We have $n$ users, each observing a sample $X_i$. A private randomization mechanism $Q_i$ is applied to $X_i$ using a random key $U_i$. Two analysts want to estimate $\mathbf{p}$. Each analyst requires a different privacy level.}
\label{Fig1_1}
\end{figure}

\subsection{Motivation} 

In general, designing private mechanisms with a small amount of randomness can be translated into communication efficiency and/or storage efficiency. For instance, when there are multiple privacy levels, each user needs to send additional information to some analysts, that is a function of the randomness used in the mechanism. Hence, using a smaller amount of randomness implies delivering a smaller number of bits to each analyst.

The private-recoverability setup ($d=2$) can be useful in applications such as census surveys, \cite{Dowrk_US-Census19}, that collect large amounts of data and are prohibitively expensive to repeat. Using our approach, we can store the randomized data on a public database (second analyst) without compromising the privacy of individuals; we can also give to the first analyst (e.g., the government, who may wish to exactly calculate the population count, or verify the validity of census results) a secret key, that can be used  to ``de-randomize'' the publicly stored data and perfectly reconstruct the user inputs. An alternative approach would be to store the data twice (once randomized in a public database and once in a secure government database), which would incur an additional storage cost, as also shown in Section \ref{Num}. Another alternative would be to use a cryptographic scheme to encode the user inputs; in this case, the resulting outputs may not allow public use in an efficient manner.\footnote{In principle, we could use homomorphic encryption that allows to compute a function on the encrypted data without decrypting it explicitly; however, such encryption schemes are computationally inefficient and expensive to deploy.}

The multi-level privacy $ d>1 $ illustrates a new technical capability of hierarchical access to the raw data that might inspire and support a variety of applications. For example, given data collected from a fleet of autonomous cars, we could imagine different privacy access levels provided to the car manufacturer itself, to police departments, to applications interested in online traffic regulation, to applications interested in long-term traffic predictions or road planning. Essentially, this capability enables providing the desired utility needed for each application while maintaining the maximum possible amount of privacy.

\subsection{Contributions} 

Our contributions are as follows.

$\bullet$  For the single analyst case $(d=1)$, we characterize the trade-off between randomness and utility for a fixed privacy level $\epsilon$, by proving an information-theoretic lower bound and a matching upper bound for a minimax private estimation problem.

$\bullet$ For private-recoverability $(d=2)$, we derive an information-theoretic lower bound on the minimum randomness required to achieve it, and prove that the Hadamard scheme proposed in~\cite{acharya2018hadamard} is order optimal. We also show that we cannot reuse random keys over time while preserving privacy of each user. Hence, to preserve privacy of $T$ samples, any $\epsilon$-DP mechanism has to use an amount of randomness equal to $T$ times the amount of randomness used for a single data sample. We also extend this result to estimating {\em heavy hitters}.

$\bullet$ In the multi-level privacy $(d>1)$ setting, a trivial scheme is to use the $d=1$ scheme multiple times, separately for each analyst. We propose instead a non-trivial scheme that uses a smaller amount of randomness with no sacrifice in utility. Our scheme publicly announces the users' outputs, and allows each analyst to remove an appropriate amount of (shared) randomness with the help of an associated key. This approach enables efficient hierarchical access to the data (for example, when analysts have different levels of authorized access).

Overall, our investigation into privacy-utility-randomness trade-offs for LDP yields (optimal) privacy mechanisms that use randomness more economically. These include new guarantees for existing schemes such as the Hadamard mechanism, as well as new multi-user and multi-level mechanisms that allow for hierarchically private data access. 

\subsection{Related work}

To the best of our knowledge, the role of limited randomness has not been previously explored either in the context of local or global differential privacy.\footnote{Except for a notable exception of \cite{imperfect-randomness_DP12}, which showed that imperfect source of randomness allows efficient protocols with global differential privacy. This is different from our problem, where our goal is to quantify the amount of randomness required (measured in terms of Shannon entropy) in local differential privacy and give privacy-utility-randomness trade-offs.}
In this work, we consider local differential privacy in the context of distribution estimation and heavy hitter estimation for reasons of simplicity.

Popular local differentially private mechanisms for distribution estimation include RAPPOR~\cite{Rappor}, randomized response (RR)~\cite{warner1965randomized}), subset selection (SS)~\cite{Ye2018,wang2016mutual}, and the Hadamard response (HR)~\cite{acharya2018hadamard}. The randomized response mechanism is known to be order optimal in the low privacy regime, and the RAPPOR scheme in the high privacy regimes~\cite{kairouz2016discrete,kairouz2014extremal}. Subset selection and the Hadamard mechanisms are order optimal in utility for all privacy regimes; additionally, the Hadamard mechanism has the advantage of communication and computational efficiency for all privacy regimes~\cite{acharya2018hadamard}. We build on this extensive literature, and show that the Hadamard mechanism is also near-optimal in terms of the amount of randomness used. 

Heavy hitter estimation under local differential privacy has been studied in~\cite{bassily2015local,qin2016heavy,hsu2012distributed, bassily2017practical, bun2018heavy}, again with unrestricted randomness. Our work adds to this line of work by showing that the Hadamard mechanism is capable of achieving order-optimal accuracy for heavy hitter estimation {\em{while}} using an order-optimal amount of randomness. 

Local differential privacy in a multi-user setting where the users and the server may have some shared randomness has also been looked at in prior work -- see ~\cite{bassily2015local, acharya2019communication, acharya2018test} among others. These works however investigate other orthogonal aspects of such multi-user protocols. Local differentially private mechanisms with bounded communication have also been studied by~\cite{acharya2019communication}; in their setup, multiple agents transmit their data in a locally private manner to an aggregator, and communication is measured by the number of bits transmitted by each user. They consider both private and public coin mechanisms, and show that the Hadamard mechanism is near optimal in terms of communication for both distribution and heavy-hitter estimation; however, unlike ours, their mechanisms do not impose any randomness constraints.

Our results in the multiple analyst setting are also related to privacy amplification by stochastic postprocessing~\cite{balle2019privacy} -- which analyzes the privacy risk achieved by applying a (stochastic) post-processing mechanism to the output of a differentially private algorithm. While these methods might also be used to provide multi-level privacy to multiple analysts, our work is different from~\cite{balle2019privacy} in the following aspect. First, their privacy amplification methodology does not apply to pure DP and applies instead to approximate DP, while our work focuses on pure DP. Second, the work in~\cite{balle2019privacy} does not include a randomness constraint, and finally, a closer look at their mechanism reveals that it does not use the optimal amount of randomness. 

Finally, a line of work on locally differentially private estimation considers the case when the inputs comprise of i.i.d.\ samples from the same distribution. ~\cite{duchi2018minimax,duchi2019lower} derive lower and upper bounds for estimation under LDP in this setting -- their work considers that all users observe i.i.d.\ samples from the same distribution, and the goal for each user is to preserve privacy of its raw sample. Our work is also different from this setting in that we focus on designing private mechanisms with finite randomness.

\subsection{Paper organization}

Section~\ref{Prelm} formally defines LDP
mechanisms under randomness constraints and presents the distribution
and heavy hitter estimation problem formulations. Section~\ref{Res}
states our main results for the single-level privacy,
private-recoverability, and multi-level privacy
settings. Section~\ref{Num} presents numerical evaluations on the
effect of parameters such as $n,\epsilon,d$ on the estimation error
and the required randomness. Section~\ref{LDP} derives an information-theoretic 
lower bound and an upper bound (achievability scheme) on the
minimax risk estimation under randomness and privacy constraints for a
single analyst. Section~\ref{Privlv} proposes a new LDP mechanism for
the multi-level privacy $d>1 $. Section~\ref{Recov} presents the
necessary and sufficient conditions on the randomness to design an
$\epsilon$-LDP mechanism with input recoverability
requirement. Section~\ref{Recov_GDP} introduces the necessary and
sufficient conditions on the randomness to preserve privacy of a
sequence of samples per user. 

\section{Preliminaries and Problem Formulation} \label{Prelm}

\textbf{Notation:} We use $\left[k\right]$ to define the set $\lbrace 1,\ldots,k\rbrace$ of integers. We use uppercase letters $X,Y$, etc., to denote random variables, and lowercase letter $x,y$, etc., to denote their realizations. 
For any two distributions $\textbf{p}$ and $\textbf{q}$ supported over a set $\mathcal{X}$, let 
$\|\textbf{p}-\textbf{q}\|_{\text{TV}} = \sup_{\mathcal{A}\subseteq\mathcal{X}}|\textbf{p}(\mathcal{A})-\textbf{q}(\mathcal{A})|$ be the total variation distance between $\textbf{p}$ and $\textbf{q}$. 
We use $\oplus$ to define the XOR operation. For $p\in[0,1]$, we use $H_2\left(p\right)$ to denote the binary entropy function defined by $H_2\left(p\right)=-p\log\left(p\right)-\left(1-p\right)\log\left(1-p\right)$, and $H\left(X\right)$ to denote the entropy of the random variable $X$. Also, we use $H\left(\textbf{p}\right)$ to denote the entropy of a random variable $X$ drawn from a distribution $\textbf{p}$.

\subsection{Differential Privacy (LDP)} Let $\mathcal{X}\triangleq\lbrace 1,\ldots,k\rbrace$ be an input alphabet and $\mathcal{Y}\triangleq\lbrace 1,\ldots,m\rbrace$ be an output alphabet, of sizes  $|\mathcal{X}|=k$ and  $|\mathcal{Y}|=m$, respectively, that are not required to be the same. A private randomization mechanism $Q$ is a conditional distribution that takes an input $X\in\mathcal{X}$ and generates a privatized output $Y\in\mathcal{Y}$. $Q$ is said to satisfy the $\epsilon$-local differential privacy ($\epsilon$-LDP)~\cite{duchi2013local}, if for every pair of inputs $x,x^{\prime}\in\mathcal{X}$, we have
\begin{equation}~\label{eqn1_3}
\sup_{y\in\mathcal{Y}}\frac{Q\left(y|x\right)}{Q\left(y|x^{\prime}\right)}\leq \exp\left(\epsilon\right),
\end{equation}
where $Q\left(y|x\right)=\text{Pr}\left[Y=y|X=x\right]$ and $\epsilon$ captures the privacy level. For small values of $\epsilon$, the adversary cannot infer whether the input was $X=x$ or $X=x^{\prime}$. Hence, a smaller privacy level $\epsilon$ implies higher privacy.

\subsection{Randomness in LDP Mechanisms}  
A private mechanism $Q$ with input $X\in\mathcal{X}$ and output $Y\in\mathcal{Y}$ is said to satisfy $\left(\epsilon,R\right)$-LDP, if for every pair of inputs $x,x^{\prime}\in\mathcal{X}$, we have
\begin{equation}
\begin{aligned}
&\sup_{y\in\mathcal{Y}}\frac{Q\left(y|x\right)}{Q\left(y|x^{\prime}\right)}\leq \exp\left(\epsilon\right),\ \text{and}\\
&H\left(Y|X=x\right)\leq R\quad \forall x\in\mathcal{X},
\end{aligned}
\end{equation}
where $H\left(Y|X=x\right)=\sum_{y\in\mathcal{Y}}Q\left(y|x\right)\log\left(\frac{1}{Q\left(y|x\right)}\right)$ denotes the entropy of the random output $Y$ conditioned on the input $X=x$. Note that an $\left(\epsilon,R\right)$-LDP mechanism is an $\epsilon$-LDP mechanism that requires an amount of randomness less than or equal to $R$-bits to be designed.

Suppose that a random key $U$ with $H\left(U\right)\leq R$ is used to design an $\left(\epsilon,R\right)$-LDP mechanism $Q$. We consider $U$ to be a random variable that takes values from a discrete set $\mathcal{U}=\lbrace u_1,\ldots,u_l\rbrace$ according to a distribution $\mathbf{q}=\left[q_1,\ldots,q_l\right]$, where $q_u=\text{Pr}\left[U=u\right]$ for $u\in\mathcal{U}$. We assume that $\mathcal{U}$ is a discrete set, since we focus on finite randomness. Let $\mathcal{U}_{yx}\subset\mathcal{U}$ be a subset of key values such that input $X=x$ is mapped to $Y=y$ when $u\in\mathcal{U}_{yx}$. The private mechanism $Q$ can be represented as 
\begin{equation}\label{eqn2_4}
Q\left(y|x\right)=\sum_{u\in\mathcal{U}_{yx}}q_u.
\end{equation}  
Note that the output $Y$ is a function of $\left(X,U\right)$. Therefore, we have $\mathcal{U}_{y^{\prime}x}\bigcap \mathcal{U}_{yx}=\phi$ for $y^{\prime}\neq y$, since there is only one output for each input. In addition, if we want \eqref{eqn2_4} to satisfy the privacy condition \eqref{eqn1_3}, we also have\footnote{Otherwise we can distinguish inputs causing $\epsilon\rightarrow\infty$.} $\bigcup_{y\in\mathcal{Y}}\mathcal{U}_{yx}=\mathcal{U}$ for each $x\in\mathcal{X}$. We will leverage this representation of randomness in LDP mechanisms to design multi-level privacy mechanisms. Figure~\ref{Fig2_1} shows an example of designing a private mechanism with binary inputs $\mathcal{X}=\lbrace 0,1\rbrace$, binary random keys $\mathcal{U}=\lbrace 0,1\rbrace$, and binary outputs $\mathcal{Y}=\lbrace 0,1\rbrace$. In this example, we can represent the output of the mechanism as a function of $\left(X,U\right)$ by $Y=X\oplus U$, where $\oplus$ denotes the XOR operation. If the random key $U$ is drawn from a distribution $\mathbf{q}=\left[\frac{e^{\epsilon}}{e^{\epsilon}+1},\frac{1}{e^{\epsilon}+1}\right]$, then it is easy to show that the mechanism is $\epsilon$-LDP.

\begin{figure}[t]
  \centerline{\includegraphics[scale=0.4,trim={0.5cm 8.9cm 9cm 0.5cm},clip]{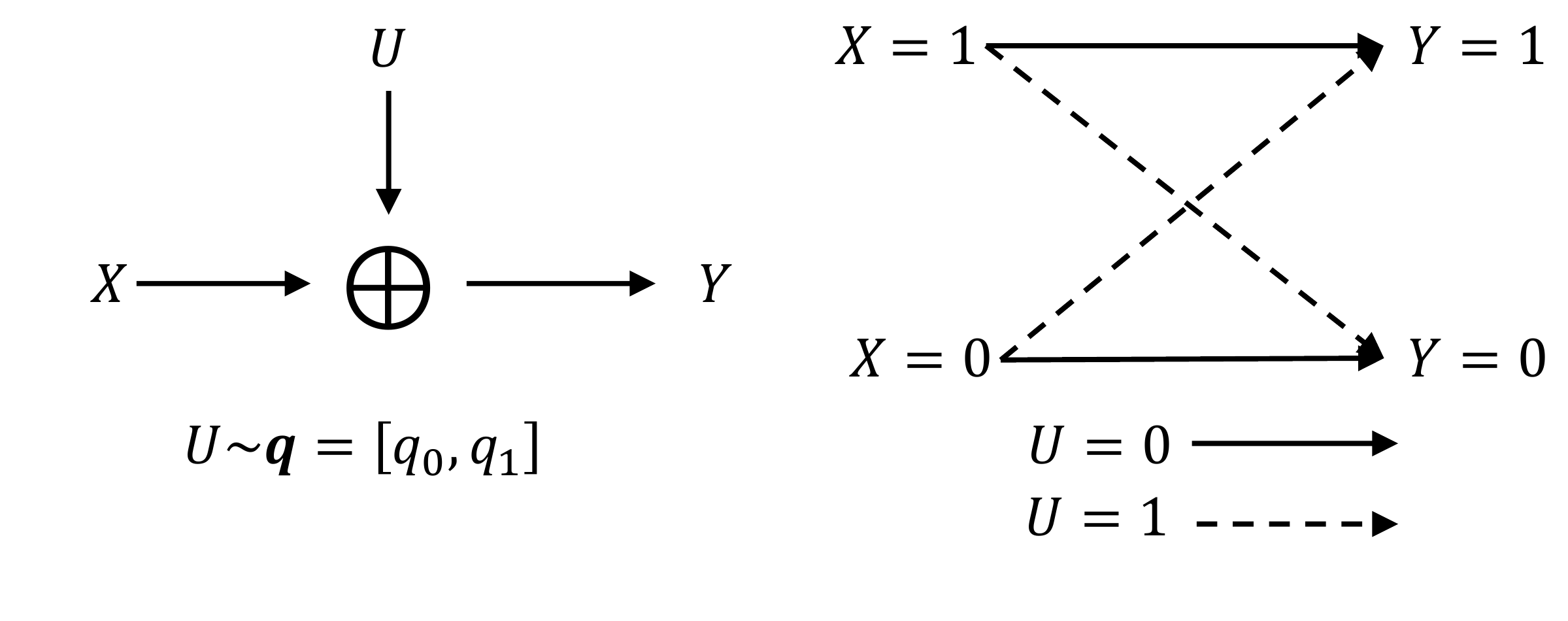}}
 \caption{An example of designing an $\epsilon$-LDP mechanism using a private key: (left) representing the output $Y$ of the mechanism $Q$ as a function of the input $X$ and the private key $U$, (right) representing the mechanism $Q$ as a probabilistic mapping from the input $X$ to the output $Y$ depending on the private key $U$.}
\label{Fig2_1}
\end{figure}

\subsection{Problem Formulation}  
We consider $n$ users who observe i.i.d.~inputs $X_1, X_2,\hdots, X_n$ (user $i$ observes input $X_i$),  drawn from an unknown discrete distribution $\mathbf{p}\in \Delta_k$, where $\Delta_k=\left\{ \mathbf{p}\in\mathbb{R}^{k}| \sum_{j=1}^{k}p_j=1, p_j\geq 0,\ \forall j\in\left[k\right]\right\}$ denotes the probability simplex over $\mathcal{X}$.
The $i$'th user has a random key $U_i$ with $H\left(U_i\right)\leq R$; we assume that $U^{n}=\left[U_1,\ldots,U_n\right]$ are independent random variables, unless otherwise stated. The $i$'th user generates (and publicly shares) an output $Y_i$, using 
an $\left(\epsilon,R\right)$-LDP mechanism $Q_i$ and her random key $U_i$.
The output $Y_i$  has a marginal distribution given by 
\begin{equation}
\mathbf{M}_{i}\left(y|\mathbf{p}\right)=\sum_{x\in\mathcal{X}}Q_i\left(y|x\right)p_x\qquad\forall y\in\mathcal{Y}_i,
\end{equation}
where  $\mathcal{X}$ and $\mathcal{Y}_i$ are the input and output alphabets. We also have $d$ analysts who want to use the users' public outputs  $Y^{n}=\left[Y_1,\ldots,Y_n\right]$ to estimate $\mathbf{p}$, each at a different level of privacy  $\epsilon_1 > \ldots > \epsilon_d$. The system model is shown in Figure~\ref{Fig1_1}.

\textbf{\em Risk Minimization:}
For simplicity of exposition, consider for now a single analyst, and
let $\mathbf{\hat{p}}=\left[\hat{p}_1,\cdots,\hat{p}_{k}\right]$ denote the analyst's estimator (this is a function $\mathbf{\hat{p}}:Y^{n}\to \mathbb{R}^{k}$ that maps the outputs $Y^{n}$ to a distribution in the simplex $\Delta_k$)\footnote{Observe that it is sufficient to consider a deterministic estimator $\hat{\mathbf{p}}$, since for any randomized estimator, there exists a deterministic estimator that dominates the performance of the randomized one.}. For given private mechanisms $Q^{n}=\left[ Q_1,\ldots,Q_n\right]$, the estimator $\hat{\mathbf{p}}$ is obtained by solving the problem 
\begin{equation}
r^{\ell}_{\epsilon,R,n,k}\left(Q^{n}\right)=\inf_{\hat{\mathbf{p}}}\sup_{\mathbf{p}\in\Delta_k}\mathbb{E}\left[\ell\left(\hat{\mathbf{p}}\left(Y^{n}\right),\mathbf{p}\right)\right],
\end{equation} 
where $r^{\ell}_{\epsilon,R,n,k}$ is the minimax risk, the expectation is taken over the randomness in the outputs $Y^{n}=\left[Y_1,\ldots,Y_n\right]$ with $Y_i\sim\mathbf{M}_i$, and $\ell:\mathbb{R}^{k}\times\mathbb{R}^{k}\to\mathbb{R}_{+}$ is a loss function  that measures the distance between two distributions in $\Delta_k$. Unless otherwise stated, we adopt as loss function the 1-norm, namely $\ell=\ell_1$ and the squared 2-norm, namely $\ell=\ell_2^{2}$. Our task is to design private mechanisms $ Q_1,\hdots,Q_n$ that minimize the minimax risk estimation, namely,
\begin{equation}~\label{eqn1_1}
\begin{aligned}
r^{\ell}_{\epsilon,R,n,k}&=\inf_{\lbrace Q_i\in\mathcal{Q}_{\left(\epsilon,R\right)}\rbrace}\ r^{\ell}_{\epsilon,R,n,k}\left(Q^{n}\right)\\
&=\inf_{\lbrace Q_i\in\mathcal{Q}_{\left(\epsilon,R\right)}\rbrace} \inf_{\hat{\mathbf{p}}}\sup_{\mathbf{p}\in\Delta_k}\mathbb{E}\left[\ell\left(\hat{\mathbf{p}}\left(Y^{n}\right),\mathbf{p}\right)\right],
\end{aligned}
\end{equation} 
where $\mathcal{Q}_{\left(\epsilon,R\right)}$ denotes the set of mechanisms that satisfy $\left(\epsilon,R\right)$-LDP. Observe that when $R\to\infty$, the problem~\eqref{eqn1_1} is reduced to the standard LDP distribution estimation studied previously in~\cite{duchi2013local, kairouz2016discrete,Ye2018,acharya2018hadamard}. The difference in the  formulation in~\eqref{eqn1_1} is the randomness constraint. 

\textbf{\em  LDP heavy hitter estimation:}
In  heavy hitter estimation, the input samples $X^{n}=[X_1,\hdots,X_n]$ do not have an associated distribution. Furthermore, the analyst is interested in estimating the frequency of each element $x\in\mathcal{X}$ with the infinity norm being the loss function (i.e., $\ell=\ell_{\infty}$). Frequency of each element $x\in\mathcal{X}$ is defined by $f\left(x\right)=\frac{\sum_{i=1}^{n}\mathbbm{1}\left(X_i=x\right)}{n}$. We then want to calculate
\begin{equation}~\label{eqn1_1_1}
\begin{aligned}
  &r^{\ell_{\infty}}_{hh,\epsilon,R,n,k}=  &\inf_{\lbrace Q_i\in\mathcal{Q}_{\left(\epsilon,R\right)}\rbrace} \inf_{\hat{\mathbf{p}}}\sup_{X^{n}\in\mathcal{X}^{n}}\mathbb{E}\left[\max_{x\in\mathcal{X}}|\hat{p}_x\left(Y^{n}\right)-f\left(x\right)|\right],\nonumber
\end{aligned}
\end{equation} 
where the expectation is taken over the randomness in the outputs $Y^{n}=\left[Y_1,\ldots,Y_n\right]$ and $\hat{\mathbf{p}}$ denotes the estimator of the analyst. Note, again, that in this case we do not make any distributional assumptions on $X_1,\hdots,X_n$.

\textbf{\em Multi-level privacy:} Consider now the general case of $d$ analysts each operating at a different level of privacy $\epsilon_1 > \ldots > \epsilon_d$. All analysts observe the users' public outputs  $Y^{n}$; additionally, analyst $j$ may also observe some side information on the user randomness. The question we ask is: what is the minimum amount of randomness $U$ per user required to maintain the privacy of each user while achieving the minimum risk estimation for each analyst?

\textbf{\em Sequence of distribution (or heavy hitter) estimation:}
We assume that each user $i$  has a random key $U_i$  to preserve the privacy of a sequence of $T$ independent samples $X_{i}^{\left(1\right)},\ldots,X_{i}^{\left(T\right)}$, where the $t$'th samples for $t\in\left[T\right]$ at all users are drawn i.i.d.\ from an unknown distribution $\mathbf{p}^{\left(t\right)}$.\footnote{As mentioned earlier, for heavy hitter estimation, the samples $X_{i}^{\left(1\right)},\ldots,X_{i}^{\left(T\right)}$ do not have an associated distribution.}  At time $t$, the $i$'th user  generates an output $Y_{i}^{(t)}$ that  may be a function of the random key $U_i$ and all input samples $\lbrace X_{i}^{\left(m\right)}\rbrace_{m=1}^{t}$.
Each of the $d$ analysts uses the outputs $Y_{i}^{(t)}, i\in[n], t\in[T]$ 
to estimate $T$ distributions $\mathbf{p}^{\left(1\right)},\ldots,\mathbf{p}^{\left(T\right)}$ (or estimate the heavy hitters).

A private mechanism $Q$ with a sequence of inputs $X^{T}=\left( X^{\left(1\right)},\ldots,X^{\left(T\right)}\right)$ and a sequence of outputs $Y^{T}=\left( Y^{\left(1\right)},\ldots,Y^{\left(T\right)}\right)$ is said to satisfy $\epsilon$-DP, if for every neighboring databases $\mathbf{x},\mathbf{x}^{\prime}\in\mathcal{X}^{T}$, we have
\begin{equation}\label{eqn1_2}
\sup_{\mathbf{y}\in\mathcal{Y}^{T}}\frac{Q\left(\mathbf{y}|\mathbf{x}\right)}{Q\left(\mathbf{y}|\mathbf{x}^{\prime}\right)}\leq \exp\left(\epsilon\right),
\end{equation}
where $Q\left(\mathbf{y}|\mathbf{x}\right)=\text{Pr}\left[Y^{T}=\mathbf{y}|X^{T}=\mathbf{x}\right]$;  and we say that two databases, $\mathbf{x}=\left(x^{\left(1\right)},\ldots,x^{\left(T\right)}\right)$ and $\mathbf{x}^{\prime}=\left(x^{\prime\left(1\right)},\ldots,x^{\prime\left(T\right)}\right)\in\mathcal{X}^{T}$ are neighboring, if there exists an index $t\in\left[T\right]$, such that $x^{\left(t\right)}\neq x^{\prime\left(t\right)}$ and $x^{\left(l\right)}= x^{\prime\left(l\right)}$ for $l\neq t$. Observe that when $T=1$, the definition of $\epsilon$-DP in~\eqref{eqn1_2} coincides with the definition of $\epsilon$-LDP in~\eqref{eqn1_3}.
 We are interested in  the question: Is there  a private mechanism that uses a smaller amount of randomness than $T$ times the amount of randomness used for a single data sample? In other words, can we perhaps reuse the randomness over time while preserving privacy? 

\section{Main Results} \label{Res} 
This section formally presents our main results. First, we characterize the minimax risk estimation under randomness and privacy constraints in Theorems~\ref{Th2_1} and~\ref{Th2_2} for single-level privacy ($d=1$). Then, we propose in Theorem~\ref{Th2_3} a new LDP privacy mechanism that provides a hierarchical access to users' samples with different privacy levels (multi-level privacy $d>1$). We present in Theorem~\ref{Th2_4} the necessary and sufficient conditions on the randomness to design an LDP mechanism with input recoverability requirement. Finally, we present in Theorem~\ref{Th2_6} the necessary and sufficient conditions on the randomness to preserve privacy of a sequence of samples under a recoverability constraint.

\subsection{Single-level Privacy, $d=1$}
We here study the fundamental trade-off between randomness and utility for a fixed privacy level $\epsilon$. In the following theorem, we derive a lower bound on the minimax risk estimation $r^{\ell_{2}^{2}}_{\epsilon,R,n,k}$ and $r^{\ell_{1}}_{\epsilon,R,n,k}$ defined in~\eqref{eqn1_1}.
\begin{theorem}~\label{Th2_1} For every $\epsilon,R\geq0$ and $k,n\in\mathbb{N}$, the minimax risk under $\ell_2$-norm loss is bounded by
\begin{equation}
r^{\ell_{2}^{2}}_{\epsilon,R,n,k}\geq \tau=
%\left\{ \begin{array}{ll}
\begin{cases}
\frac{k\left(e^{\epsilon}+1\right)^2}{16ne^{\epsilon}\left(e^{\epsilon}-1\right)^{2}}& \text{if}\ R\geq H_2\left(\frac{e^{\epsilon}}{e^{\epsilon}+1}\right),\\
\frac{ke^{\epsilon}}{16 n p_R^2 \left( e^{\epsilon}-1\right)^{2}}& \text{if}\ R< H_2\left(\frac{e^{\epsilon}}{e^{\epsilon}+1}\right),
\end{cases}
%\end{array}
%\right.
\end{equation}
where $p_{R}\leq 0.5$ is the inverse of  the binary entropy function $p_R=H_2^{-1}\left(R\right)$. The minimax risk under $1$-norm loss is bounded by $r^{\ell_{1}}_{\epsilon,R,n,k}\geq \sqrt{k\tau/8}$.
\end{theorem}

The main contribution in our proof (see Section~\ref{LDP_AS})
is a formulation of a non-convex optimization problem to bound the minimax risk  under privacy and randomness constraints, and obtaining a tight bound on its solution for every value of privacy level $\epsilon$ and randomness $R$. 
\begin{remark} 
In~\cite{Ye2018}, the authors derive the following lower bound on the minimax risk estimation without randomness constraints ($R\to\infty$)
\begin{equation}~\label{eqn2_1}
r^{\ell_{2}^{2}}_{\epsilon,\infty,n,k}\geq \left\{ \begin{array}{ll}
\frac{k\left(e^{\epsilon}+1\right)^2}{512n\left(e^{\epsilon}-1\right)^{2}}& \text{for}\ e^{\epsilon}< 3, \\
\frac{k}{64 n \left( e^{\epsilon}-1\right)}& \text{for}\ e^{\epsilon}\geq 3.
\end{array}
\right.
\end{equation}
For $\epsilon=\mathcal{O}(1)$ and $R\geq H_2\left(\frac{e^{\epsilon}}{e^{\epsilon}+1}\right)$ (which includes  $R\to\infty$  as well), our lower bound from Theorem~\ref{Th2_1} gives $r^{\ell_{2}^{2}}_{\epsilon,R,n,k}=\Omega\left(\frac{k}{n\epsilon^{2}}\right)$, which 
coincides with (\ref{eqn2_1}). However, our lower bound is tighter for all values of $\epsilon\in\left[0,\infty\right)$ with smaller constant factors. 
\end{remark}

We next show that there exists an achievable scheme for all values of $\epsilon,R\geq0$ that matches (up to a constant factor) the lower bound given in Theorem~\ref{Th2_1} for $\epsilon=\mathcal{O}\left(1\right)$ and $R\geq 0$.
\begin{theorem}~\label{Th2_2} For any $\epsilon,R\geq0$, there exists $\left(\epsilon,R\right)$-LDP mechanisms $Q_1,\hdots,Q_n$ and an estimator $\hat{\mathbf{p}}$ such that the error $\mathcal{E}:=\sup_{\mathbf{p}\in\Delta_k}\mathbb{E}\left[\|\hat{\mathbf{p}}\left(Y^{n}\right)-\mathbf{p}\|_{2}^{2}\right]$ is bounded by
%\begin{small}
\begin{equation}~\label{eqn2_3}
\mathcal{E} \leq \eta=
\left\{ \begin{array}{ll}
\frac{2k \left(e^{\epsilon}+1\right)^{2}}{n \left(e^{\epsilon}-1\right)^2} & \text{if}\ R\geq H_2\left(\frac{e^{\epsilon}}{e^{\epsilon}+1}\right), \\
\frac{2k e^{2\epsilon}}{np_R^{2} \left(e^{\epsilon}-1\right)^2} & \text{if}\ R< H_2\left(\frac{e^{\epsilon}}{e^{\epsilon}+1}\right).
\end{array}
\right.
\end{equation}
The error under $\ell_1$-norm loss is bounded by $\sup_{\mathbf{p}\in\Delta_k}\mathbb{E}\left[\|\hat{\mathbf{p}}\left(Y^{n}\right)-\mathbf{p}\|_{1}\right]\leq \sqrt{k\eta}$.
\end{theorem}
We prove Theorem~\ref{Th2_2} constructively in Section~\ref{LDP-AC}, by adapting the Hadamard response scheme given in~\cite{acharya2019communication} to our setting of limited randomness.
Theorems~\ref{Th2_1} and \ref{Th2_2} together imply the following characterization for $r^{\ell_{2}^{2}}_{\epsilon,R,n,k}$ and $r^{\ell_{1}}_{\epsilon,R,n,k}$, for the case when $\epsilon=\mathcal{O}(1)$:
\begin{corollary}~\label{Cor2_1} For $\epsilon=\mathcal{O}\left(1\right)$ and $R\geq0$, we have
\begin{equation}
r^{\ell_{2}^{2}}_{\epsilon,R,n,k}=
%\left\{\begin{array}{ll}
\begin{cases}
\Theta\left(\frac{k}{n\epsilon^{2}}\right) & \text{if}\ R\geq H_2\left(\frac{e^{\epsilon}}{e^{\epsilon}+1}\right),\\
\Theta\left(\frac{k}{np_R^2\epsilon^{2}}\right) & \text{if}\ R < H_2\left(\frac{e^{\epsilon}}{e^{\epsilon}+1}\right),
\end{cases}
%\end{array} 
%\right.
\end{equation}
and $r^{\ell_{1}}_{\epsilon,R,n,k}=\sqrt{k r^{\ell_{2}^{2}}_{\epsilon,R,n,k}}$. 
\end{corollary}
We next provide a comparison between well-known mechanisms from randomness perspective. Table~\ref{t1} describe the amount of randomness required to implement different $\epsilon$-LDP mechanisms: RAPPOR~\cite{Rappor}, Randomized Response (RR)~\cite{warner1965randomized}, Hadamard Response (HR)~\cite{acharya2018hadamard}, and Binary Hadamard (BH)~\cite{acharya2019communication}.

\begin{table}[t!]
\centering
\begin{tabular}{ |c || c | c | c | c |   }
\hline
 & RAPPOR & RR & HR & BH \\ 
 \hline\hline
 Randomness per user ($R$ in bits)  & $kH_2\left(\frac{e^{\epsilon}}{e^{\epsilon}+1}\right)$  & $\log\left(k-1+e^{\epsilon}\right)-\frac{\epsilon e^{\epsilon}}{k-1+e^{\epsilon}}$ & $\leq\log\left(2k\frac{3e^{\epsilon}-1}{e^{\epsilon}}\right)-\frac{\epsilon e^{\epsilon}}{3e^{\epsilon}-1}$ & $H_2\left(\frac{e^{\epsilon}}{e^{\epsilon}+1}\right)$ \\  
 \hline
 Minimax risk ($r_{\epsilon,R,n,k}^{\ell_{2}^{2}}$) & $\mathcal{O}\left(\frac{k}{n\epsilon^2}\right)$ & $\mathcal{O}\left(\frac{k^2}{n\epsilon^2}\right)$ & $\mathcal{O}\left(\frac{k}{n\epsilon^2}\right)$ & $\mathcal{O}\left(\frac{k}{n\epsilon^2}\right)$  \\
 \hline  
\end{tabular}
\caption{Randomness requirement to implement each private mechanism and its corresponding minimax risk under $\ell_{2}^{2}$ loss function for $\epsilon=\mathcal{O}\left(1\right)$.}~\label{t1}
\end{table}

Observe that all private mechanisms are order optimal in the high privacy regime except for the RR scheme. However, only the BH scheme uses the smallest amount of randomness $R=H_2\left(\frac{e^{\epsilon}}{e^{\epsilon}+1}\right)$ per user, while the other mechanisms require a larger amount of randomness. Table~\ref{t1} considers only the regime of randomness $R\geq H_2\left(\frac{e^{\epsilon}}{e^{\epsilon}+1}\right)$, since the privacy-utility trade-off when the amount of randomness $R< H_2\left(\frac{e^{\epsilon}}{e^{\epsilon}+1}\right)$ has not been studied before. Corollary~\ref{Cor2_1} characterizes the privacy-utility trade-offs for all regions of randomness $R$. 

\begin{remark}\label{remark_critical-rand} Observe that when $R< H_2\left(\frac{e^{\epsilon}}{e^{\epsilon}+1}\right)$, there exists a trade-off between $R$ and $r^{\ell_{2}^{2}}_{\epsilon,R,n,k}$ -- as $R$ increases, $r^{\ell_{2}^{2}}_{\epsilon,R,n,k}$ decreases proportionally to $1/p_R^{2}$. However, when $R\geq H_2\left(\frac{e^{\epsilon}}{e^{\epsilon}+1}\right)$, the minimax risk is not affected by $R$. Hence, $R=H_2\left(\frac{e^{\epsilon}}{e^{\epsilon}+1}\right)$ is a critical point that defines the minimum amount of randomness required for each user to generate an $\epsilon$-LDP mechanism, while achieving the optimal utility at the analyst.
\end{remark}
\begin{remark}
Corollary~\ref{Cor2_1} also characterizes the number of users $n$ (sample complexity) required to estimate the distribution $\mathbf{p}$ with estimation error at most $\alpha$ for given privacy level $\epsilon$ and randomness $R$ bits per user is (where $k$ is the input alphabet size):
\begin{equation}
n
=%\left\{\begin{array}{ll}
\begin{cases}
\Theta\left(\frac{k}{\alpha \epsilon^2}\right) & \text{if}\ R\geq H_2\left(\frac{e^{\epsilon}}{e^{\epsilon}+1}\right),\\
\Theta\left(\frac{k}{\alpha p_{R}^{2}\epsilon^2}\right) & \text{if}\ R < H_2\left(\frac{e^{\epsilon}}{e^{\epsilon}+1}\right).
\end{cases}
%\end{array}
%\right.
\end{equation} 
A remark analogous to Remark~\ref{remark_critical-rand} also holds here.
\end{remark}

\subsection{Multi-level Privacy, $d>1$}\label{multilevel-privacy}
Here, we study the case of $d$ different analysts, with privacy levels $\epsilon_1>\cdots>\epsilon_d$, and $\epsilon_j=\mathcal{O}\left(1\right)$ for $j\in\left[d\right]$. A trivial scheme is to use the $d=1$ scheme multiple times, separately for each analyst: each user $i\in\left[n\right]$ generates $d$  samples $\left(Y_{i}^{1},\ldots,Y_{i}^{d}\right)$ from its input sample $X_i$. The $j$th  sample $Y_{i}^{j}$ is delivered privately to the $j$th analyst. Note that the $j$th  sample must be generated from an $\epsilon_j$-LDP. It then follows from Corollary~\ref{Cor2_1} that the minimum risk for the $j$th analyst is given by $r^{\ell_{2}^{2}}_{\epsilon_j,\infty,n,k}=\Theta\left(\frac{k}{n\epsilon_j^{2}}\right) $, which requires each user to have $R_j\geq H_2\left(\frac{e^{\epsilon_j}}{e^{\epsilon_j}+1}\right)$ bits of randomness, and results in a total amount of randomness $$R_{\text{total}}^{\text{trivial}}=\sum_{j=1}^{d} H_2\left(\frac{e^{\epsilon_j}}{e^{\epsilon_j}+1}\right).$$ We propose a new scheme, in which each user generates a single output that is publicly accessible by all analysts; each analyst is given a part of the random key that was used to privatize the data, and leverages this key to reduce the perturbation of the public output. The next theorem is proved in Section~\ref{Privlv}.

\begin{theorem}~\label{Th2_3} There exists a private mechanism using a total amount of randomness given by
$R_{\text{total}}^{\text{proposed}}=\sum_{j=1}^{d} H_2\left(q_{j}\right)$,
  such that the $j$th analyst achieves the minimum risk estimation $r^{\ell_{2}^{2}}_{\epsilon_j,\infty,n,k}=\Theta\left(\frac{k}{n\epsilon_j^{2}}\right)$, while preserving privacy of each user with privacy level $\epsilon_j$ for  $j\in[d]$. Here, for every $j\in[d]$, $q_j$ is defined as follows 
  (where $z_j=\frac{1}{e^{\epsilon_j}+1}$):
\begin{equation}~\label{eqn2_2}
q_j=
%\left\{\begin{array}{ll}
\begin{cases}
z_j&\text{if}\ j=1, \\
\frac{z_{j}-z_{j-1}}{1-2z_{j-1}}&\text{if}\ j>1.
\end{cases}
%\end{array}\right.
\end{equation} 
\end{theorem}  
\begin{remark} 
Note that $z_j>z_{j-1}$ as $\epsilon_{j-1}> \epsilon_{j}$. 
Moreover, we also have $z_j=1/\left(e^{\epsilon_j}+1\right)< 0.5$ for all $j\in\left[d\right]$. 
As a result, we can show that for $j>1$, we have
\begin{equation}
q_j=\frac{z_j-z_{j-1}}{1-2z_{j-1}}=z_j-\frac{z_{j-1}\left(1-2z_{j}\right)}{1-2z_{j-1}}<z_j.
\end{equation}
Hence, we get that $H_2\left(q_j\right)<H_2\left(z_j\right)$ holds for all $j>1$. Therefore,
our proposed scheme uses a strictly smaller
amount of randomness than the trivial scheme.
\end{remark}

\subsection{Private Recoverability, $d=2$}
We here consider a legitimate analyst with permission to access the data $\lbrace X_i\rbrace_{i=1}^{n}$, i.e., $\epsilon_1\to\infty$, and an untrusted analyst with privacy level $\epsilon_2<\infty$. The $i$th user uses a random private key $U_i$ and her mechanism $Q_i$ to generate an output $Y_i$ that is publicly accessible by both analysts. 
\begin{ddd}[\textbf{LDP-Rec mechanisms}]\label{LDP-Rec}
We say that a private mechanism $Q$ is $\epsilon$-LDP-Rec, if it is an $\epsilon$-LDP mechanism and it is possible to recover the input $X$ from output $Y$ and the key $U$.
\end{ddd}
 We derive necessary and sufficient conditions on the random keys $\lbrace U_i \rbrace $ and the  mechanisms $\lbrace Q_i  \rbrace$, such that the legitimate analyst can recover  $ X_i$ from observing $U_i$ and $Y_i$, while preserving privacy level $\epsilon_2$ against the untrusted analyst who does not have access to the keys.

\begin{figure}[t]
  \centerline{\includegraphics[scale=0.4,trim={3.5cm 9.8cm 15cm 1.6cm},clip]{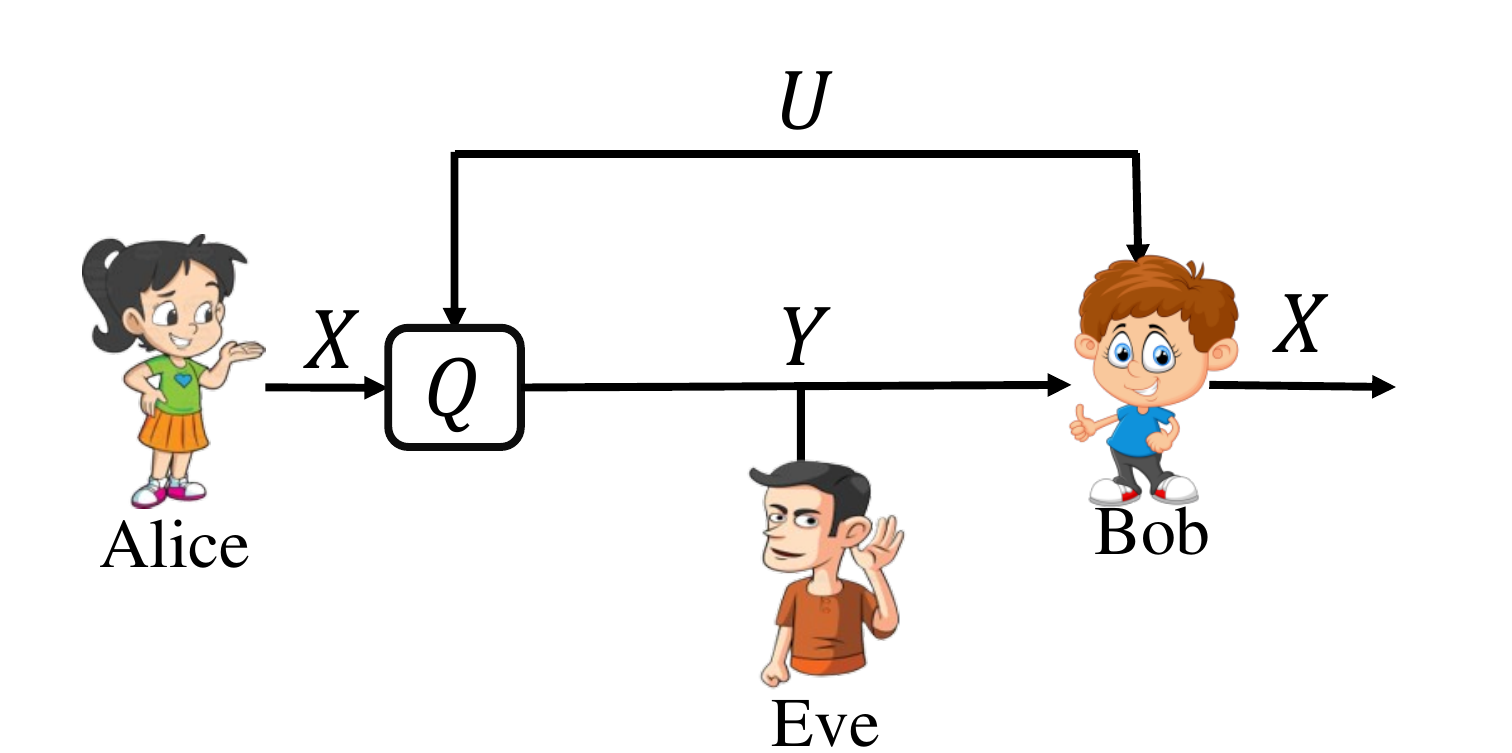}}
 \caption{Private-Recoverability: Alice has data $X$. An $\epsilon$-LDP-Rec mechanism $Q$ is applied to $X$ using a random key $U$ to generate output $Y$. Bob is capable to recover $X$ from $Y$ and $U$. Eve only observes $Y$.}
\label{Fig2_1}
\end{figure}

We first consider a simplified setting as shown in Figure~\ref{Fig2_1}. Alice (an arbitrary user \footnote{Since the input samples $X_1,\ldots,X_n$ are i.i.d., and the random keys $U_1,\ldots,U_n$ are independent random variables, it is sufficient to study the private-recoverable mechanism for any single user.}) has a sample $X\in\mathcal{X}$. Alice wants to send her sample $X$ to Bob (the legitimate analyst) while keeping her sample $X$ private against Eve (the untrusted analyst) with differential privacy level $\epsilon$. Eve has access to the message between Alice and Bob. However, Alice has a random key $U$ shared with Bob that Eve does not have access to. Let $Y$ be the output of the private mechanism $Q$ used by Alice. The following theorem (which we prove in Section~\ref{Recov-A}) provides necessary and sufficient conditions on the random key $U$ and the privatized output $Y$ to generate an $\epsilon$-LDP-Rec mechanism.

\begin{remark}\label{private-recover-general}
Observe that in the simplified model in Figure~\ref{Fig2_1}, we do not impose any assumptions on the input $X$. Furthermore, we do not impose any assumptions about the task for Eve. Hence, our model and results in Theorem~\ref{Th2_4} are applicable to any task for Eve including distribution estimation, heavy hitter estimation, or learning from sample $X$.
\end{remark}

\begin{theorem}~\label{Th2_4}
Let $Q$ be an $\epsilon$-LDP-Rec mechanism that uses a random key $U\in\mathcal{U}$ and an input $X\in\mathcal{X}$ to produce a privatized output $Y\in\mathcal{Y}$. The following conditions are necessary and sufficient to allow recovery of $X$ from $(U,Y)$:\\

(1)  $|\mathcal{U}|\geq |\mathcal{Y}| \geq |\mathcal{X}|$. \\

(2) The entropy of the random key must satisfy $H\left(U\right)\geq H\left(U_{\min}^{s^{*}}\right)$,
where $s^{*}=\arg\min\limits_{s\in\lbrace \ceil{l},\floor{l}\rbrace}H\left(U_{\min}^{s}\right)$ for $l=k\frac{e^{\epsilon}\left(\epsilon-1\right)+1}{\left(e^{\epsilon}-1\right)^2}$ and $U_{\min}^{s}$ is a random variable with support size equal to $|\mathcal{X}|=k$ and has the following distribution:

\begin{align*}
\mathbf{q}_{\min}^{s}=[1/t,\hdots,1/t,e^{\epsilon}/t,\hdots,e^{\epsilon}/t],
\end{align*}
where $t=(se^{\epsilon}+k-s)$, the first $k-s$ terms are equal to $1/t$ and the remaining $s$ terms are equal to $e^{\epsilon}/t$.
%\end{enumerate}
\end{theorem}
We now discuss the effect of $\epsilon$ on the structure of optimal distribution $\mathbf{q}_{\min}^{s^*}$ for $U_{\min}^{s^*}$:
{\sf (i)} When $\epsilon \gg \log(k)$, the optimal $s^*=1$, and the corresponding $\mathbf{q}_{\min}^{1}$ has its first $k-1$ terms equal to $1/(e^{\epsilon}+k-1)$ and the last term equal to $e^{\epsilon}/(e^{\epsilon}+k-1)$. This distribution is equivalent to the one used in the Randomized Response (RR) model proposed in~\cite{warner1965randomized}. 
{\sf (ii)} When $\epsilon\to0$, the optimal $s^*$ is around $k/2$, and the corresponding $\mathbf{q}_{\min}^{k/2}$ has its first $k/2$ terms equal to $2/k(e^{\epsilon}+1)$ and the remaining $k/2$ terms equal to $2e^{\epsilon}/k(e^{\epsilon}+1)$.
{\sf (iii)} When $\epsilon=0$, the distribution $q_{\min}^s$ becomes uniform (irrespective of the value of $s$). Thus, when $\epsilon$ decreases, the distribution $\mathbf{q}_{\min}^{s}$ approaches to the uniform distribution. On the other hand, when $\epsilon$ increases, the distribution $\mathbf{q}_{\min}^{s}$ becomes skewed. It turns out that the minimum randomness required to generate an $\epsilon$-LDP-Rec mechanism for input recoverability is a non-increasing function of $\epsilon$. In other words, more privacy requires more randomness.

\begin{remark} Consider the cryptosystem introduced by Shannon in~\cite{shannon1949}, where Alice wants to send a secure message $X$ to Bob using a shared random key $U$. Let $Y$ be the encrypted message sent to Bob. Eve eavesdrops the channel between Alice and Bob and observes $Y$. This cryptosystem achieves {\em perfect secrecy} if and only if $I\left(X;Y\right)=0$. %, where $I\left(X;Y\right)$ denotes the mutual information between $X$ and $Y$ that measures the amount of information that $Y$ contains about $X$. 
Shannon showed that perfect secrecy requires $H\left(U\right)\geq H\left(X\right)$. Since the distribution of $X$ is not known to any node (Alice, Bob, and Eve), this implies $H\left(U\right)\geq\max_{p_X\in\Delta_k}H(X)=\log k$. We can easily verify that the $\epsilon$-LDP-Rec mechanism satisfies a cryptosystem with secrecy measure $\max_{\mathbf{p}\in\Delta_k}I\left(X;Y\right)\leq \epsilon$. Hence, a perfect secrecy system with unknown input distribution is a $0$-LDP-Rec mechanism, which is a special case of our problem. Moreover, the $\epsilon$-LDP-Rec mechanism with data recovery is a cryptosystem leaking an amount of information measured by $\max_{\mathbf{p}\in\Delta_k}I\left(X;Y\right)\leq \epsilon$. 
\end{remark}

Observe that Theorem~\ref{Th2_4} does not provide performance guarantees for Eve, it only guarantees privacy for Alice with respect to Eve, and recoverability for Bob. Hence, we can ask the question: Does there exist an $\epsilon$-LDP-Rec mechanism using the smallest amount of randomness and guaranteeing the smallest error for distribution estimation or heavy hitter estimation for Eve (the untrusted analyst)? In the following theorem (which we prove in Section~\ref{Recov-B}), we show that such a mechanism exists.

\begin{theorem}~\label{Th2_5} The Hadamard Response mechanism from~\cite{acharya2018hadamard} satisfies private-recoverability, and is utility-wise order-optimal for distribution estimation and heavy hitter estimation while using an order-optimal amount of randomness. 
\end{theorem}

\subsection{Sequence of Distribution (or Heavy Hitter) Estimation}
We again start from the setting in Figure~\ref{Fig2_1}, but with the modification that Alice (an arbitrary user) wants to send to Bob (a legitimate analyst) $T$ independent samples $X^{T}=\left(X^{\left(1\right)},\ldots,X^{\left(T\right)}\right)$, where $X^{\left(t\right)}\in\mathcal{X}$, while keeping them private against Eve (an untrusted analyst) with differential privacy level $\epsilon$. 
Eve  has access to the sequence of outputs  $Y^{T}=\left(Y^{\left(1\right)},\ldots,Y^{\left(T\right)}\right)$ that Alice produces,  but not to the random key $U$ that Alice and Bob share.  Note that each output $Y^{\left(t\right)}$ might be a function of all input samples $X_1^{t}=\left(X^{\left(1\right)},\ldots,X^{\left(t\right)}\right)$ and the key $U$. Furthermore, the output $Y^{\left(t\right)}$ can take values from a set $\mathcal{Y}^{\left(t\right)}$ that is not required to be the same as $\mathcal{Y}^{\left(t^{\prime}\right)}$ for $t\neq t^{\prime}$. Let $\mathcal{Y}^{T}=\mathcal{Y}^{\left(1\right)}\times\cdots\times \mathcal{Y}^{\left(T\right)}$. The following theorem is proved in Section~\ref{Recov_GDP}.

We can define $\epsilon$-DP-Rec mechanisms in the same way as we defined $\epsilon$-LDP-Rec mechanisms in Definition~\ref{LDP-Rec}: A mechanism $Q$ is $\epsilon$-DP-Rec, if it satisfies \eqref{eqn1_2}, and allows the recovery of input $X$ from the output $Y$ and the key $U$.
\begin{theorem}~\label{Th2_6}
Let $Q$ be an $\epsilon$-DP-Rec mechanism that uses a random key $U\in\mathcal{U}$ and an input database $X^{T}\in\mathcal{X}^{T}$ to create an output $Y^{T}\in\mathcal{Y}^{T}$. The following conditions  are necessary and sufficient to  allow recovery of the input $X^{T}$ from $(U,Y^{T})$.\\
(1) $|\mathcal{U}|\geq |\mathcal{Y}^{T}| \geq |\mathcal{X}^{T}|$. \\
(2) The entropy of the random key must satisfy $H\left(U\right)\geq  T \min\limits_{s^{*}\in\lbrace \ceil{l},\floor{l}\rbrace} H\left(U_{\min}^{s^{*}}\right)$, where $U_{\min}^{s}$ is the same random variable with support size $|\mathcal{X}|=k$, as defined in Theorem~\ref{Th2_4}.
\end{theorem}

\begin{figure}[t!]
 \centering
 \begin{subfigure}[t]{0.49\linewidth}
 \centerline{\includegraphics[scale=0.32]{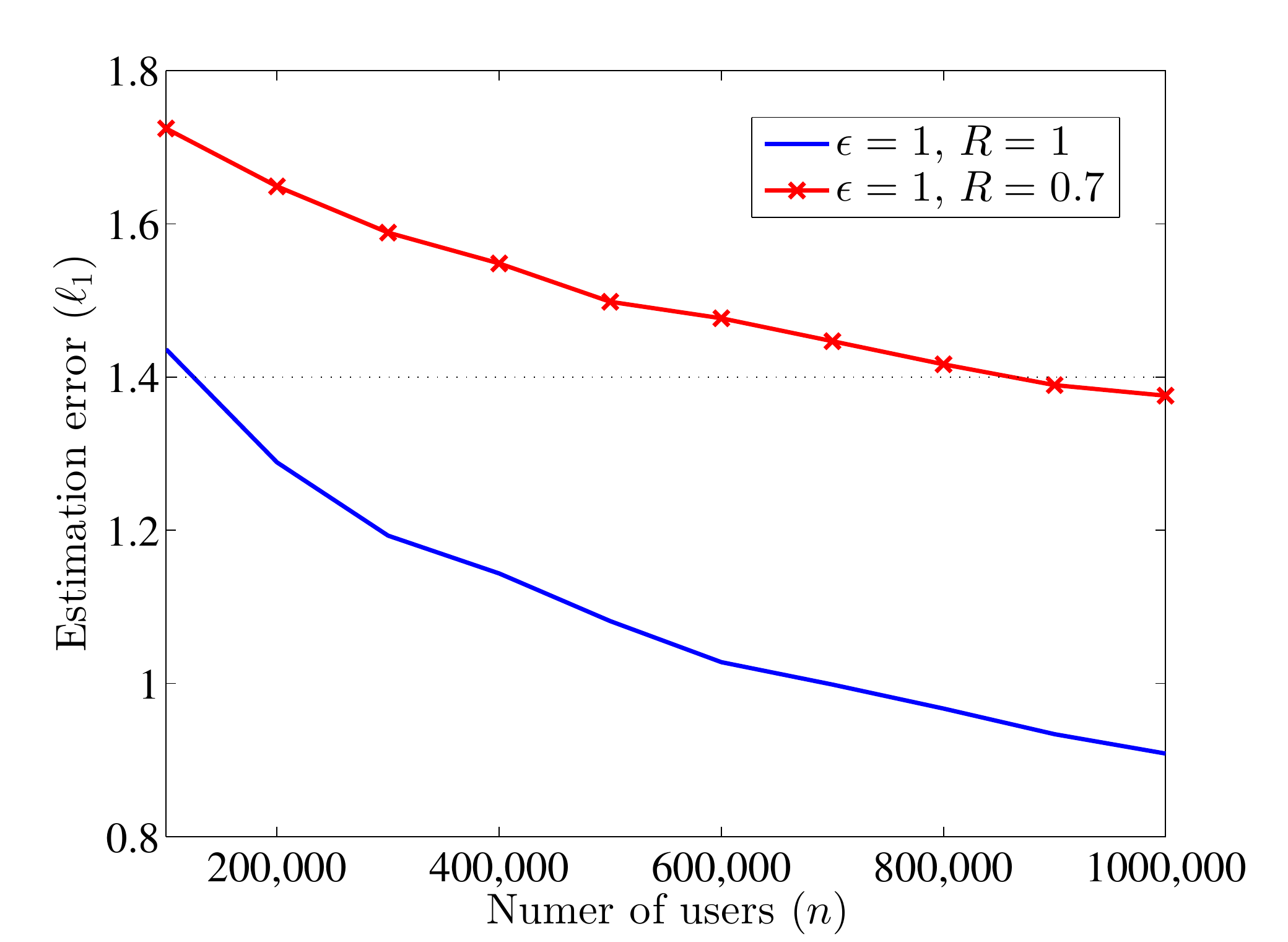}}
    \caption{$\ell_1$-estimation error for input alphabet size $k=1000$, privacy level $\epsilon=1$, and $\mathbf{p}=\text{Geo}\left(0.8\right)$.}
    ~\label{Fig4_1}
 \end{subfigure}
 \hfill
\begin{subfigure}[t]{0.49\linewidth}
    \centerline{\includegraphics[scale=0.407,width=7cm,height=5cm]{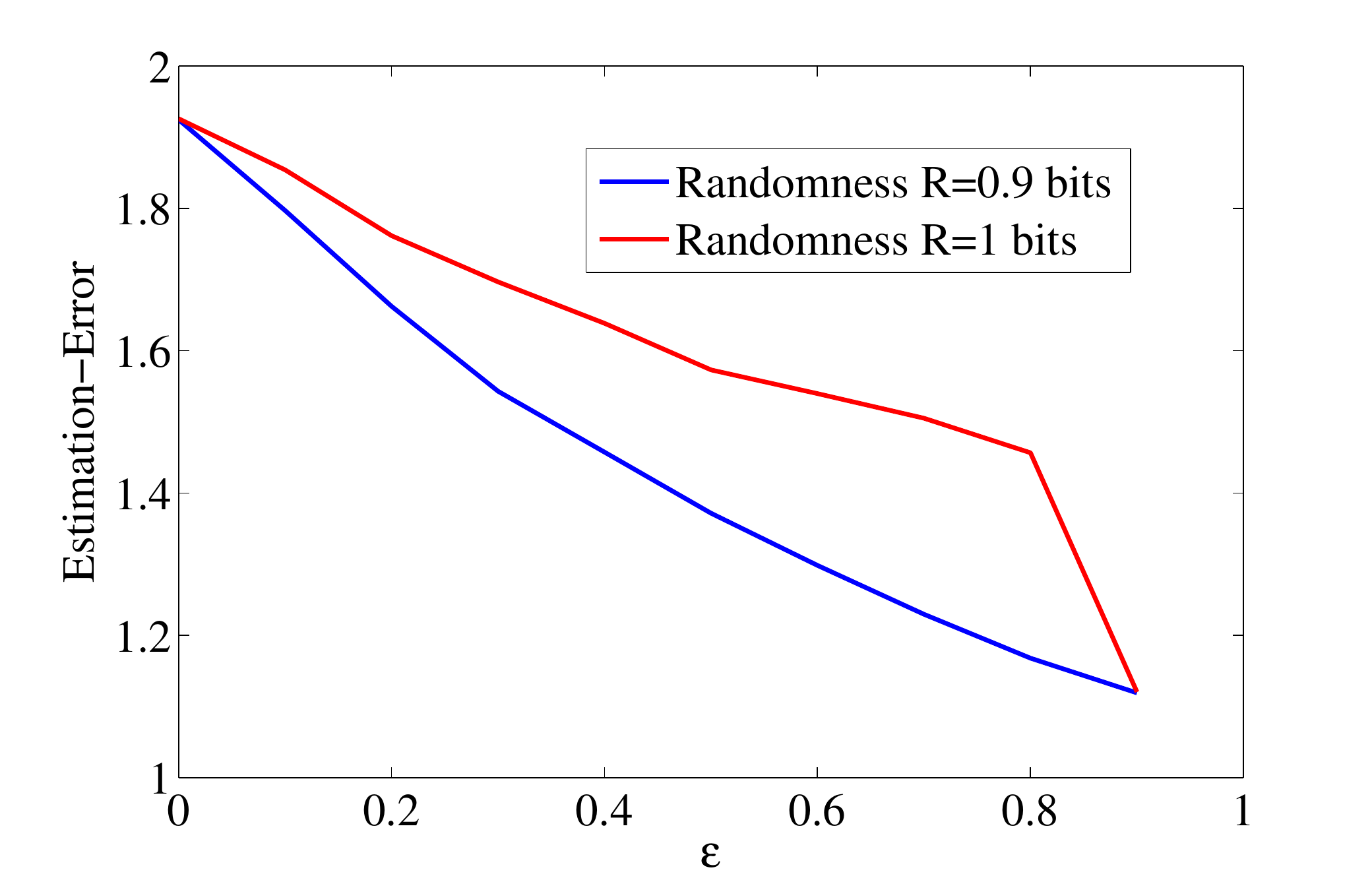}}
 \caption{Estimation error for input alphabet size $k=1000$, number of users $n=500000$, and $\mathbf{p}=\text{Geo}\left(0.8\right)$.}
 ~\label{Fig_2}
 \end{subfigure}
 \caption{Single-level privacy}
\end{figure}

Theorem~\ref{Th2_6} shows that the minimum amount of randomness required to preserve privacy of $T$ samples is equal to $T$ times the amount of randomness required to preserve privacy of a single sample. That is, for $\epsilon$-DP-Rec, it is optimal to use an $\epsilon$-LDP-Rec mechanism $T$ times. 

\begin{remark}
Observe that Theorem~\ref{Th2_6} is applicable in a $n$-user setting (by setting $T=n$), where user $i$ has a single sample $X^{(i)}$, and all users have access to a shared random key $U$. So we have that shared randomness among users does not help in reducing the overall required amount of randomness. 
\end{remark}

\section{Numerical Evaluation} \label{Num}

In this section, we numerically validate our theoretical results through simulation.

\begin{figure}[t!]
 \centering
 \begin{subfigure}[t]{0.49\linewidth}
    \centerline{\includegraphics[scale=0.32]{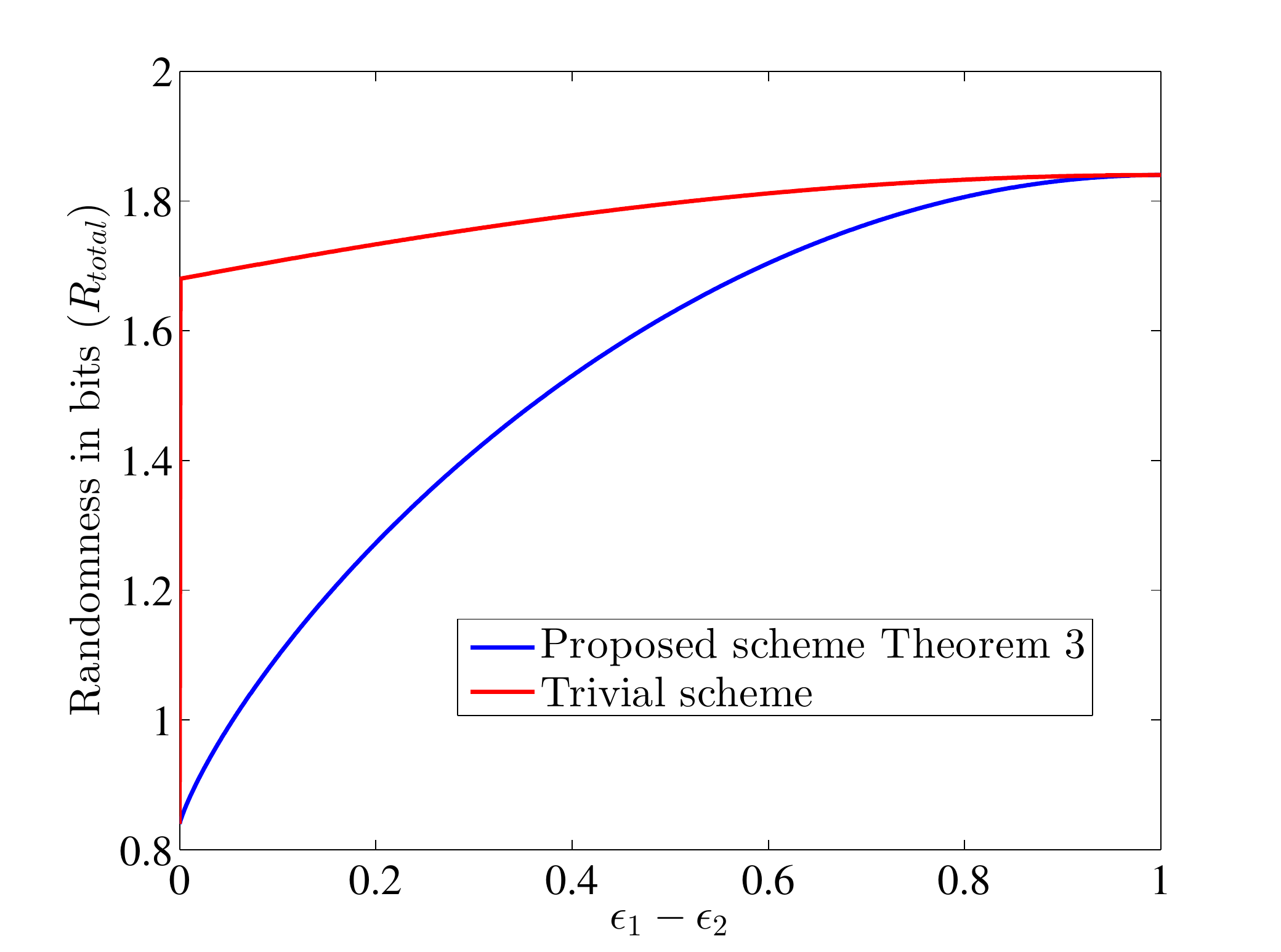}}
    \caption{Comparison between our privacy scheme proposed in Theorem~\ref{Th2_3} and the trivial scheme for two privacy levels $\epsilon_1=1$ and $\epsilon_2=[0.01:1]$.}
    ~\label{Fig4_2a}
 \end{subfigure}
 \hfill
\begin{subfigure}[t]{0.49\linewidth}
    \centerline{\includegraphics[scale=0.32]{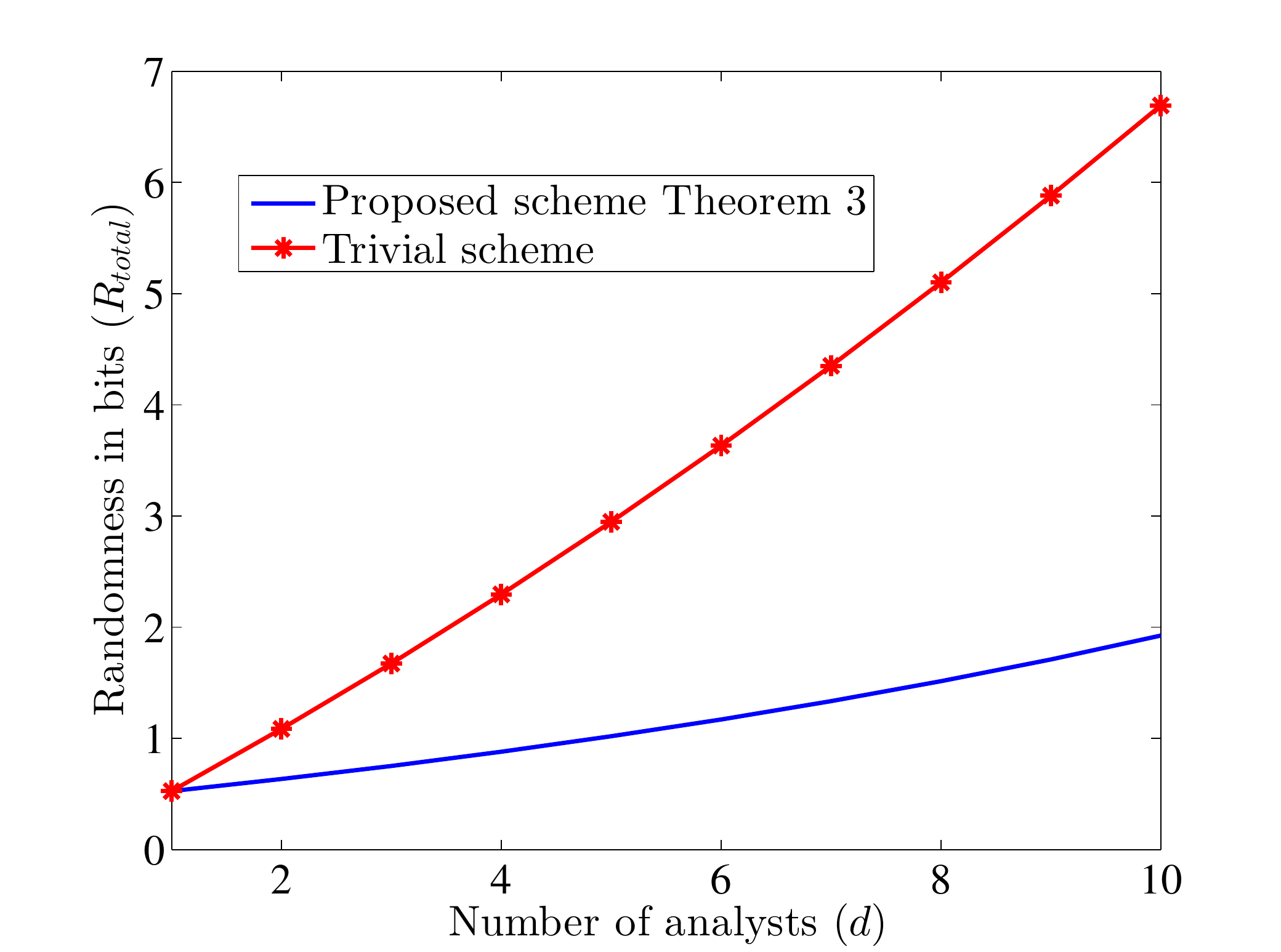}}
    \caption{ Comparison between our privacy scheme proposed in Theorem~\ref{Th2_3} and the trivial scheme for $d$ privacy levels $\epsilon_1=2$ and $\epsilon_j=\epsilon-0.1j$ for $j\in\left[2:d\right]$.}
    ~\label{Fig4_2b}
 \end{subfigure}
 \caption{Multi-level privacy}\vspace{-6mm}
\end{figure}

\textbf{Single-level privacy:} In this part, we investigate the performance of the estimator presented in Theorem~\ref{Th2_2} for a single-level privacy. Each point is obtained by averaging over $20$ runs. In Figure~\ref{Fig4_1}, we plot the estimation error for the $\ell=\ell_1$ loss function ($\|\mathbf{p}-\hat{p}\left(Y^{n}\right)\|_1$) for estimating a discrete distribution $\mathbf{p}\in\Delta_{k}$. The input size is $k=1000$, the number of users is $n\in\left[10^{5}:10^{6}\right]$, and the privacy level is $\epsilon=1$ for two values of randomness $R\in\lbrace 0.7,1\rbrace$ bits per user. The input samples are drawn from a Geometric distribution with parameter $q=0.8$ ($\text{Geo}\left(0.8\right)$), in which $p_i= C q^{i-1}\left(1-q\right)$ for $i\in\left[k\right]$, where $C$ is a normalization term. Figure~\ref{Fig4_1} shows that the number of users required to achieve a certain estimation error increases as the amount of randomness per user decreases. For instance, to achieve an $\ell_1$-error equal to $1.4$, we need $n\approx 150,000$ users if $R=1$ bits per user, while we need $n\approx 850,000$ users if $R=0.7$ bits per user. 

Figure~\ref{Fig_2} depicts the $\ell_1$ estimation error as a function of the privacy level $\epsilon$ for input size $k=1000$ and number of users $n=500000$ for two different values of randomness $R\in\lbrace 1,0.6\rbrace$ bits per user. As we discussed in Theorem~$1$, for each privacy level $\epsilon$, there is a critical point of randomness $R=H\left(e^{\epsilon}/\left(e^{\epsilon}+1\right)\right)$. When each user has $R<H\left(e^{\epsilon}/\left(e^{\epsilon}+1\right)\right)$ bits of randomness, then the $\ell_1$ estimation loss increases as the randomness $R$ decreases. While when each user has $R\geq H\left(e^{\epsilon}/\left(e^{\epsilon}+1\right)\right)$ bits of randomness, the estimation error is not affected by the amount of randomness $R$. In Figure~\ref{Fig_2}, we find that the $\ell_1$ error depends on the randomness $R$ for all $\epsilon<0.8$, since we have $R=0.9<H\left(e^{\epsilon}/\left(e^{\epsilon}+1\right)\right)$ for all $\epsilon<0.8$.

\textbf{Multi-level privacy:} Figure~\ref{Fig4_2a} and Figure~\ref{Fig4_2b} compare our proposed scheme in Theorem~\ref{Th2_3} with the trivial scheme with respect to the total amount of randomness used. In the trivial scheme, each user generates $d$ different privatized samples, one for each analyst. In Figure~\ref{Fig4_2a} we consider two privacy levels $\epsilon_1=1$ and $\epsilon_2\leq\epsilon_1$. We find that when $\epsilon_1-\epsilon_2$ is small, then the trivial scheme requires approximately twice the total amount of randomness used in our  scheme. However, when $\epsilon_1- \epsilon_2$ is large, then our scheme and the trivial scheme use similar amounts of randomness. In Figure~\ref{Fig4_2b}, we consider $d\in\left[1:10\right]$,   $\epsilon_1=2$ and $\epsilon_j=\epsilon_1 - 0.1 j$, for $j\in\lbrace 2,\ldots,d\rbrace$. We find that the gap between  the amount of randomness used in our  scheme and the trivial scheme  increases with $d$.

\textbf{Private-recoverability:} Observe that each user needs $\log\left(k\right)$ bits to store her input sample $X\in\left[ k\right]$, since she does not know the distribution $X\sim\mathbf{p}$. In private-recoverability, we can recover  $X$ from observing  $Y$ and $U$; hence, we only need to store $U$. Figure~\ref{Fig4_3} plots the number of bits required to store $U$ (see Theorem~\ref{Th2_4}) as a function of the privacy level $\epsilon$ and different values of input size $k\in\lbrace 10,100,1000\rbrace$. The black lines represent the $\log\left(k\right)$ bits required to store $X$ (an additional secure copy). Note that the amount of bits needed to store $U$ is strictly smaller than $\log\left(k\right)$ for $\epsilon>0$, and decreases as the privacy level $\epsilon$ increases. Observe that the gain in Fig~\ref{Fig4_3} is per user. Hence, the total amount of saving in storage would be considerable when the number of users is large and $\epsilon>0$. For example, when $\epsilon=5$, alphabet size $k=2,4,10$, we get gain in efficiency $\frac{\log\left(k\right)-H\left(U\right)}{\log\left(k\right)}$ of $94.2\%$, $91.4\%$, and $85\%$ respectively.

\begin{figure}[t!]
   \centerline{\includegraphics[scale=0.35]{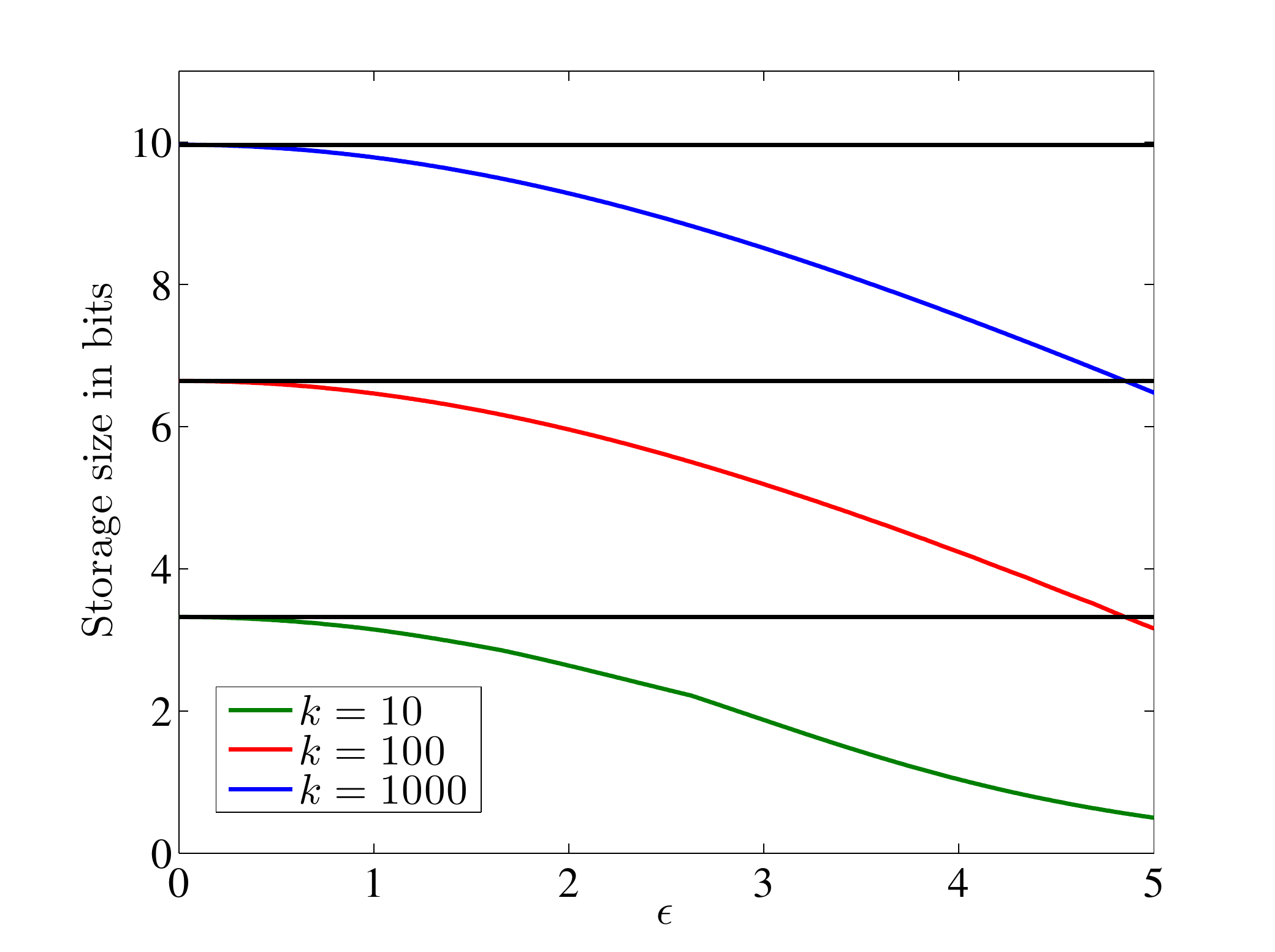}}
    \caption{Comparison between  storage required for $X$ and a random key $U$, for input alphabet sizes $k\in\lbrace 10,100,1000\rbrace$. The black lines represent $\log\left(k\right)$.}
    ~\label{Fig4_3}\vspace{-6mm}
 \end{figure}

\section{Single-level Privacy (Proofs of Theorem~\ref{Th2_1} and Theorem~\ref{Th2_2})} \label{LDP} 

\subsection{Lower Bound on The Minimax Risk Estimation Using Assouad's Method}~\label{LDP_AS}
Now we prove the lower bound on the minimax risk given in Theorem~\ref{Th2_1} (see page~\pageref{Th2_1}). We first follow similar steps as in~\cite{duchi2018minimax,Ye2018} to reduce the minimax problem into multiple binary testing problems using Assouad's method. We note that~\cite{duchi2018minimax,Ye2018} do not consider a randomness constraint. Hence, we formulate an optimization problem to obtain a lower bound on the minimax risk estimation with a randomness constraint. Finding a tight bound on the solution of this problem is the main step in our proof. We also provide an alternative proof of Theorem~\ref{Th2_1} by using Fisher information, which leads to a tight bound for $\ell=\ell_{2}^{2}$ with smaller constant factors (see Appendix~\ref{LDP_F}).

Let $|\mathcal{X}|=k$ be the input alphabet size. Let $\lbrace \mathbf{p}^{\nu}\rbrace$ be a set of distributions parameterized by $\nu=\left(\nu_1,\ldots,\nu_{k/2}\right)\in\mathcal{V}=\lbrace -1,1\rbrace^{k/2}$. The distribution $\mathbf{p}^{\nu}=\left(p_{1}^{\nu},\ldots,p_{k}^{\nu}\right)$ is given by:
\begin{equation}
p_{j}^{\nu}=\left\{\begin{array}{ll}
\frac{1}{k}+\delta \nu_j& \text{if}\ j\in\lbrace 1,\ldots,k/2\rbrace\\
\frac{1}{k}-\delta \nu_{j-k/2} & \text{if}\ j\in\lbrace k/2+1,\ldots,k\rbrace\\
\end{array} \right.,
\end{equation}
where $0\leq\delta\leq 1/k$ is a parameter that will be chosen later. Let $Y^{n}=\left[Y_1,\ldots,Y_n\right]$ and $\mathcal{Y}^{n}=\mathcal{Y}_1\times\cdots\times\mathcal{Y}_n$. Following~\cite{duchi2018minimax}, for any loss function $\ell\left(\hat{\mathbf{p}},\mathbf{p}\right)=\sum_{j=1}^{k}\phi\left(\hat{p}_j-p_j\right)$, where $\phi:\mathbb{R}\to\mathbb{R}_{+}$ is a symmetric function, we have\footnote{Observe that for loss function $\ell=\ell_{2}^{2}$, we have $\phi\left(x\right)=x^{2}$, and for loss function $\ell=\ell_{1}$, we have $\phi\left(x\right)=|x|$.}
\begin{equation}
\ell\left(\hat{\mathbf{p}}\left(y^{n}\right),\mathbf{p}^{\nu}\right)=\sum_{j=1}^{k}\phi\left(\hat{p}_{j}\left(y^{n}\right)-p_j^{\nu}\right)\geq\phi\left(\delta\right)\sum_{j=1}^{k/2}\mathbbm{1}\left(\text{sgn}\left(\hat{p}_j\left(y^{n}\right)-\frac{1}{k}\right)\neq \nu_j\right),
\end{equation}
where $\text{sgn}\left(x\right)=1$ if $x\geq 0$ and $\text{sgn}\left(x\right)=0$ otherwise. Suppose that user $i$ chooses a private mechanism $Q_i\in\mathcal{Q}_{\left(\epsilon,R\right)}$ that generates an output $Y_i\in\mathcal{Y}_i$. Let $\mathbf{M}_{i}^{\nu}$ be the output distribution on $\mathcal{Y}_i$ for an input distribution $\mathbf{p}^{\nu}$ on $\mathcal{X}$ defined by
\begin{equation}~\label{eqn3_30}
\mathbf{M}_{i}^{\nu}\left(y\right)=\sum_{j=1}^{k}Q_{i}\left(y|X_i=j\right)p_{j}^{\nu}.
\end{equation}
Let $\mathbf{M}^{n}_{+j}$ and $\mathbf{M}^{n}_{-j}$ denote the marginal distribution on $\mathcal{Y}^{n}$ conditioned on $\nu_j=+1$ and $\nu_j=-1$, respectively, where
\begin{align*}
\mathbf{M}^{n}_{+j}\left(y^{n}\right) &= \frac{1}{|\mathcal{V}|}\sum_{\nu:\nu_j=+1}\prod_{i=1}^{n}\mathbf{M}^{\nu}_i\left(y_i\right) \\
\mathbf{M}^{n}_{-j}\left(y^{n}\right) &= \frac{1}{|\mathcal{V}|}\sum_{\nu:\nu_j=-1}\prod_{i=1}^{n}\mathbf{M}^{\nu}_i\left(y_i\right).
\end{align*}
Thus, the minimax risk can be bounded using the following lemma whose proof is presented in Appendix~\ref{AppK}. 
\begin{lemma}~\label{lemm3_4} For the family of distributions $\left\{ \mathbf{p}^{\nu}:\nu\in\mathcal{V}=\lbrace -1,1\rbrace^{k/2}\right\}$, and a loss function $\ell\left(\hat{\mathbf{p}},\mathbf{p}\right)=\sum_{j=1}^{k}\phi\left(\hat{p}_j-p_j\right)$ defined above, we have
\begin{equation}~\label{eqn3_10}
r^{\ell}_{\epsilon,R,n,k}\geq\phi\left(\delta\right)\frac{k}{2}\left(1-\sqrt{\frac{n}{2}\sup_{j\in\left[k/2\right]}\sup_{i\in\left[n\right]}\sup_{\nu:\nu_j=1}\sup_{Q_i\in\mathcal{Q}_{\left(\epsilon,R\right)}} D_{\text{KL}}\left(\mathbf{M}^{\nu}_{i}||\mathbf{M}^{\nu-2e_j}_{i}\right)}\right)
\end{equation}
\end{lemma}

Fix arbitrary $i\in\left[n\right]$, $j\in\left[k/2\right]$ and $\nu\in\mathcal{V}$. We have
{\allowdisplaybreaks

\begin{align}
D_{\text{KL}}&\left(\mathbf{M}^{\nu}_{i}||\mathbf{M}^{\nu-2e_j}_{i}\right)\stackrel{\left(a\right)}{\leq} D_{\text{KL}}\left(\mathbf{M}^{\nu}_{i}||\mathbf{M}^{\nu-2e_j}_{i}\right)+D_{\text{KL}}\left(\mathbf{M}^{\nu-2e_j}_{i}||\mathbf{M}^{\nu}_{i}\right)\\
&=\sum_{y\in\mathcal{Y}_i} \left( \mathbf{M}^{\nu}_{i}\left(y\right)-\mathbf{M}^{\nu-2e_j}_{i}\left(y\right)\right)\log\left(\frac{\mathbf{M}^{\nu}_{i}\left(y\right)}{\mathbf{M}^{\nu-2e_j}_{i}\left(y\right)}\right)\nonumber\\
&\stackrel{\left(b\right)}{\leq}\sum_{y\in\mathcal{Y}_i} \frac{\left( \mathbf{M}^{\nu}_{i}\left(y\right)-\mathbf{M}^{\nu-2e_j}_{i}\left(y\right)\right)^2}{\mathbf{M}^{\nu-2e_j}_{i}\left(y\right)}\stackrel{\left(c\right)}{=}\sum_{y\in\mathcal{Y}_i} \delta^2\frac{\left( Q_i\left(y|j\right)-Q_i\left(y|j+k/2\right)\right)^2}{\sum_{j^{\prime}=1}^{k}Q_i\left(y|j^{\prime}\right)p^{\nu-2e_j}_{j^{\prime}}}\nonumber\\
&\stackrel{\left(d\right)}{\leq}2\delta^2 e^{\epsilon}\sum_{y\in\mathcal{Y}_i} \frac{\left( Q_i\left(y|j\right)-Q_i\left(y|j+k/2\right)\right)^2}{Q_i\left(y|j\right)+Q_i\left(y|j+k/2\right)}~\label{eqn3_31},
\end{align} 
}
where step $\left(a\right)$ follows from the fact that $D_{\text{KL}}\left(.||.\right)$ is not negative. Step $\left(b\right)$ follows from the inequality $\log\left(x\right)\leq x-1$. Step $\left(c\right)$ follows from the definition of $\mathbf{M}^{\nu}_{i}$ in~\eqref{eqn3_30}. Step $\left(d\right)$ follows from bounding the denominator as follows: 
\begin{equation}
\begin{aligned}
\sum\limits_{j^{\prime}=1}^{k}Q_i\left( y|j^{\prime}\right)p_{j^{\prime}}^{\nu-2e_j}&\geq e^{-\epsilon}\frac{Q_i\left( y|j\right)+Q_i\left( y|j+k/2\right)}{2}\sum_{j^{\prime}=1}^{k}p_{j^{\prime}}^{\nu-2e_j}\\
&=e^{-\epsilon}\frac{Q_i\left( y|j\right)+Q_i\left( y|j+k/2\right)}{2},
\end{aligned}
\end{equation} 
where we use the fact that $Q_i\left( y|j^{\prime}\right)\geq e^{-\epsilon} Q_i\left( y|j\right)$ and $Q_i\left( y|j^{\prime}\right)\geq e^{-\epsilon} Q_i\left( y|j+k/2\right),\ \forall j^{\prime}\in\left[k\right]$.

\begin{lemma}~\label{lemma_1}  
For any randomized mechanism $Q\in\mathcal{Q}_{\left(\epsilon,R\right)}$ that generates an output $Y\in\mathcal{Y}$, we have
\begin{equation}~\label{eqna_1}
\begin{aligned}
\sup_{Q\in\mathcal{Q}_{\left(\epsilon,R\right)}}\sum_{y\in\mathcal{Y}}\frac{\left(Q\left(y|j\right)-Q\left(y|j+k/2\right)\right)^2}{Q\left(y|j\right)+Q\left(y|j+k/2\right)}\leq\left\{\begin{array}{ll}
2\frac{\left(e^\epsilon-1\right)^2}{\left( e^\epsilon+1\right)^2}& \text{if}\ R\geq H_2\left(\frac{e^\epsilon}{e^\epsilon+1}\right)\\
2\frac{p_R^2\left(e^{\epsilon}-1\right)^2}{ e^{2\epsilon}}& \text{if}\ R< H_2\left(\frac{e^\epsilon}{e^\epsilon+1}\right)
\end{array}
\right.
\end{aligned}
\end{equation}
\end{lemma}
 This lemma presents an upper  bound on equation~\eqref{eqn3_31} as a function of the randomness $R$ for any private mechanism $Q\in\mathcal{Q}_{\left(\epsilon,R\right)}$. To prove this lemma, we first show that the optimization problem~\eqref{eqna_1} is non-convex due to the randomness constraint. We then prove that the maximum value of this function~\eqref{eqna_1} is obtained when the output of the mechanism $Q\in\mathcal{Q}_{\left(\epsilon,R\right)}$ is binary. Then, we obtain a tight bound numerically for the binary output. 
\begin{proof}[Proof of Lemma \ref{lemma_1}]
Without loss of generality assume that $\mathcal{Y}=\lbrace y_1,\ldots,y_m\rbrace$ with $|\mathcal{Y}|=m$. For ease of notation, we write $Q\left(y_l|j\right)=q_{l,j}$ and $Q\left(y_l|j+k/2\right)=q_{l,j+k/2}$. The problem~\eqref{eqna_1} can be formulated as follows
\begin{align}
\textbf{P1:} \qquad&\max_{\lbrace q_{l,j},q_{l,j+k/2}\rbrace_{l=1}^{m}} \sum_{l=1}^{m}\frac{\left(q_{l,j}-q_{l,j+k/2}\right)^2}{q_{l,j}+q_{l,j+k/2}}~\label{eqnb_1}\\
\text{s.t.}&\quad H\left(\left[q_{1,j},\ldots,q_{m,j}\right]\right)\leq R,\qquad H\left(\left[q_{1,j+k/2},\ldots,q_{m,j+k/2}\right]\right)\leq R~\label{eqnb_4}\\
&\quad e^{-\epsilon}\leq \frac{q_{l,j}}{q_{l,j+k/2}}\leq e^{\epsilon},\hspace{3cm} \forall l\in\left[m\right]\nonumber\\
&\quad q_{l,j}\geq 0,\qquad q_{l,j+k/2}\geq 0,\hspace{2cm} \forall l\in\left[m\right]\nonumber\\
&\quad \sum_{l=1}^{m}q_{l,j}=1,\qquad \sum_{l=1}^{m}q_{l,j+k/2}=1\nonumber ~\label{eqnb_7}
\end{align}
Note that the objective function~\eqref{eqnb_1} is jointly convex in both $\lbrace q_{l,j} \rbrace_{l=1}^{m}$ and $\lbrace q_{l,j+k/2}\rbrace_{l=1}^{m}$. However, the optimization problem \textbf{P1} is non-convex due to two reasons. First, we maximize a convex function, and second the entropy constraints~\eqref{eqnb_4} are sub-level sets of a concave function and are non-convex constraints. However, we can solve the optimization problem~\textbf{P1} by exploiting the results of Lemma~\ref{lemmb_1} below. 
\begin{lemma}~\label{lemmb_1} The optimal solution of the non-convex optimization problem \textbf{P1} is obtained when the output size is $m=2$.
\end{lemma}
The proof of Lemma~\ref{lemmb_1} is presented in Appendix~\ref{AppB}. 
%In this lemma, we prove that the optimal solution of~\textbf{P1} is obtained at $m=2$. Hence, 
Since the output alphabet is binary, we can efficiently plot the feasible region of~\textbf{P1} for $m=2$ as depicted in Figure~\ref{Figb_1}. Since we maximize a convex function, the optimal solution is at the boundary of the feasible set. Furthermore, the objective function~\eqref{eqnb_1} is symmetric on $q_{1,j},\ q_{1,j+k/2}$ for $m=2$. As a result, the optimal solution is given by. 

\begin{figure*}[t!]
\begin{subfigure}{0.5\linewidth}
    \centerline{\includegraphics[scale=0.355]{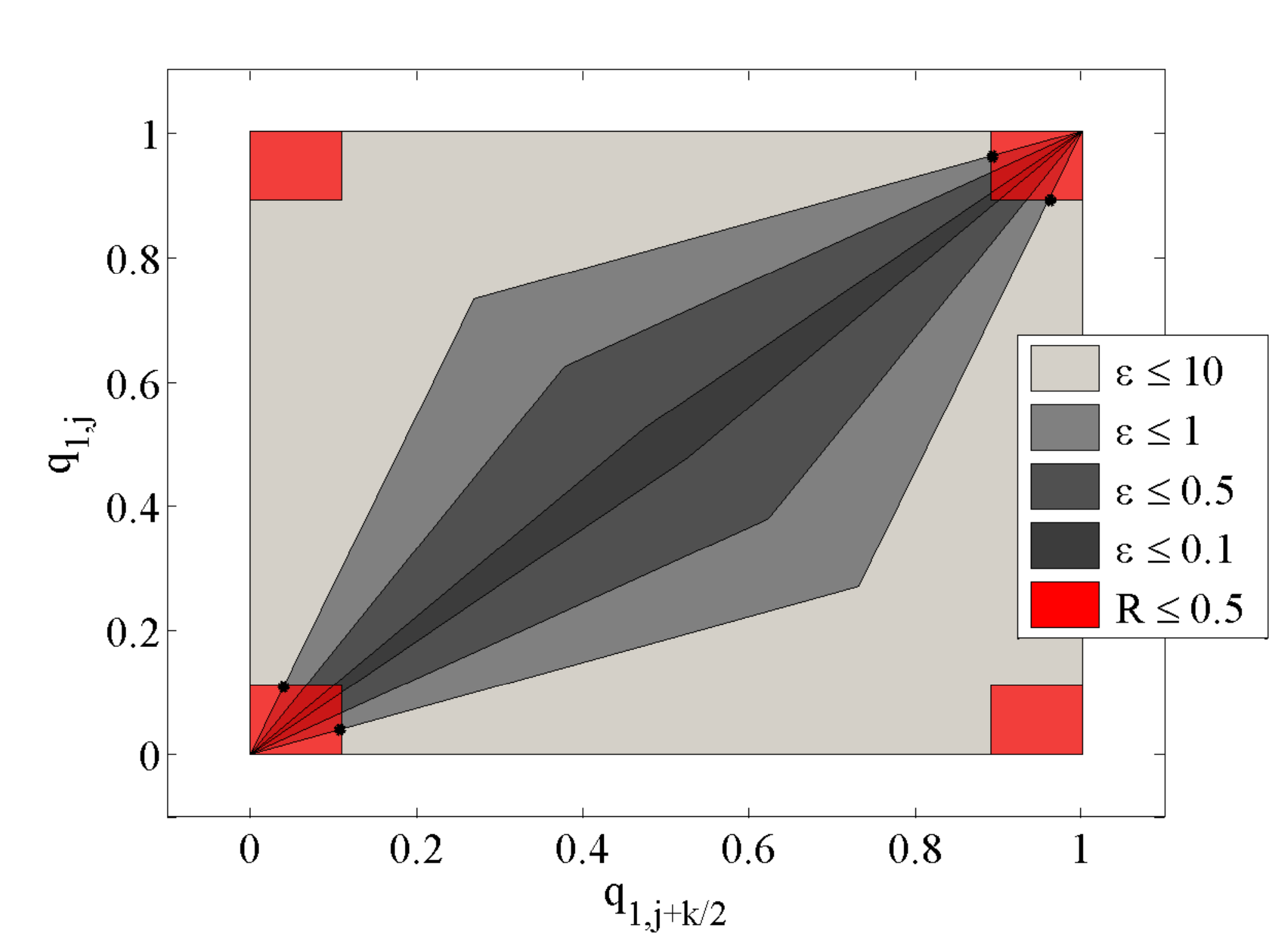}}
 \end{subfigure}
 \begin{subfigure}{0.49\linewidth}
    \centerline{\includegraphics[scale=0.37]{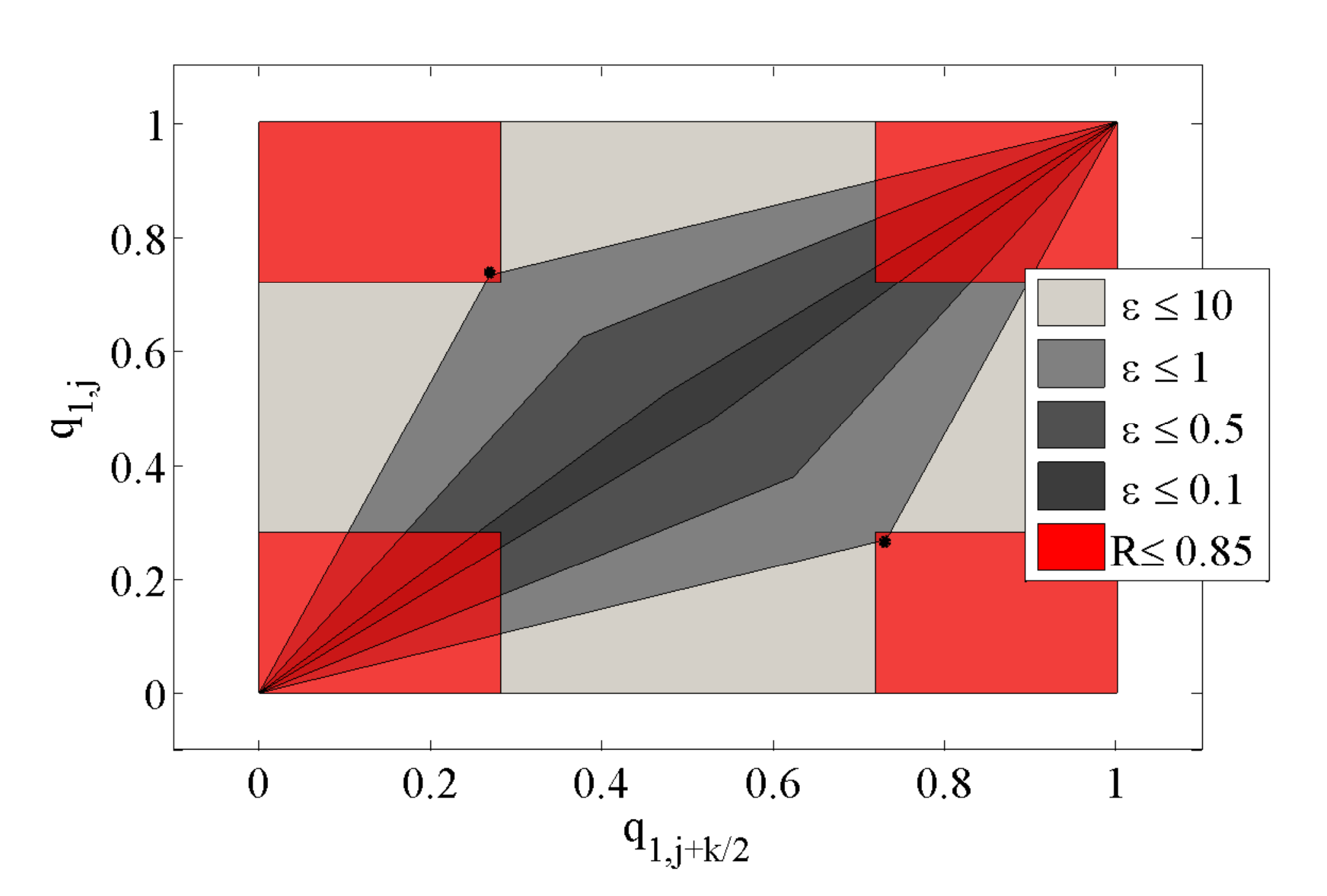}}
 \end{subfigure}
 \caption{The feasible region of the optimization problem~\textbf{$P1$} for $m=2$. In $\left(a\right)$, we have $R=0.5<H_2\left(\frac{e^{\epsilon}}{e^{\epsilon}+1}\right)$ for $\epsilon=1$, and hence the optimal point is one of the black points.  In $\left(b\right)$, we have $R=0.85>H_2\left(\frac{e^{\epsilon}}{e^{\epsilon}+1}\right)$ for $\epsilon=1$, and hence the optimal point is one of the black vertices.}
 ~\label{Figb_1}
\end{figure*}

\begin{equation}~\label{eqnb_2}
q_{1,j}^{*}=\left\{\begin{array}{ll}
\frac{e^\epsilon}{e^\epsilon+1}& \text{if}\ R\geq H_2\left(\frac{e^\epsilon}{e^\epsilon+1}\right)\\
p_R& \text{if}\ R< H_2\left(\frac{e^\epsilon}{e^\epsilon+1}\right)
\end{array}
\right.,\qquad 
q_{1,j+k/2}^{*}=\left\{\begin{array}{ll}
\frac{1}{e^\epsilon+1}& \text{if}\ R\geq H_2\left(\frac{e^\epsilon}{e^\epsilon+1}\right)\\
\frac{p_R}{e^{\epsilon}}& \text{if}\ R< H_2\left(\frac{e^\epsilon}{e^\epsilon+1}\right)
\end{array}
\right.,
\end{equation}
where $q_{2,j}^{*}=1-q^{*}_{1,j}$, and $q_{2,j+k/2}^{*}=1-q_{1,j+k/2}^{*}$. Substituting from~\eqref{eqnb_2} into the objective function~\eqref{eqnb_1}, we get
\begin{equation}
 \sum_{l=1}^{m}\frac{\left(q_{l,j}-q_{l,j+k/2}\right)^2}{q_{l,j}+q_{l,j+k/2}}\leq \left\{\begin{array}{ll}
2\frac{\left(e^\epsilon-1\right)^2}{\left( e^\epsilon+1\right)^2}& \text{if}\ R\geq H_2\left(\frac{e^\epsilon}{e^\epsilon+1}\right)\\
2\frac{p_R^2\left(e^{\epsilon}-1\right)^2}{ e^{2\epsilon}}& \text{if}\ R< H_2\left(\frac{e^\epsilon}{e^\epsilon+1}\right)
\end{array}
\right.
\end{equation} 
Hence, the proof is completed for Lemma~\ref{lemma_1}.
\end{proof}
Using the bound from Lemma~\ref{lemma_1} in \eqref{eqn3_31} and taking supremum over all $Q_i\in\mathcal{Q}_{\epsilon,R}$, we get
\begin{equation}~\label{eqn3_8}
\begin{aligned}
\sup_{Q_i\in\mathcal{Q}_{\left(\epsilon,R\right)}} D_{\text{KL}}\left(\mathbf{M}^{\nu}_{i}||\mathbf{M}^{\nu-2e_j}_{i}\right)&\leq 2\delta^2 e^{\epsilon} \sup_{Q_i\in\mathcal{Q}_{\left(\epsilon,R\right)}} \sum_{y\in\mathcal{Y}_i} \frac{\left( Q_i\left(y|j\right)-Q_i\left(y|j+k/2\right)\right)^2}{Q_i\left(y|j\right)+Q_i\left(y|j+k/2\right)}\\
&=2\delta^2 e^{\epsilon} \left\{\begin{array}{ll}
2\frac{\left(e^\epsilon-1\right)^2}{\left( e^\epsilon+1\right)^2}& \text{if}\ R\geq H_2\left(\frac{e^\epsilon}{e^\epsilon+1}\right)\\
2\frac{p_R^2\left(e^{\epsilon}-1\right)^2}{ e^{2\epsilon}}& \text{if}\ R< H_2\left(\frac{e^\epsilon}{e^\epsilon+1}\right)
\end{array}
\right.
\end{aligned}
\end{equation}
Substituting from~\eqref{eqn3_8} into~\eqref{eqn3_10}, we get
\begin{equation}
\begin{aligned}
r^{\ell}_{\epsilon,R,n,k}&\geq  \left\{\begin{array}{ll}
\phi\left(\delta\right)\frac{k}{2}\left(1-\sqrt{2\delta^2 ne^{\epsilon} \frac{\left(e^{\epsilon}-1\right)^2}{\left(e^{\epsilon}+1\right)^2}}\right)& \text{if}\ R\geq H_2\left(\frac{e^\epsilon}{e^\epsilon+1}\right)\\
\phi\left(\delta\right)\frac{k}{2}\left(1-\sqrt{2\delta^2 n \frac{p_R^2\left(e^{\epsilon}-1\right)^2}{e^{\epsilon}}}\right)& \text{if}\ R< H_2\left(\frac{e^\epsilon}{e^\epsilon+1}\right)
\end{array}
\right.
\end{aligned}
\end{equation}
By setting $\delta^2=\frac{\left(e^{\epsilon}+1\right)^2}{8ne^{\epsilon}\left(e^{\epsilon}-1\right)^2}$ if $R\geq H_2\left(\frac{e^\epsilon}{e^\epsilon+1}\right)$ and $\delta^2=\frac{e^{\epsilon}}{8np_R^2\left(e^{\epsilon}-1\right)^2}$ if $R\geq H_2\left(\frac{e^\epsilon}{e^\epsilon+1}\right)$, we get
\begin{equation}
\begin{aligned}
r^{\ell}_{\epsilon,R,n,k}&\geq  \left\{\begin{array}{ll}
\phi\left(\sqrt{\frac{\left(e^{\epsilon}+1\right)^2}{8ne^{\epsilon}\left(e^{\epsilon}-1\right)^2}}\right)\frac{k}{4}& \text{if}\ R\geq H_2\left(\frac{e^\epsilon}{e^\epsilon+1}\right)\\
\phi\left(\sqrt{\frac{e^{\epsilon}}{8np_R^2\left(e^{\epsilon}-1\right)^2}}\right)\frac{k}{4}& \text{if}\ R< H_2\left(\frac{e^\epsilon}{e^\epsilon+1}\right)
\end{array}
\right.
\end{aligned}
\end{equation}
For the loss function $\ell=\ell_{2}^{2}$, we set $\phi\left(x\right)=x^{2}$ and for $\ell=\ell_{1}$, we set $\phi\left(x\right)=|x|$. This completes the proof of Theorem~\ref{Th2_1} with a slightly worse constant of 32 instead of 16 in the denominator. We provide a different proof of Theorem~\ref{Th2_1} in Appendix~\ref{LDP_F} using Fisher information that gives the exact bound as stated in  Theorem~\ref{Th2_1}.

\subsection{Upper Bound on The Minimax Risk Estimation Using Hadamard Response}~\label{LDP-AC}
In this section, we prove Theorem~\ref{Th2_2} (see page~\pageref{Th2_2}) by proposing a private mechanism by adapting the Hadamard response given in~\cite{acharya2019communication}, where each user answers to a yes-no question such that the probability of telling the truth depends on the amount of randomness~$R$. 
Each user $i\in\left[n\right]$ has a binary output $Y_i\in\lbrace 0,1\rbrace$. The $\left(\epsilon,R\right)$-LDP mechanism of the $i$-th user is defined by
\begin{equation}~\label{eqn3_9}
Q\left(Y_i=1|X\right)=\left\{\begin{array}{ll}
q& \text{if}\ X\in B_i\\
\frac{q}{e^{\epsilon}}& \text{if}\ X\notin B_i\\
\end{array}\right.
\end{equation}
where $B_i\subset\left[k\right]$ is a subset of inputs, and $q$ is a probability value that will be determined later such that $H_2\left(q\right)\leq R$.  Let $K=2^{\ceil{\log\left(k\right)}}$ denote the smallest power of $2$ larger than $k$, and $H_{K}$ be the $K\times K$ Hadamard matrix. In the following, we assume an extended distribution $\overline{\mathbf{p}}$ over the set $\mathcal{X}=\left[K\right]$ with $|\mathcal{X}|=K$ that is obtained by zero-padding the original distribution $\mathbf{p}$ with $\left( K-k\right)$ zeros, i.e., $\overline{\mathbf{p}}=\left[\overline{p}_1,\ldots,\overline{p}_{K}\right]=\left[p_1,\ldots,p_k,0,\ldots,0\right]$. 
For $j\in\left[K\right]$, let $B^{j}$ be a set of row indices that have $1$ in the $j$-th column of the Hadamard matrix $H_{K}$. For example, when $K=4$, the Hadamard matrix is given by
\begin{equation}
H_{4}=\begin{bmatrix}
1&1&1&1\\
1&-1&1&-1\\
1&1&-1&-1\\
1&-1&-1&1\\
\end{bmatrix}
\end{equation}
Hence, $B^{1}=\lbrace 1,2,3,4\rbrace$, $B^{2}=\lbrace 1,3\rbrace$, $B^{3}=\lbrace 1,2\rbrace$, and $B^{4}=\lbrace 1,4\rbrace$. We divide the users into $K$ sets ($\mathcal{US}_1,\ldots,\mathcal{US}_K$), where each set contains $n/K$ users. For each user $i\in\mathcal{US}_j$, we set  $B_i=B^{j}$. Let $p\left(B^{j}\right)=\text{Pr}\left[X\in B^{j}\right]=\sum_{x\in B^{j}}\overline{p}_x$, and $s_j=\text{Pr}\left[Y_i=1\right]$ for $i\in\mathcal{U}_j$. Then, we can easily see that 
\begin{equation}
\begin{aligned}
s_j&=p\left(B^{j}\right)q+\left(1-p\left(B^{j}\right)\right)\frac{q}{e^{\epsilon}}\\
&=p\left(B^{j}\right)q\left(\frac{e^{\epsilon}-1}{e^{\epsilon}}\right)+\frac{q}{e^{\epsilon}}
\end{aligned}
\end{equation}
Let $\hat{s}_j=\frac{1}{|\mathcal{US}_j|}\sum_{i\in\mathcal{US}_j}\mathbbm{1}\left\{Y_i=1\right\}$ denote the estimate of $s_j$. Then, we can estimate $p\left(B^{j}\right)$ as $\hat{p}\left(B^{j}\right)=\frac{e^{\epsilon}}{q\left(e^{\epsilon}-1\right)}\left(\hat{s}_j-\frac{q}{e^{\epsilon}}\right)$. Observe that the relation between the distribution $\overline{\mathbf{p}}$ and $\mathbf{p}\left(B\right)=\left[p\left(B^{1}\right),\ldots,p\left(B^{K}\right)\right]$ is given by~\cite[Eq.~$13$]{acharya2019communication}
\begin{equation}
\mathbf{p}\left(B\right)=\frac{H_{K}\overline{\mathbf{p}}+\mathbf{1}_{K}}{2},
\end{equation}
where $\mathbf{1}_K$ denotes a vector of $K$ ones. Hence, we can estimate the distribution $\overline{\mathbf{p}}$ as
\begin{equation}
\hat{\overline{\mathbf{p}}}=H_{K}^{-1}\left(2\hat{\mathbf{p}}\left(B\right)-\mathbf{1}_{K}\right)=\frac{1}{K}H_{K}\left(2\hat{\mathbf{p}}\left(B\right)-\mathbf{1}_{K}\right).
\end{equation}
\begin{lemma}~\label{lemm3_3} For arbitrary $\mathbf{p}\in\Delta_k$, we have
\begin{equation}
\mathbb{E}\left[\|\mathbf{p}-\hat{\mathbf{p}}\|_{2}^{2}\right]\leq \frac{2k e^{2\epsilon}}{n q^{2}\left(e^{\epsilon}-1\right)^2}.
\end{equation}
\end{lemma}
The proof is exactly the same as the proof in~\cite[Theorem~$5$]{acharya2019communication}. By setting $q=\frac{e^{\epsilon}}{e^{\epsilon}+1}$ if $R\geq H_2\left(\frac{e^{\epsilon}}{e^{\epsilon}+1}\right)$ and $q=p_R$ if $R< H_2\left(\frac{e^{\epsilon}}{e^{\epsilon}+1}\right)$, we get
\begin{equation}
r_{\epsilon,R,n,k}^{\ell_{2}^{2}}\leq\left\{ \begin{array}{ll} \frac{2k \left(e^{\epsilon}+1\right)^{2}}{n \left(e^{\epsilon}-1\right)^2} & \text{if}\ R\geq H_2\left(\frac{e^{\epsilon}}{e^{\epsilon}+1}\right),\\
\frac{2k e^{2\epsilon}}{np_R^{2} \left(e^{\epsilon}-1\right)^2} & \text{if}\ R< H_2\left(\frac{e^{\epsilon}}{e^{\epsilon}+1}\right).
\end{array}
\right.
\end{equation}
The difference in our mechanism is that we design the private mechanism~\eqref{eqn3_9} for all values of randomness $R$. This completes the proof of Theorem~\ref{Th2_2}.

\section{Multi-level Privacy (Proof of Theorem~\ref{Th2_3})} 

\label{Privlv}

This section proves Theorem~\ref{Th2_3} (see page~\pageref{Th2_3}) by establishing a new technique using a smaller amount of randomness than the trivial scheme mentioned in Section~\ref{multilevel-privacy} while achieving the minimum risk estimation for each analyst. Our proposed mechanism for multi-level privacy (where $\epsilon_1>\hdots>\epsilon_d$) is a cascading mechanism, where in each step, we add a random key to the output of the previous step (see Figure~\ref{Fig5_1}, for example). The common output of the mechanism is accessible by all analysts. However, each analyst would have a different privacy level depending on the amount of randomness shared with it. Thus, each analyst uses the shared random key to partially undo the randomization of the common output to get less privacy and higher utility. Let $z_j=\frac{1}{e^{\epsilon_j}+1}$ for $j\in\left[d\right]$. For $i\in\left[n\right]$, let $\lbrace U_{i}^{1},\ldots, U_i^{d}\rbrace$ be a set of $d$ Bernoulli random variables, where $U_{i}^{j}$ has a parameter $q_j=\text{Pr}\left[U_i^{j}=1\right]$ given by
\begin{equation}~\label{eqn4_4}
q_j=\left\{\begin{array}{ll}
z_j&\text{if}\ j=1,\\
\frac{z_{j}-z_{j-1}}{1-2z_{j-1}}&\text{if}\ j>1.\\
\end{array}\right.
\end{equation} 
We first use the Hadamard response proposed in~\cite{acharya2019communication} for getting the first step of our mechanism (see Section~\ref{LDP-AC} for more details). Let $H_{K}$ be the $K\times K$ Hadamard matrix. Let $B^{l}$  be a set of the row indices that have $1$ in the $l$-th column of Hadamard matrix $H_{K}$ for $l\in\left[K\right]$. We divide the users into $K$ sets ($\mathcal{US}_1,\ldots,\mathcal{US}_K$), where each set contains $n/K$ users. We assign a set $B_i=B^{l}$ representing a subset of inputs for each user $i\in\mathcal{US}_l$. Then, user $i$ generates a virtual output $Y_{i}^{1}\in\lbrace 0,1\rbrace$ as follows
\begin{equation}~\label{eqn4_2}
Y_{i}^{1}=\left\{\begin{array}{ll}
1& \text{if}\ \left(X_i\in B_i\ \text{and}\ U_{i}^{1}=0\right)\ \text{or}\ \left(X\notin B_i\ \text{and}\ U_{i}^{1}=1\right),\\
0& \text{otherwise}.\\
\end{array}\right.
\end{equation}  
Observe that the representation of $Y_{i}^{1}$ in~\eqref{eqn4_2} is exactly the same as in~\eqref{eqn3_9} by setting $q=\text{Pr}\left[U_{i}^{1}=0\right]=\frac{e^{\epsilon}}{e^{\epsilon}+1}$. We represent $Y_{i}^{1}$ with this form to explicitly show the random keys used to design the Hadamard scheme presented in Section~\ref{LDP-AC}. Let $Y_{i}^{j}$ be the virtual output generated by user $i$ for the $j$th analyst, which is given by
\begin{equation}
Y_{i}^{j}=Y_{i}^{1}\oplus U_{i}^{2}\oplus \ldots\oplus U_{i}^{j},
\end{equation} 
where $\oplus$ denotes the bitwise XOR. Hence, we add randomization to the first step of the Hadamard scheme. User $i$ transmits the output $Y_{i}^{d}$ to all analysts. The private scheme is shown in Figure~\ref{Fig5_1}.
\begin{figure}[t]
  \centerline{\includegraphics[scale=0.38,trim={0cm 10cm 13cm 0cm},clip]{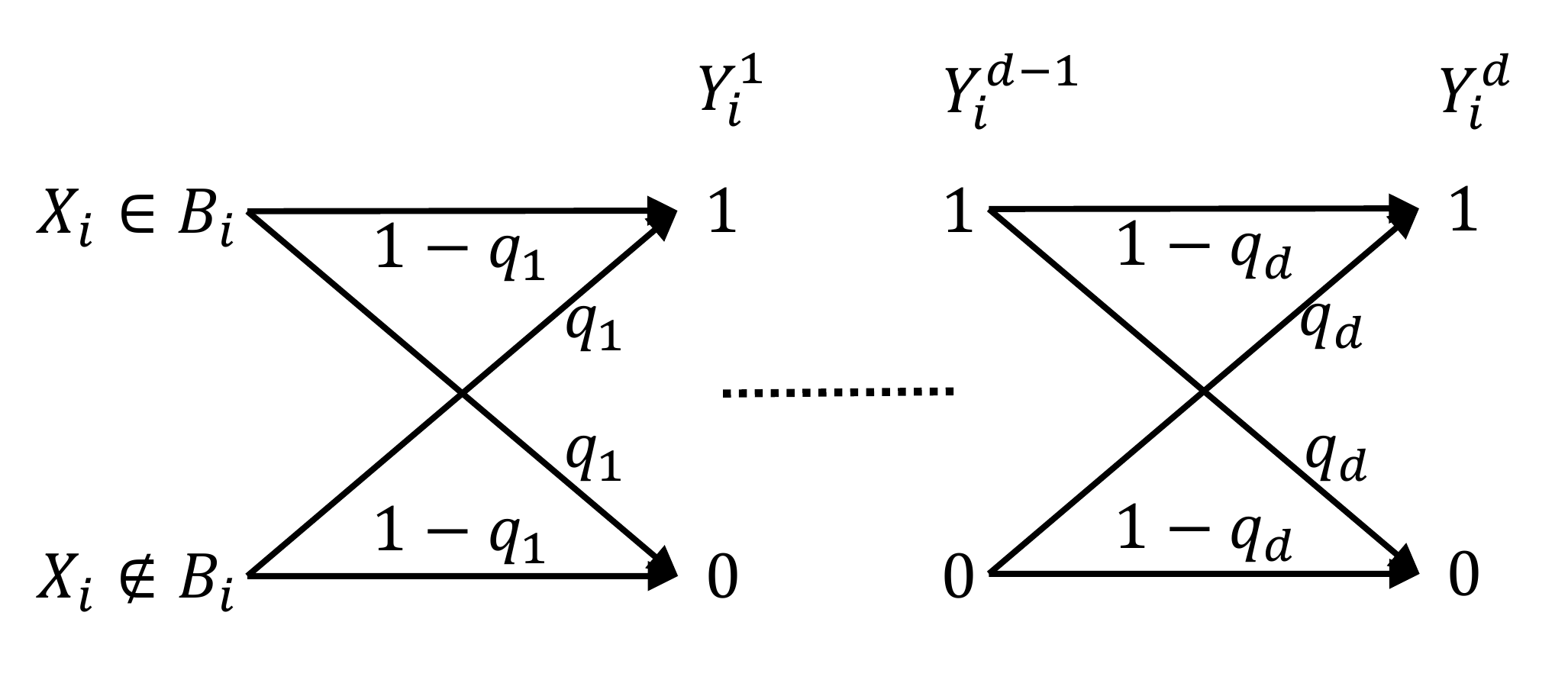}}
 \caption{Multiple privacy levels mechanism.}\vspace{-7mm}
~\label{Fig5_1}
\end{figure}

\begin{lemma}~\label{lemm4_1} 
The $j$th output of user $i$ satisfies $\epsilon_j$-LDP, i.e.,
\begin{equation}~\label{eqn4_1}
\sup_{y_i^{j}\in\lbrace 0,1\rbrace}\sup_{x_i,x_i^{\prime}\in\mathcal{X}}\frac{\text{Pr}\left[Y_{i}^{j}=y_{i}^{j}|X_i=x_i\right]}{\text{Pr}\left[Y_{i}^{j}=y_{i}^{j}|X_i=x_i^{\prime}\right]}\leq e^{\epsilon_j}
\end{equation}
\end{lemma}
We prove Lemma~\ref{lemm4_1} in Appendix~\ref{AppD}. Note that each analyst has access to the public outputs $\lbrace Y_{1}^{d},\ldots,Y_{n}^{d}\rbrace$ which is $\epsilon_d$-LDP. Additionally, user $i$ sends a random key $L_{i}^{j}=U_{i}^{d}\oplus\ldots\oplus U_{i}^{j+1}$ to the $j$th analyst. Using the random keys $\lbrace L_{1}^{j},\ldots,L_{n}^{j}\rbrace$, the $j$th analyst can construct the private outputs $\lbrace Y_{1}^{j},\ldots,Y_{n}^{j}\rbrace$ which are $\epsilon_j$-LDP, where $Y_{i}^{j}=Y_{i}^{d}\oplus L_{i}^{j}$. Observe that the privatized output $Y_{i}^{j}$ has a conditional distribution given by
\begin{equation}
Q_{i}\left(Y_{i}^{j}|X_i\right)=\left\{\begin{array}{ll}
\frac{e^{\epsilon_j}}{e^{\epsilon_j}+1}& \text{if}\ X_i\in B_i\\
\frac{1}{e^{\epsilon_j}+1}& \text{if}\ X_i\not\in B_i\\
\end{array}
\right.
\end{equation}
which coincides with the private mechanism given in~\eqref{eqn3_9} with $q=\frac{e^{\epsilon_j}}{e^{\epsilon_j}+1}$. From Lemma~\ref{lemm3_3}, for privacy level $\epsilon_j=\mathcal{O}\left(1\right)$, we get that
\begin{equation}
r_{\epsilon,R,n,k}^{\ell_{2}^{2},j}=\mathcal{O}\left(\frac{k}{n\epsilon_{j}^{2}}\right),
\end{equation}
for analyst $j$, which coincides with the lower bound stated in Corollary~\ref{Cor2_1}. Observe that the total amount of randomness per user in the proposed mechanism is given by
\begin{equation}
\begin{aligned}
R_{\text{total}}^{\text{proposed}}&=\sum_{j=1}^{d} H\left(U^{j}\right)=\sum_{j=1}^{d} H_2\left(q_{j}\right)\leq R_{\text{total}}^{\text{trivial}},
\end{aligned}
\end{equation}
where $q_j$ is defined in~\eqref{eqn4_4}. 
Note that the last inequality is strict for $d>1$, which follows from the argument presented in Section~\ref{multilevel-privacy}.
This completes the proof of Theorem~\ref{Th2_3}.
  
\section{Private Recoverability (Proofs of Theorem~\ref{Th2_4} and Theorem~\ref{Th2_5})} 
\label{Recov} 

In this section, we prove Theorem~\ref{Th2_4} (see page~\pageref{Th2_4}) and Theorem~\ref{Th2_5} (see page~\pageref{Th2_5}).

\subsection{Proof of Theorem~\ref{Th2_4}}~\label{Recov-A}
This section proves the necessary and sufficient conditions on the random key $U$ and the privatized output $Y$ to design an $\epsilon$-LDP-Rec mechanism. We first prove that $|\mathcal{Y}|\geq |\mathcal{X}|$ is necessary to recover $X$ from $Y$ and $U$. We then prove that each input $x\in\mathcal{X}$ should be mapped with non-zero probability to every output $y\in\mathcal{Y}$; hence, we get $|\mathcal{U}|\geq |\mathcal{Y}|$, since each input $x\in\mathcal{X}$ can be mapped with non-zero probability to at most $|\mathcal{U}|$ outputs. The main part of our proof is bounding the randomness of the key $U$ in the second condition. We first prove in Lemma~\ref{lemm5_2} that for any $\epsilon$-LDP-Rec mechanism designed using a random key of size greater than the input size, there exists another $\epsilon$-LDP-Rec mechanism designed using a random key of size equal to the input size with the same or smaller amount of randomness. Thus, we can assume that $|\mathcal{U}|=|\mathcal{X}|$ and minimize the entropy of the random key $U$ over all possible distributions and under the $\epsilon$-LDP constraint. Since entropy is a concave function of the distribution, we get a non-convex problem. However, we can obtain an exact solution for the problem due to the structure of the privacy constraints that form a closed polytope. For the sufficiency part, we prove in Lemma~\ref{lemm5_1} that we can construct an $\epsilon$-LDP-Rec mechanism using the random key $U_{\min}^{s^{*}}$ defined in Theorem~\ref{Th2_4} that satisfies the two necessary conditions.

Before we proceed into the proof of Theorem~\ref{Th2_4}, we first present the following two lemmas whose proofs are given in Appendix~\ref{AppE} and Appendix~\ref{AppF}, respectively.
\begin{lemma}~\label{lemm5_1} For given a random key $U\in\mathcal{U}$ with size $|\mathcal{U}|=k$ having a distribution $\mathbf{q}=\left[q_1,\ldots,q_k\right]$ such that $\frac{q_{\max}}{q_{\min}}\leq e^{\epsilon}$, where $q_{\max}=\max\limits_{j\in\left[k\right]}q_j$ and $q_{\min}=\min\limits_{j\in\left[k\right]}q_j$, there exists an $\epsilon$-LDP-Rec mechanism with input $X\in\left[k\right]$ and an output $Y\in\left[k\right]$ designed using $U$.
\end{lemma}
This lemma shows that we can design an $\epsilon$-LDP mechanism with output size equal to the input size if we have a random key with size equal the input size and having a distribution such that $\frac{q_{\max}}{q_{\min}}\leq e^{\epsilon}$.
\begin{lemma}~\label{lemm5_2} Suppose that an $\epsilon$-LDP-Rec mechanism with an input $X\in\left[k\right]$ and an output $Y\in\mathcal{Y}$ is designed using a random key $U\in\mathcal{U}$ with size $|\mathcal{U}|=m> k$. Then there exists an $\epsilon$-LDP-Rec mechanism with an input $X\in\left[k\right]$ and an output $Y\in\left[k\right]$ designed using a random key $U^{\prime}\in\left[k\right]$ such that
$H\left(U\right)\geq H\left(U^{\prime}\right)$.
\end{lemma}
Now, we are ready to prove Theorem~\ref{Th2_4}.
We prove the first necessary condition of Theorem~\ref{Th2_4} in two parts:
We can show $|\mathcal{Y}|\geq |\mathcal{X}|$ using the recoverability constraint and
$|\mathcal{U}|\geq |\mathcal{Y}|$ using the privacy constraint. We prove these in Appendix~\ref{AppF-2}.

From Lemma~\ref{lemm5_2} and the first necessary condition, we see that the $\epsilon$-LDP-Rec mechanism with the smallest amount of randomness is obtained when $|\mathcal{U}|=|\mathcal{Y}|=|\mathcal{X}|=k$. Hence, we restrict our attention to this case only. Let $U\in\left[k\right]$ be a random key having a distribution $\mathbf{q}=\left[q_1,\ldots,q_k\right]$. Without loss of generality, we assume that $q_1\leq q_2\leq \ldots\leq q_k$. Before we prove the necessity of the second condition, we claim that $q_k/q_1\leq e^{\epsilon}$. We prove this using both privacy and recoverability constraints in Appendix~\ref{AppF-2}.

Now, we are ready to prove the necessity of the second condition. Our objective is to find the minimum entropy of the random key $U$ with size $|\mathcal{U}|=k$ such that the private mechanism is $\epsilon$-LDP and the sample $X$ can be recovered from observing $Y$ and the random key $U$. The problem can be formulated as follows
\begin{align}~\label{eqn5_1}
\min\limits_{\mathbf{q}=\left[q_1,\ldots,q_k\right]}&\ H\left(U\right)=-\sum_{j=1}^{k}q_j\log\left(q_j\right)\\
s.t.,&\  1\leq\frac{q_j}{q_1}\leq e^{\epsilon}\ \forall j\in\left[k\right]~\label{eqn5_2}\\
& \sum_{j=1}^{k}q_j=1,
\,\,\, q_j\geq 0\ \forall j\in\left[k\right]~\label{eqn5_3}
\end{align}
where the constraint~\eqref{eqn5_2} is obtained from the claim proved above. Observe that the constraints~\eqref{eqn5_2}-\eqref{eqn5_3} form a closed polytope. Furthermore, the objective function~\eqref{eqn5_1} is a concave function on $\mathbf{q}$. Since we \textit{minimize} a concave function over a polytope, the global optimum point is one of the vertices of the polytope~\cite{rosen1983global}. Since we have a single equality constraint, a vertex has to satisfy at least $k-1$ inequality constraints with equality. Observe that none of the inequalities in~\eqref{eqn5_3} can be satisfied with equality, otherwise the privacy constraints in~\eqref{eqn5_2} would be violated. Thus, the optimal vertex is of the form $$\mathbf{q}=\left[\underbrace{q_1,\ldots,q_1}_{k-s\ \text{terms}},\underbrace{e^{\epsilon}q_1,\ldots,e^{\epsilon}q_1}_{s\ \text{terms}}\right]$$
such that $s$ of inequalities from $\frac{q_j}{q_1}\leq e^{\epsilon}$ are satisfied with equality and $\left( k-s-1\right)$ of inequalities from $1\leq \frac{q_j}{q_1}$ are satisfied with equality, where $s$ is a variable to be optimized. 
Hence, the optimal distribution has the form
\begin{equation}~\label{eqn5_8}
\mathbf{q}^{s}=\left[\underbrace{\frac{1}{se^{\epsilon}+k-s},\ldots,\frac{1}{se^{\epsilon}+k-s}}_{k-s\ \text{terms}},\underbrace{\frac{e^{\epsilon}}{se^{\epsilon}+k-s},\ldots,\frac{e^{\epsilon}}{se^{\epsilon}+k-s}}_{s\ \text{terms}}\right],
\end{equation}   
where $s$ is an integer parameter chosen to minimize the entropy as follows
\begin{equation}~\label{eqn5_4}
\begin{aligned}
s^{*}=&\arg\min_{s\in\left[k\right]} \sum_{j=1}^{k}q^{s}_{j}\log\left(\frac{1}{q^{s}_{j}}\right)
=\arg\min_{s\in\left[k\right]}\ \log\left(s\left(e^{\epsilon}-1\right)+k\right)-\frac{s\epsilon e^{\epsilon} }{ s\left(e^{\epsilon}-1\right)+k}\\
&=\arg\min_{s\in\left[k\right]}\ \log\left(s\left(e^{\epsilon}-1\right)+k\right)+\frac{\epsilon e^{\epsilon} k}{\left(e^{\epsilon}-1\right)\left( s\left(e^{\epsilon}-1\right)+k\right)}-\frac{\epsilon e^{\epsilon}}{e^{\epsilon}-1}.
\end{aligned}
\end{equation}
In order to solve the optimization problem~\eqref{eqn5_4}, we relax the problem by assuming $s$ is a real number taking values in $[0,k]$. The optimization problem in~\eqref{eqn5_4} is non-convex in for general values of $\epsilon$ and $k$. Thus, we get all local minima by setting the derivative to zero along with the boundary points $s\in\lbrace 0,k\rbrace$. Then we check all these critical points to obtain the global minimum point. However, we can see that at the boundary points $s\in\lbrace 0,k\rbrace$, the objective function is equal to $\log\left(k\right)$ which is the maximum entropy for any random variable with support size $k$. Hence, the optimal solution is one of the local minimums. We can verify that the objective function has only one local minimum point by setting the derivative with respect to $s$ to zero. Thus, we get
\begin{equation}~\label{eqn5_5}
\tilde{s}=k\frac{e^{\epsilon}\left(\epsilon-1\right)+1 }{\left(e^{\epsilon}-1\right)^{2}},
\end{equation}
where $\tilde{s}$ denotes the local minimum point. Since~\eqref{eqn5_4} is a continuous function in the real variable $s$, the optimal discrete point $s^{*}$ is within the local minimum $\tilde{s}$. Hence, we get the closest integer to the real value in~\eqref{eqn5_5}. As a result, we get
$$H\left(U\right)\geq H\left(U_{\min}^{s^{*}}\right),$$ where $s^{*}=\arg\min\limits_{s\in\lbrace \ceil{l},\floor{l}\rbrace}H\left(U_{\min}^{s}\right)$ for $l=k\frac{e^{\epsilon}\left(\epsilon-1\right)+1 }{\left(e^{\epsilon}-1\right)^{2}}$, and $U_{\min}^{s}$ is a random variable having a distribution $\mathbf{q}^{s^{*}}$ given in~\eqref{eqn5_8}. Hence, the proof of the necessary part is completed.

The sufficiency part is straightforward: Note that the random key $U_{\min}^{s^{*}}$ defined in Theorem~\ref{Th2_4} satisfies the necessary conditions, and Lemma~\ref{lemm5_1}, we can construct an $\epsilon$-LDP-Rec mechanism using the random key $U_{\min}^{s^{*}}$. Thus, these conditions are sufficient.

\subsection{Proof of Theorem~\ref{Th2_5}}~\label{Recov-B}
In this section, we show that the Hadamard response (HR) scheme proposed in~\cite{acharya2018hadamard} is, in fact, an $\epsilon$-LDP-Rec mechanism, where it is possible to recover the input $X$ from the output $Y$ and randomness $U$. Furthermore, we show that it is order optimal from a randomness perspective\footnote{We mention that the Hadamard mechanism in~\cite{acharya2018hadamard} is symmetric with non-binary outputs, while the Hadamard response in~\cite{acharya2019communication} has only binary outputs.}.

We briefly describe the HR mechanism, and then analyze its performance. We refer to~\cite{acharya2018hadamard} for more details. The HR mechanism is parameterized by two parameters: $K$ denotes the support size of the private mechanism output ($\mathcal{Y}=\left[K\right]$), and $s\leq K$ is a positive integer. For each $x\in\mathcal{X}$, let $\mathcal{C}_x\subseteq\left[K\right]$ be a subset of outputs of size $|\mathcal{C}_x|=s$. The private mechanism for HR is defined by
\begin{equation}
Q\left(y|X\right)=\left\{\begin{array}{ll}
\frac{e^{\epsilon}}{s e^{\epsilon}+K-s}& \text{if}\ y\in\mathcal{C}_x\\
\frac{1}{s e^{\epsilon}+K-s}& \text{if}\ y\notin \mathcal{C}_x\\
\end{array} \right.
\end{equation} 
We can easily show that this is a symmetric mechanism, i.e., it can be represented using a private key $U$ of size $|K|$ that is independent of the mechanism input $X$. Furthermore the distribution of the private key $U$ is given by $$\mathbf{q}^{\text{HR}}=\left[\underbrace{\frac{1}{se^{\epsilon}+K-s},\ldots,\frac{1}{se^{\epsilon}+K-s}}_{K-s\ \text{terms}},\underbrace{\frac{e^{\epsilon}}{se^{\epsilon}+K-s},\ldots,\frac{e^{\epsilon}}{se^{\epsilon}+K-s}}_{s\ \text{terms}}\right]$$
It remains to choose $K$, $s$, and $\lbrace \mathcal{C}_x\rbrace_{x\in\mathcal{X}}$ for fixed $\epsilon$ and input size $|\mathcal{X}|=k$. In~\cite[Section~$5$]{acharya2018hadamard}, the authors proposed $K=B\times b$ and $s=b/2$, where $B=2^{\ceil{\log_2\left(\min\lbrace e^{\epsilon},2k\rbrace\right)}-1}$, and $b=2^{\ceil{\log_2\left(\frac{k}{B}+1\right)}}$. Furthermore, each set $\mathcal{C}_x$ is a subset of rows indices of the Hadamard matrix. These parameters are chosen such that $s$ is close to $\max\lbrace \frac{k}{e^{\epsilon}},1\rbrace$, and $K$ is approximately the smallest power of $2$ greater than $k$. The reason behind using values that are powers of $2$ is to exploit the structure of the Hadamard matrix. In~\cite[Theorem~$7$]{acharya2018hadamard}, the authors proved that the minimax risk of HR for $\ell_{2}^{2}$ loss function is given by
\begin{equation}
r_{\epsilon,n,k}^{\ell_{2}^{2}}\leq\left\{\begin{array}{ll}
\mathcal{O}\left(\frac{k}{n\epsilon^{2}}\right)& \text{for}\ \epsilon <1\\
\mathcal{O}\left(\frac{k}{ne^{\epsilon}}\right)& \text{for}\ 1\leq\epsilon \leq \log\left(k\right)\\
\mathcal{O}\left(\frac{1}{n}\right)& \text{for}\ \epsilon >\log\left(k\right)
\end{array}\right.
\end{equation} 
which is order optimal for all privacy levels. In addition, the authors in~\cite{acharya2019communication} have shown that the HR scheme is order optimal for heavy hitter estimation in the high privacy regime ($\epsilon=\mathcal{O}\left(1\right)$). In the following, we analyze the performance of HR with respect to the randomness of the private mechanism. Observe that for fixed $\epsilon$ and $k$, the parameters $K$, $B$, and $b$ of HR is bounded by
$$\frac{\min\lbrace e^{\epsilon},2k\rbrace}{2}\leq B\leq \min\lbrace e^{\epsilon},2k\rbrace,\quad \frac{k}{\min\lbrace e^{\epsilon},2k\rbrace}\leq b\leq \frac{4k}{\min\lbrace e^{\epsilon},2k\rbrace},\quad k\leq K\leq 4k.$$
Hence, the entropy of the private key used to generate the HR private mechanism is bounded by
\begin{equation}~\label{eqn5_6}
\begin{aligned}
H^{\text{HR}}\left(U\right)&=\log\left(\frac{b}{2} e^{\epsilon}+K-\frac{b}{2}\right)-\frac{\epsilon e^{\epsilon}\frac{b}{2}}{\frac{b}{2}e^{\epsilon}+K-\frac{b}{2}}\\
&\leq\log\left(\frac{2k}{\min\lbrace e^{\epsilon},2k\rbrace}\left(e^{\epsilon}-1\right)+4k\right)-\frac{\epsilon e^{\epsilon}}{e^{\epsilon}-1+2\min\lbrace e^{\epsilon},2k\rbrace}\\
&= \left\{\begin{array}{ll}
\log\left(2k\frac{3e^{\epsilon}-1}{e^{\epsilon}}\right)-\frac{\epsilon e^{\epsilon}}{3e^{\epsilon}-1}& \text{if}\ \epsilon\leq\log\left(k\right)+1, \\
\log\left(e^{\epsilon}+4k-1\right)-\frac{\epsilon e^{\epsilon}}{e^{\epsilon}+4k-1}& \text{if}\ \epsilon>\log\left(k\right)+1.
\end{array}\right.
\end{aligned}
\end{equation}
The minimum entropy of the private key to generate an $\epsilon$-LDP-Rec mechanism is bounded by (Theorem~\ref{Th2_4})
\begin{equation}~\label{eqn5_7}
\begin{aligned}
H^{\min}\left(U\right)&= \log\left(s^{*}e^{\epsilon}+k-s^{*}\right)-\frac{\epsilon e^{\epsilon}s^{*}}{s^{*}e^{\epsilon}+k-s^{*}}\\
&\geq \left\{\begin{array}{ll}
 \log\left(k\left(\frac{\epsilon e^{\epsilon}}{e^{\epsilon}-1}\right)\right)-\frac{\epsilon e^{\epsilon}}{e^{\epsilon}+\frac{\left( e^{\epsilon}-1\right)^2}{e^{\epsilon}\left(\epsilon-1\right)+1}-1} &\text{if}\ \epsilon\leq \log\left(k\right),\\
 \log\left(e^{\epsilon}+k-1\right)-\frac{\epsilon e^{\epsilon}}{e^{\epsilon}+k-1} &\text{if}\ \epsilon>\log\left(k\right).\\
\end{array}\right.
\end{aligned}
\end{equation}
From~\eqref{eqn5_6} and~\eqref{eqn5_7}, we can verify that HR is randomness-order-optimal for all privacy levels $\epsilon$.

\section{Sequence of Distribution Estimation (Proof of Theorem~\ref{Th2_6})} 
\label{Recov_GDP} 

In this section, we prove Theorem~\ref{Th2_6} (see page~\pageref{Th2_6}). The main idea of our proof is as follows. The first condition is obtained in a similar manner as in the proof of Theorem~\ref{Th2_4}. For the second condition, we relate the minimum amount of randomness required to preserve privacy of $T$ samples to the minimum amount of randomness required to preserve privacy of $T-1$ samples. In particular, we prove that $H\left(U\right)\geq H\left(U_{\min,T-1}\right)+H\left(U_{\min,1}\right)$, where $H\left(U_{\min,t}\right)$ is the minimum amount of randomness of a  key when we have a database of $t$ input samples.

\begin{ddd} Let $U\in\mathcal{U}$ be a random key drawn from a discrete distribution $\mathbf{q}=\left[q_{1},\cdots,q_{k^{T}}\right]$ with a support size $|\mathcal{U}|=k^{T}$, where $q_{u}=\text{Pr}\left[U=u\right]$. We say that the distribution $\mathbf{q}$ satisfies $\epsilon$-DP, if there exists a bijective function $f:\mathcal{X}^{T}\to \left[1:k^{T}\right]$ from the dataset $\mathcal{X}^{T}$ to integers $\left[1:k^{T}\right]$, such that for every neighboring databases $\mathbf{x},\mathbf{x}^{\prime}\in\left[k\right]^{T}$, we have
\begin{equation}~\label{eqn6_1}
\frac{q_{f\left(\mathbf{x}\right)}}{q_{f\left(\mathbf{x}^{\prime}\right)}}\leq e^{\epsilon}.
\end{equation}
\end{ddd}
We begin our proof with the following lemma which is a generalized version of Lemma~\ref{lemm5_1}. 
We prove it in Appendix~\ref{AppG}. 

\begin{lemma}~\label{lemm6_1} Consider an input database $\mathbf{x}=\left(x^{\left(1\right)},\ldots,x^{\left(T\right)}\right)\in\left[k\right]^{T}$, and a random key $U\in\mathcal{U}=\lbrace u_1,\cdots,u_{k^{T}}\rbrace$ distributed according to an $\epsilon$-DP distribution $\mathbf{q}=\left[q_{1},\cdots,q_{k^{T}}\right]$. Then, there exists an $\epsilon$-DP-Rec mechanism $Q:\left[k\right]^{T}\to \left[k\right]^{T}$ that uses $U$ to create an output $Y^{T}\in \left[k\right]^{T}$, such that we can recover the input database $X^{T}$ from $(U,Y^{T})$. 
\end{lemma}
We can prove the first necessary condition of Theorem~\ref{Th2_6} (which is to show $|\mathcal{U}|\geq|\mathcal{Y}^{T}|\geq |\mathcal{X}^{T}|$) in the same way as we proved that for Theorem~\ref{Th2_4}.
For completeness, we provide a proof of it in Appendix~\ref{AppG}. 
Now we prove the necessity of the second condition. Consider an arbitrary $\epsilon$-DP-Rec mechanism $Q$ with output $Y^{T}\in\mathcal{Y}^{T}$ using a random key $U\in\mathcal{U}$, where $|\mathcal{Y}^{T}|=m\geq k^{T}$ and $|\mathcal{U}|=l\geq m$. Let $U\sim \mathbf{q}$, where $\mathbf{q}=\left[q_1,\ldots,q_l\right]$ such that $q_u=\text{Pr}\left[U=u\right]$ for $u\in\mathcal{U}$. Let $\mathcal{U}_{\mathbf{y}\mathbf{x}}\subset\mathcal{U}$ be a subset of key values such that the input $X^{T}=\mathbf{x}$ is mapped to $Y^{T}=\mathbf{y}$ when $U\in\mathcal{U}_{\mathbf{y}\mathbf{x}}$. Thus, the private mechanism $Q$ can be represented as 
\begin{equation}
Q\left(\mathbf{y}|\mathbf{x}\right)=\sum_{u\in\mathcal{U}_{\mathbf{y}\mathbf{x}}}q_u.
\end{equation}  
Observe that $\sum_{\mathbf{y}\in\mathcal{Y}^{T}}Q\left(\mathbf{y}|\mathbf{x}\right)=1$, since $Q\left(\mathbf{y}|\mathbf{x}\right)$ is a conditional distribution for any given $\mathbf{x}\in\left[k\right]^{T}$. Since $Q$ is an $\epsilon$-DP-Rec mechanism, it follows from the recoverability constraint that each input $\mathbf{x}$ is mapped to $\mathbf{y}$ using a different set of key values ($\mathcal{U}_{\mathbf{y}\mathbf{x}}\bigcap\mathcal{U}_{\mathbf{y}\mathbf{x}^{\prime}}=\phi$). Thus, for each $\mathbf{y}\in\mathcal{Y}^{T}$, we have $s_{\mathbf{y}}=\sum_{\mathbf{x}\in\left[k\right]^{T}}Q\left(\mathbf{y}|\mathbf{x}\right)\leq 1$. Furthermore, we get $\sum_{\mathbf{y}\in\mathcal{Y}^{T}}\sum_{\mathbf{x}\in\left[k\right]^{T}}Q\left(\mathbf{y}|\mathbf{x}\right)=\sum_{\mathbf{y}\in\mathcal{Y}^{T}}s_{\mathbf{y}}=k^{T}$. 

We sort the $k^{T}$ databases in $\mathcal{X}^{T}$ in lexicographic order by arranging them in increasing order of $x^{\left(1\right)}$. Then, we arrange the databases that have the same $x^{\left(1\right)}$ in increasing order of $x^{\left(2\right)}$ and so on. For example, database $\mathbf{x}=\left(x^{\left(1\right)},\ldots,x^{\left(i\right)},x^{\left(i+1\right)},\ldots,x^{\left(T\right)}\right)$ will appear before the database $\tilde{\mathbf{x}}=\left(x^{\left(1\right)},\ldots,x^{\left(i\right)},\tilde{x}^{\left(i+1\right)},\ldots,\tilde{x}^{\left(T\right)}\right)$ when $x^{\left(i+1\right)}<\tilde{x}^{\left(i+1\right)}$. Furthermore, we denote $\mathbf{x}_{i}$ as the $i$th database in the lexicographic order for $i\in\left[k\right]^{T}$. Observe that $s_{\mathbf{y}}=\sum_{\mathbf{x}\in\left[k\right]^{T}}Q\left(\mathbf{y}|\mathbf{x}\right)$ for given $\mathbf{y}\in\mathcal{Y}^{T}$. Thus, the probabilities $\mathbf{P}^{\mathbf{y}}= \left[P^{\mathbf{y}}_1,\ldots,P^{\mathbf{y}}_{k^{T}}\right]$ construct a valid distribution with support size $k^{T}$, where $P_{j}^{\mathbf{y}}=\frac{Q\left(\mathbf{y}|\mathbf{x}_j\right)}{s_{\mathbf{y}}}$ for $j\in\left[k\right]^{T}$. Furthermore, for every neighboring databases $\mathbf{x},\mathbf{x}^{\prime}\in\left[k\right]^{T}$, we have
\begin{equation}~\label{eqn6_7}
\frac{\nicefrac{Q\left(\mathbf{y}|\mathbf{x}\right)}{s_{\mathbf{y}}}}{\nicefrac{Q\left(\mathbf{y}|\mathbf{x}^{\prime}\right)}{s_{\mathbf{y}}}}=\frac{Q\left(\mathbf{y}|\mathbf{x}\right)}{Q\left(\mathbf{y}|\mathbf{x}^{\prime}\right)} \stackrel{\left(a\right)}{\leq} e^{\epsilon},
\end{equation} 
where step $\left(a\right)$ follows from the fact that $Q$ is an $\epsilon$-DP-Rec mechanism. Hence, the distribution $\mathbf{P}^{\mathbf{y}}$ is $\epsilon$-DP distribution. The proof of the following lemma is presented in Appendix~\ref{AppH}.

\begin{lemma}~\label{lemm6_2} For every output $\mathbf{y}\in\mathcal{Y}^{T}$, we have
$H\left(\mathbf{P}^{\mathbf{y}}\right)\geq H\left(U_{\min,T-1}\right)+H\left(U_{\min,1}\right),$
where $H\left(U_{\min,t}\right)$ denotes the minimum randomness of a private key when we have a database of $t$ samples for $t\in\lbrace 1,\ldots,T\rbrace$.
\end{lemma}

Using Lemma~\ref{lemm6_2}, we can prove Theorem~\ref{Th2_6} as follows.

\begin{align}
H\left(U\right)&=\frac{1}{k^{T}}\sum_{\mathbf{x}\in\left[k\right]^{T}}H\left(U\right)\stackrel{\left(a\right)}{\geq} \frac{1}{k^{T}}\sum_{\mathbf{x}\in\left[k\right]^{T}}H\left(Y^{T}|X^{T}=\mathbf{x}\right) \notag\\
&=\frac{1}{k^{T}}\sum_{\mathbf{x}\in\left[k\right]^{T}}\sum_{\mathbf{y}\in\mathcal{Y}^{T}}-Q\left(\mathbf{y}|\mathbf{x}\right)\log\left(Q\left(\mathbf{y}|\mathbf{x}\right)\right) \notag \\
&=\frac{1}{k^{T}}\sum_{\mathbf{y}\in\mathcal{Y}^{T}}\left[s_{\mathbf{y}}\left(\sum_{\mathbf{x}\in\left[k\right]^{T}}-\frac{Q\left(\mathbf{y}|\mathbf{x}\right)}{s_{\mathbf{y}}}\log\left(\frac{Q\left(\mathbf{y}|\mathbf{x}\right)}{s_{\mathbf{y}}}\right)\right)-s_{\mathbf{y}}\log\left(s_{\mathbf{y}}\right)\right] \notag \\
&=\frac{1}{k^{T}}\sum_{\mathbf{y}\in\mathcal{Y}^{T}}\big[s_{\mathbf{y}}H\left(\mathbf{P}^{\mathbf{y}}\right)-s_{\mathbf{y}}\log\left(s_{\mathbf{y}}\right)\big] \notag \\
&\stackrel{\left(b\right)}{\geq} \frac{1}{k^{T}}\sum_{\mathbf{y}\in\mathcal{Y}^{T}} \big[s_{\mathbf{y}}\left(H\left(U_{\min,T-1}\right)+H\left(U_{\min,1}\right)\right)-s_{\mathbf{y}}\log\left(s_{\mathbf{y}}\right)\big] \notag\\
&\stackrel{\left(c\right)}{\geq} H\left(U_{\min,T-1}\right)+H\left(U_{\min,1}\right), \label{eqn6_6}
\end{align}
where step $\left(a\right)$ follows from the fact that $Q\left(\mathbf{y}|\mathbf{x}\right)$ is a function of $U$. Step $\left(b\right)$ follows from Lemma~\ref{lemm6_2}. The inequality $\left(c\right)$ follows from solving the problem 
\begin{equation}~\label{opt_10}
\begin{aligned}
\min_{\lbrace s_{\mathbf{y}}\rbrace}&\  \sum_{\mathbf{y}\in\mathcal{Y}^{T}}s_{\mathbf{y}}\left[H\left(U_{\min,T-1}\right)+H\left(U_{\min,1}\right)\right]-s_{\mathbf{y}}\log\left(s_{\mathbf{y}}\right)\\
s.t.&\ \sum_{\mathbf{y}\in\mathcal{Y}^{T}}s_{\mathbf{y}}=k^{T}\quad  \text{ and } 0\leq s_{\mathbf{y}}\leq 1, \ \forall\ \mathbf{y}\in\mathcal{Y}^{T}
\end{aligned}
\end{equation}
Note that $f\left(x\right)=-x\log\left(x\right)$ is a concave function on $0\leq x\leq 1$. Therefore, the objective function in~\eqref{opt_10} is concave in $\lbrace s_{\mathbf{y}}\rbrace$. The minimum value of a concave function is one of the vertices which is obtained when all the inequalities are satisfied by equalities. 
By setting $k^{T}$ of the $s_{\mathbf{y}}$'s to be one and setting the remaining $|\mathcal{Y}^{T}|-k^T$ of $s_{\mathbf{y}}$'s to be zero, the objective value in \eqref{opt_10} becomes $k_T$, which gives inequality (c).

Now, from~\eqref{eqn6_6}, we conclude that $H\left(U\right)\geq TH\left(U_{\min,1}\right)$,
where $H\left(U_{\min,1}\right)$ is the minimum amount of randomness required to design an $\epsilon$-LDP-Rec mechanism given in Theorem~\ref{Th2_4}. This completes the proof of Theorem~\ref{Th2_6}.

%%%%%%%%%%%%%%%%%%%%%%%%%%%%%%%%%%%%%%%%%%%%%%%%%%%%%%%%%%%%%%%%%%%%%%%%%%%%%%%%%%%%%%%%%%%%%%%%%%%%%%%%%%%%
%%%%%%%%%%%%%%%%%%%%%%%%%%%%%%%%%%%%%%%%%%%%%%%%%%%%%%%%%%%%%%%%%%%%%%%%%%%%%%%%%%%%%%%%%%%%%%%%%%%%%%%%%%%%
%%%%%%%%%%%%%%%%%%%%%%%%%%%%%%%%%%%%%%%%%%%%%%%%%%%%%%%%%%%%%%%%%%%%%%%%%%%%%%%%%%%%%%%%%%%%%%%%%%%%%%%%%%%%
%%%%%%%%%%%%%%%%%%%%%%%%%%%%%%%%%%%%%%%%%%%%%%%%%%%%%%%%%%%%%%%%%%%%%%%%%%%%%%%%%%%%%%%%%%%%%%%%%%%%%%%%%%%%
%\bibliographystyle{alpha}
%\bibliography{mybibfile}

\newcommand{\etalchar}[1]{$^{#1}$}

%%%%%%%%%%%%%%%%%%%%%%%%%%%%%%%%%%%%%%%%%%%%%%%%%%%%%%%%%%%%%%%%%%%%%%%%%%%%%%%%%%%%%%%%%%%%%%%%%%%%%%%%%%%%
%%%%%%%%%%%%%%%%%%%%%%%%%%%%%%%%%%%%%%%%%%%%%%%%%%%%%%%%%%%%%%%%%%%%%%%%%%%%%%%%%%%%%%%%%%%%%%%%%%%%%%%%%%%%
%%%%%%%%%%%%%%%%%%%%%%%%%%%%%%%%%%%%%%%%%%%%%%%%%%%%%%%%%%%%%%%%%%%%%%%%%%%%%%%%%%%%%%%%%%%%%%%%%%%%%%%%%%%%
%%%%%%%%%%%%%%%%%%%%%%%%%%%%%%%%%%%%%%%%%%%%%%%%%%%%%%%%%%%%%%%%%%%%%%%%%%%%%%%%%%%%%%%%%%%%%%%%%%%%%%%%%%%%
%%%%%%%%%%%%%%%%%%%%%%%%%%%%%%%%%%%%%%%%%%%%%%%%%%%%%%%%%%%%%%%%%%%%%%%%%%%%%%%%%%%%%%%%%%%%%%%%%%%%%%%%%%%%
\newpage
\appendix

\section{Lower Bound on The Minimax Risk Estimation Using Fisher Information}~\label{LDP_F}
In this section, we introduce an alternative proof of Theorem~\ref{Th2_1}. Our proof is inspired by the approach in~\cite{barnes2019learning} that uses Fisher information to bound the minimax risk estimation under communication constraints. The main idea of our proof is to formulate a non-convex optimization problem to bound the Fisher information matrix under privacy and randomness constraints. Let $\overline{\mathcal{P}}\subset \Delta_k$ be a subset of simplex $\Delta_k$ defined by 
$$\overline{\mathcal{P}}=\left\{ \mathbf{p}\in\mathbb{R}^{k}:\sum\limits_{j=1}^{k}p_j=1,\ \frac{1}{k}\leq p_j\leq \frac{2}{k} ,\ p_{j+k/2}=\frac{2}{k}-p_{j},\ \forall j\in\left[k/2\right]\right\}.$$
For every $\mathbf{p}\in\overline{\mathcal{P}}$, the number of free variables is $k/2$, where each parameter $p_{j+k/2}$ is associated with the variable $p_j$, $\forall\ j\in\left[k/2\right]$. For a given distribution $\mathbf{p}\in\Delta_k$, we define the marginal distribution on the output $Y$ as
\begin{equation}
\mathbf{M}\left(y|\mathbf{p}\right)=\sum_{j=1}^{k}Q\left(Y=y|X=j\right)p_j.
\end{equation}    
Let $S_{\mathbf{p}}\left(y\right)$ denote the $k/2$-vector score function of $Y$ given by
\begin{equation}
\begin{aligned}
S_{\mathbf{p}}\left(y\right)&=\left[S_{p_1}\left(y\right),\ldots,S_{p_{k/2}}\left(y\right)\right]\\
&=\left[\frac{\partial\log\left( \mathbf{M}\left(y|\mathbf{p}\right)\right)}{\partial p_1},\ldots,\frac{\partial\log\left( \mathbf{M}\left(y|\mathbf{p}\right)\right)}{\partial p_{k/2}}\right].
\end{aligned}
\end{equation}
Then, the Fisher information matrix for estimating $\mathbf{p}\in\overline{\mathcal{P}}$ from $Y$ is given by
\begin{equation}
I_{Y}\left(\mathbf{p}\right)=\mathbb{E}\left[S_{\mathbf{p}}\left(y\right)S_{\mathbf{p}}\left(y\right)^{T}\right],
\end{equation}
where the expectation is taken over the randomness in the output $Y$. Now, consider the following inequalities
\begin{equation}~\label{eqn3_4}
\begin{aligned}
r^{\ell_{2}^{2}}_{\epsilon,R,n,k}&=\inf_{\lbrace Q_i\in\mathcal{Q}_{\left(\epsilon,R\right)}\rbrace} \inf_{\hat{\mathbf{p}}}\sup_{\mathbf{p}\in\Delta_k}\mathbb{E}\left[\ell_{2}^{2}\left(\hat{\mathbf{p}}\left(\mathbf{Y}^{n}\right),\mathbf{p}\right)\right] \\
&\geq  \inf_{\lbrace Q_i\in\mathcal{Q}_{\left(\epsilon,R\right)}\rbrace} \inf_{\hat{\mathbf{p}}}\sup_{\mathbf{p}\in\overline{\mathcal{P}}}\mathbb{E}\left[\ell_{2}^{2}\left(\hat{\mathbf{p}}\left(\mathbf{Y}^{n}\right),\mathbf{p}\right)\right]\\
&\stackrel{\left(a\right)}{\geq} \frac{\left(k/2\right)^2}{\sup\limits_{\lbrace Q_i\in\mathcal{Q}_{\left(\epsilon,R\right)}\rbrace} \sup\limits_{\mathbf{p}\in\overline{\mathcal{P}}}\text{Tr}\left(I_{Y^{n}}\left(\mathbf{p}\right)\right)+\frac{k}{2}\pi^{2}}
\end{aligned}
\end{equation}
where $I_{Y^{n}}\left(\mathbf{p}\right)$ denotes the Fisher information matrix for estimating $\mathbf{p}$ from $Y^{n}=\left[Y_1,\ldots,Y_n\right]$, and $\text{Tr}\left(I_{Y^{n}}\left(\mathbf{p}\right)\right)$ denotes the trace of the Fisher information matrix $I_{Y^{n}}\left(\mathbf{p}\right)$. Step $\left(a\right)$ follows from the van Trees inequality~\cite{barnes2019learning}[Eqn.$4$-$8$]. Our goal is to bound the term $\sup_{\lbrace Q_i\in\mathcal{Q}_{\left(\epsilon,R\right)}\rbrace} \sup_{\mathbf{p}\in\overline{\mathcal{P}}}\text{Tr}\left(I_{Y^{n}}\left(\mathbf{p}\right)\right)$. For a given distribution $\mathbf{p}\in\overline{\mathcal{P}}$, the random variables $Y_1,\ldots,Y_n$ are independent. As a result, the trace of the Fisher information matrix for estimating $\mathbf{p}$ from $Y_1,\ldots,Y_n$ is bounded by
\begin{equation}~\label{eqn3_3}
\begin{aligned}
\sup\limits_{\lbrace Q_i\in\mathcal{Q}_{\left(\epsilon,R\right)}\rbrace} &\sup_{\mathbf{p}\in\overline{\mathcal{P}}}\text{Tr}\left(I_{Y^{n}}\left(\mathbf{p}\right)\right) \\
&\stackrel{\left(a\right)}{=}\sup\limits_{\lbrace Q_i\in\mathcal{Q}_{\left(\epsilon,R\right)}\rbrace}\sup_{\mathbf{p}\in\overline{\mathcal{P}}}\sum_{i=1}^{n}\text{Tr}\left(I_{Y_i}\left(\mathbf{p}\right)\right)\\
&\leq \sup\limits_{\lbrace Q_i\in\mathcal{Q}_{\left(\epsilon,R\right)}\rbrace} \sup_{\mathbf{p}\in\overline{\mathcal{P}}} n\sup_{i\in\left[n\right]} \text{Tr}\left(I_{Y_i}\left(\mathbf{p}\right)\right)\\
&\stackrel{\left(b\right)}{\leq}  \left\{ \begin{array}{ll}
2nk \frac{e^{\epsilon}\left(e^\epsilon-1\right)^2}{\left(e^\epsilon+1\right)^2}& \text{if}\ R\geq H_2\left(\frac{e^\epsilon}{e^\epsilon+1}\right)\\
2nk \frac{p_{R}^2 \left(e^\epsilon-1\right)^2}{e^{\epsilon}}& \text{if}\ R< H_2\left(\frac{e^\epsilon}{e^\epsilon+1}\right)
\end{array}
\right.
\end{aligned}
\end{equation} 
where step $\left(a\right)$ follows from the chain rule of the Fisher information~\cite{zamir1998proof}[Lemma~$1$]. Step $\left(b\right)$ follows from Lemma~\ref{lemm3_1} presented below. Substituting from~\eqref{eqn3_3} into~\eqref{eqn3_4}, we get
\begin{equation}
r^{\ell_{2}^{2}}_{\epsilon,R,n,k}\geq \left\{ \begin{array}{ll} 
 \frac{k\left(e^\epsilon+1\right)^2}{16ne^\epsilon\left(e^\epsilon-1\right)^2}& \text{if}\ R\geq H_2\left(\frac{e^\epsilon}{e^\epsilon+1}\right)\\
  \frac{ke^{\epsilon}}{16np_{R}^2 \left(e^\epsilon-1\right)^2}& \text{if}\ R< H_2\left(\frac{e^\epsilon}{e^\epsilon+1}\right)
\end{array}
\right.
\end{equation}
for $n\geq 4\frac{e^{\epsilon}}{p_R^2\left(e^{\epsilon}-1\right)^2}$. 

\begin{lemma}~\label{lemm3_1} For any $\left(\epsilon,R\right)$-LDP mechanism, the trace of the Fisher information matrix $I_{Y}\left(\mathbf{p}\right)$ is bounded by
\begin{equation}
\sup_{Q\in\mathcal{Q}_{\left(\epsilon,R\right)}}\sup_{\mathbf{p}\in\overline{\mathcal{P}}} \text{Tr}\left(I_{Y}\left(\mathbf{p}\right)\right)\leq \left\{ \begin{array}{ll}
2k \frac{e^{\epsilon}\left(e^\epsilon-1\right)^2}{\left(e^\epsilon+1\right)^2}& \text{if}\ R\geq H_2\left(\frac{e^\epsilon}{e^\epsilon+1}\right)\\
2k \frac{p_{R}^2 \left(e^\epsilon-1\right)^2}{e^{\epsilon}}& \text{if}\ R< H_2\left(\frac{e^\epsilon}{e^\epsilon+1}\right)
\end{array}
\right.
\end{equation} 
where $ H_2\left(\frac{e^\epsilon}{e^\epsilon+1}\right)$ is the Shannon entropy, and $p_R<0.5$ denotes the inverse Shannon entropy $p_R=h^{-1}\left(R\right)$.
\end{lemma}

\begin{proof}
For a given distribution $\mathbf{p}\in\overline{\mathcal{P}}$, we have
\begin{equation}
\begin{aligned}
S_{p_j}\left(y\right)&=\frac{\partial\log\left( \mathbf{M}\left(y|\mathbf{p}\right)\right)}{\partial p_j}\\
&=\frac{Q\left(y|j\right)-Q\left(y|j+k/2\right)}{\mathbf{M}\left(y|\mathbf{p}\right)},
\end{aligned}
\end{equation}
for $j\in\left[k/2\right]$. By taking the expectation with respect to $Y$, we get
\begin{equation}
\mathbb{E}\left[ S_{p_j}\left(Y\right)^2\right]=\sum_{y\in\mathcal{Y}}\frac{\left(Q\left(y|j\right)-Q\left(y|j+k/2\right)\right)^2}{\sum_{j^{\prime}=1}^{k}Q\left(y|j^{\prime}\right)p_{j^{\prime}}}
\end{equation}
Thus, the trace of the Fisher information matrix is given by
\begin{equation}
\begin{aligned}
\text{Tr}\left(I_{Y}\left(\mathbf{p}\right)\right)&=\sum_{j=1}^{k/2}\mathbb{E}\left[ S_{p_j}\left(Y\right)^2\right]\\
&=\sum_{j=1}^{k/2}\sum_{y\in\mathcal{Y}}\frac{\left(Q\left(y|j\right)-Q\left(y|j+k/2\right)\right)^2}{\sum_{j^{\prime}=1}^{k}Q\left(y|j^{\prime}\right)p_{j^{\prime}}}\\
&\leq \frac{k}{2}\max_{j\in\left[k/2\right]}\sum_{y\in\mathcal{Y}}\frac{\left(Q\left(y|j\right)-Q\left(y|j+k/2\right)\right)^2}{\sum_{j^{\prime}=1}^{k}Q\left(y|j^{\prime}\right)p_{j^{\prime}}}\\
&\stackrel{\left(a\right)}{\leq} k e^{\epsilon}\max_{j\in\left[k/2\right]}\sum_{y\in\mathcal{Y}}\frac{\left( Q\left(y|j\right)-Q\left(y|j+k/2\right)\right)^2}{Q\left(y|j\right)+Q\left(y|j+k/2\right)}\\
&\stackrel{\left(b\right)}{\leq} \left\{ \begin{array}{ll}
2k \frac{e^{\epsilon}\left(e^\epsilon-1\right)^2}{\left(e^\epsilon+1\right)^2}& \text{if}\ R\geq H_2\left(\frac{e^\epsilon}{e^\epsilon+1}\right)\\
2k \frac{p_{R}^2 \left(e^\epsilon-1\right)^2}{e^{\epsilon}}& \text{if}\ R< H_2\left(\frac{e^\epsilon}{e^\epsilon+1}\right)
\end{array}
\right.
\end{aligned}
\end{equation} 
where step $\left(a\right)$ follows from the fact that $Q\left( y|j^{\prime}\right)\geq e^{-\epsilon} Q\left( y|j\right)$ and $Q\left( y|j^{\prime}\right)\geq e^{-\epsilon} Q\left( y|j+k/2\right),\ \forall j^{\prime}\in\left[k\right]$. Thus, we have 
\begin{equation}
\begin{aligned}
\sum\limits_{j^{\prime}=1}^{k}Q\left( y|j^{\prime}\right)p_{j^{\prime}}&\geq e^{-\epsilon}\frac{Q\left( y|j\right)+Q_i\left( y|j+k/2\right)}{2}\sum_{j^{\prime}=1}^{k}p_{j^{\prime}}\\
&=e^{-\epsilon}\frac{Q\left( y|j\right)+Q\left( y|j+k/2\right)}{2}
\end{aligned}
\end{equation} 
Step $\left(b\right)$ follows from Lemma~\ref{lemma_1} presented at the end of Section~\ref{LDP_AS}. This completes the proof of Lemma~\ref{lemm3_1}.
\end{proof}

\section{Proof of Lemma~\ref{lemm3_4}}
\label{AppK} 
We start our proof by Assoud's method. 

\begin{lemma}~\label{lemm3_2}(Assouad's Method~\cite{duchi2018minimax}) For the family of distributions $\left\{ \mathbf{p}^{\nu}:\nu\in\mathcal{V}=\lbrace -1,1\rbrace^{k/2}\right\}$, and a loss function $\ell\left(\hat{\mathbf{p}},\mathbf{p}\right)=\sum_{j=1}^{k}\phi\left(\hat{p}_j-p_j\right)$ defined in Section~\ref{LDP_AS}, we have
\begin{equation}~\label{eqn3_5}
\begin{aligned}
r^{\ell}_{\epsilon,R,n,k}\left(Q^{n}\right)&=\inf_{\hat{\mathbf{p}}}\sup_{\mathbf{p}\in\Delta_k}\mathbb{E}\left[\ell\left(\hat{\mathbf{p}}\left(Y^{n}\right),\mathbf{p}\right)\right]\\
&\geq \phi\left(\delta\right)\sum_{j=1}^{k/2}\left(1-||\mathbf{M}^{n}_{+j}-\mathbf{M}^{n}_{-j}||_{\text{TV}}\right)
\end{aligned}
\end{equation}
\end{lemma}
For completeness, we present the proof of Lemma~\ref{lemm3_2} in Appendix~\ref{AppC}. Let $\lbrace e_j\rbrace_{j=1}^{k/2}$ be the standard basis of $\mathbb{R}^{k/2}$. Consider now the following inequalities:
{\allowdisplaybreaks
\begin{equation}~\label{eqn3_6}
\begin{aligned}
\sum_{j=1}^{k/2}\left(1-\left\|\mathbf{M}^{n}_{+j}-\mathbf{M}^{n}_{-j}\right\|_{\text{TV}}\right)&\stackrel{\left(a\right)}{\geq} \sum_{j=1}^{k/2}\left(1-\frac{1}{|\mathcal{V}|}\sum_{\nu:\nu_j=1}||\left(\prod_{i=1}^{n}\mathbf{M}^{\nu}_{i}\right)-\left(\prod_{i=1}^{n}\mathbf{M}^{\nu-2e_j}_{i}\right)||_{\text{TV}}\right)\\
&\geq \sum_{j=1}^{k/2}\left(1-\sup_{\nu:\nu_j=1}||\left(\prod_{i=1}^{n}\mathbf{M}^{\nu}_{i}\right)-\left(\prod_{i=1}^{n}\mathbf{M}^{\nu-2e_j}_{i}\right)||_{\text{TV}}\right)\\
&\stackrel{\left(b\right)}{\geq} \sum_{j=1}^{k/2}\left(1-\sup_{\nu:\nu_j=1}\sqrt{\frac{1}{2}D_{\text{KL}}\left(\left(\prod_{i=1}^{n}\mathbf{M}^{\nu}_{i}\right)||\left(\prod_{i=1}^{n}\mathbf{M}^{\nu-2e_j}_{i}\right)\right)}\right)\\
&\stackrel{\left(c\right)}{\geq}\sum_{j=1}^{k/2}\left(1-\sqrt{\frac{1}{2}\sup_{\nu:\nu_j=1}\sum_{i=1}^{n}D_{\text{KL}}\left(\mathbf{M}^{\nu}_{i}||\mathbf{M}^{\nu-2e_j}_{i}\right)}\right)\\
&=\frac{k}{2}\left(1-\frac{2}{k}\sum_{j=1}^{k/2}\sqrt{\frac{1}{2}\sup_{\nu:\nu_j=1}\sum_{i=1}^{n}D_{\text{KL}}\left(\mathbf{M}^{\nu}_{i}||\mathbf{M}^{\nu-2e_j}_{i}\right)}\right)\\
&\stackrel{\left(d\right)}{\geq} \frac{k}{2}\left(1-\sqrt{\frac{1}{k}\sum_{j=1}^{k/2}\sup_{\nu:\nu_j=1}\sum_{i=1}^{n}D_{\text{KL}}\left(\mathbf{M}^{\nu}_{i}||\mathbf{M}^{\nu-2e_j}_{i}\right)}\right)\\
&\geq \frac{k}{2}\left(1-\sqrt{\frac{n}{2}\sup_{j\in\left[k/2\right]}\sup_{i\in\left[n\right]}\sup_{\nu:\nu_j=1}D_{\text{KL}}\left(\mathbf{M}^{\nu}_{i}||\mathbf{M}^{\nu-2e_j}_{i}\right)}\right)\\
\end{aligned}
\end{equation}
}
where step $\left(a\right)$ follows from the triangular inequality. Step $\left(b\right)$ follows from Pinsker's inequality that states that for any two distributions $\mathbf{P}$ and $\mathbf{Q}$, we get $\|\mathbf{P}-\mathbf{Q}\|_{\text{TV}}\leq\sqrt{\frac{1}{2}D\left(\textbf{P}||\textbf{Q}\right)}$~\cite[Lemma~$2.5$]{tsybakov2008introduction}. Step $\left(c\right)$ follows from the properties of KL-divergence. Step $\left(d\right)$ follows from the concavity of function $\sqrt{x}$. Substituting from~\eqref{eqn3_6} into~\eqref{eqn3_5}, we get
\begin{equation}~\label{eqn3_7}
\begin{aligned}
r^{\ell}_{\epsilon,R,n,k}&=\inf_{\lbrace Q_i\in\mathcal{Q}_{\left(\epsilon,R\right)}\rbrace} r^{\ell}_{\epsilon,R,n,k}\left(Q^{n}\right)\\
&\geq  \inf_{\lbrace Q_i\in\mathcal{Q}_{\left(\epsilon,R\right)}\rbrace} \phi\left(\delta\right)\frac{k}{2}\left(1-\sqrt{\frac{n}{2}\sup_{j\in\left[k/2\right]}\sup_{i\in\left[n\right]}\sup_{\nu:\nu_j=1}D_{\text{KL}}\left(\mathbf{M}^{\nu}_{i}||\mathbf{M}^{\nu-2e_j}_{i}\right)}\right)\\
&=\phi\left(\delta\right)\frac{k}{2}\left(1-\sqrt{\frac{n}{2}\sup_{j\in\left[k/2\right]}\sup_{i\in\left[n\right]}\sup_{\nu:\nu_j=1}\sup_{Q_i\in\mathcal{Q}_{\left(\epsilon,R\right)}} D_{\text{KL}}\left(\mathbf{M}^{\nu}_{i}||\mathbf{M}^{\nu-2e_j}_{i}\right)}\right)
\end{aligned}
\end{equation}
Hence the proof is completed.

\section{Proof of Lemma~\ref{lemmb_1}}
\label{AppB} 

\begin{lemma} The optimal solution of the non-convex optimization problem \textbf{P1} is obtained when the the output size is $m=2$.
\end{lemma}
\begin{proof}
Note that if $m=1$, then the optimal value of \textbf{P1} will be zero, and hence, we have $m\geq 2$. In the following, we prove that the optimal solution is achievable at $m=2$. Let $$f\left(\mathbf{q}^{m}_j,\mathbf{q}^{m}_{j+k/2}\right)=\sum_{l=1}^{m}\frac{\left(q_{l,j}-q_{l,j+k/2}\right)^{2}}{q_{l,j}+q_{l,j+k/2}}$$ denote the objective function of the problem~\textbf{P1}, where $\mathbf{q}^{m}_j=\left[q_{1,j},\ldots,q_{m,j}\right]$ and $\mathbf{q}^{m}_{j+k/2}=\left[q_{1,j+k/2},\ldots,q_{m,j+k/2}\right]$. Suppose that the optimal solution is obtained at $m>2$. In other words, there exist two distributions $\mathbf{q}_{j}^{m}$ and $\mathbf{q}^{m}_{j+k/2}$ with size $m>2$ that maximize the objective function $f\left(\mathbf{q}^{m}_j,\mathbf{q}^{m}_{j+k/2}\right)$ and satisfy the constraints~\eqref{eqnb_4}-\eqref{eqnb_7}. We prove that if $\mathbf{q}_{j}^{m}$ and $\mathbf{q}^{m}_{j+k/2}$ are optimal, then there exist two distributions $\tilde{\mathbf{q}}_{j}^{m-1}$ and $\tilde{\mathbf{q}}^{m-1}_{j+k/2}$ with support size $m-1$ that satisfy the problem constraints and achieve at least the same objective value as $\mathbf{q}_{j}^{m}$ and $\mathbf{q}^{m}_{j+k/2}$. Let $\tilde{\mathbf{q}}_{j}^{m-1}=\left[q_{1,j},\ldots,q_{m-2,j},q_{m-1,j}+q_{m,j}\right]$ and $\tilde{\mathbf{q}}_{j+k/2}^{m-1}=\left[q_{1,j+k/2},\ldots,q_{m-2,j+k/2},q_{m-1,j}+q_{m,j+k/2}\right]$. We can easily verify that $H\left(\tilde{\mathbf{q}}_{j}^{m-1}\right)\leq R$ as $H\left(\mathbf{q}_{j}^{m}\right)\leq R$ and $H\left(\tilde{\mathbf{q}}_{j+k/2}^{m-1}\right)\leq R$ as $H\left(\mathbf{q}_{j+k/2}^{m}\right)\leq R$. Furthermore, we have
\begin{equation}
e^{-\epsilon}=e^{-\epsilon} \frac{q_{m-1,j+k/2}+q_{m,j+k/2}}{q_{m-1,j+k/2}+q_{m,j+k/2}}\leq\frac{q_{m-1,j}+q_{m,j}}{q_{m-1,j+k/2}+q_{m,j+k/2}}\leq e^{\epsilon}\frac{q_{m-1,j+k/2}+q_{m,j+k/2}}{q_{m-1,j+k/2}+q_{m,j+k/2}} = e^{\epsilon}
\end{equation}
Hence, the distributions $\tilde{\mathbf{q}}_{j}^{m-1}$ and $\tilde{\mathbf{q}}^{m-1}_{j+k/2}$ satisfy the constraints of the problem~\textbf{P1}. Consider the following inequalities
\begin{equation}
\begin{aligned}
&f\left(\tilde{\mathbf{q}}^{m}_j,\tilde{\mathbf{q}}^{m-1}_{j+k/2}\right)-f\left(\mathbf{q}^{m}_j,\mathbf{q}^{m}_{j+k/2}\right)\\
&\quad=\frac{\left(q_{m-1,j}+q_{m,j}-q_{m-1,j+k/2}+q_{m,j+k/2}\right)^{2}}{q_{m-1,j}+q_{m,j}+q_{m-1,j+k/2}+q_{m,j+k/2}}-\left[\frac{\left(q_{m-1,j}-q_{m-1,j+k/2}\right)^{2}}{q_{m-1,j}+q_{m-1,j+k/2}}+\frac{\left(q_{m,j}-q_{m,j+k/2}\right)^{2}}{q_{m,j}+q_{m,j+k/2}}\right]\\
&\quad\stackrel{\left(a\right)}{\geq} \frac{\left(q_{m-1,j}+q_{m,j}-q_{m-1,j+k/2}+q_{m,j+k/2}\right)^{2}}{q_{m-1,j}+q_{m,j}+q_{m-1,j+k/2}+q_{m,j+k/2}}-2\frac{\left(\frac{q_{m-1,j}+q_{m,j}}{2}-\frac{q_{m-1,j+k/2}+q_{m,j+k/2}}{2}\right)^{2}}{\frac{q_{m-1,j}+q_{m,j}}{2}+\frac{q_{m-1,j+k/2}+q_{m,j+k/2}}{2}}\\
&\quad=0
\end{aligned}
\end{equation}
where step $\left(a\right)$ follows from the convexity of the function $\left(x-y\right)^{2}/\left( x+y\right)$ for $x,y\in\left[0:1\right]$. Hence the distributions $\tilde{\mathbf{q}}^{m}_j,\tilde{\mathbf{q}}^{m-1}_{j+k/2}$ have at least the same objective value as $\mathbf{q}_{j}^{m}$ and $\mathbf{q}^{m}_{j+k/2}$.
\end{proof}

\section{Proof of Lemma~\ref{lemm3_2}}
\label{AppC} 
Consider an arbitrary estimator $\hat{\mathbf{p}}$, then we have
\begin{equation}
\begin{aligned}
\sup_{\mathbf{p}\in\Delta_k} & \mathbb{E}\left[\ell\left(\hat{\mathbf{p}}\left(Y^{n}\right),\mathbf{p}\right)\right]\geq \sup_{\nu\in\mathcal{V}}\mathbb{E}\left[\ell\left(\hat{\mathbf{p}}\left(Y^{n}\right),\mathbf{p}^{\nu}\right)\right]\\
&\geq \frac{1}{|\mathcal{V}|}\sum_{\nu\in\mathcal{V}}\mathbb{E}\left[\ell\left(\hat{\mathbf{p}}\left(Y^{n}\right),\mathbf{p}^{\nu}\right)\right]\\
&\geq \phi\left(\delta\right) \frac{1}{|\mathcal{V}|}\sum_{\nu\in\mathcal{V}}\mathbb{E}\left[\sum_{j=1}^{k/2}\mathbbm{1}\left(\psi_{j}\left(Y^{n}\right)\neq \nu_j\right)\right]\\
&\geq \phi\left(\delta\right) \sum_{j=1}^{k/2}\left( \frac{1}{|\mathcal{V}|}\sum_{\nu\in\mathcal{V}:\nu_j=+1}\mathbb{E}\left[\mathbbm{1}\left(\psi_{j}\left(Y^{n}\right)\neq +1\right)\right]+\frac{1}{|\mathcal{V}|}\sum_{\nu\in\mathcal{V}:\nu_j=-1}\mathbb{E}\left[\mathbbm{1}\left(\psi_{j}\left(Y^{n}\right)\neq -1\right)\right]\right)\\
&\geq \phi\left(\delta\right) \sum_{j=1}^{k/2}\inf_{\psi}\left( \frac{1}{|\mathcal{V}|}\sum_{\nu\in\mathcal{V}:\nu_j=+1}\text{Pr}\left[\psi_{j}\left(Y^{n}\right)\neq +1\right]+\frac{1}{|\mathcal{V}|}\sum_{\nu\in\mathcal{V}:\nu_j=-1}\text{Pr}\left[\psi_{j}\left(Y^{n}\right)\neq -1\right]\right)\\
&= \phi\left(\delta\right) \sum_{j=1}^{k/2}\frac{1}{2}\inf_{\psi}\left(\mathbf{M}_{+j}^{n}\left[\psi_{j}\left(\mathbf{Y}^{n}\right)\neq +1\right]+\mathbf{M}_{+j}^{n}\left[\psi_{j}\left(\mathbf{Y}^{n}\right)\neq -1\right]\right)\\
&\geq \phi\left(\delta\right)\sum_{j=1}^{k/2}\left(1-||\mathbf{M}^{n}_{+j}-\mathbf{M}^{n}_{-j}||_{\text{TV}}\right)
\end{aligned}
\end{equation}
where $\psi=\left(\psi_1,\ldots,\psi_{k/2}\right)$ is a vector of test functions.

\section{Proof of Lemma~\ref{lemm4_1}}
\label{AppD} 
 We claim that the conditional distribution on $Y_{i}^{j}|X_i$ is given by
\begin{equation}~\label{eqnd_1}
\text{Pr}\left[Y_{i}^{j}=1|X_i\right]=\left\{\begin{array}{ll}
\frac{e^{\epsilon_j}}{e^{\epsilon_j}+1}& \text{if}\ X_i\in B_i\\
\frac{1}{e^{\epsilon_j}+1}& \text{if}\ X_i\notin B_i\\
\end{array}\right.
\end{equation}
which is $\epsilon_j$-LDP. We prove our claim by induction. For the basis step, we can easily verify that $Y_{i}^{1}$ defined in~\eqref{eqn4_2} follows the conditional distribution in~\eqref{eqnd_1}. For the induction step, suppose that our claim is true for $j$. Observe that $Y_{i}^{j+1}=Y_{i}^{j}\oplus U_{i}^{j+1}$. Hence, we have
\begin{equation}
\begin{aligned}
&\text{Pr}\left[Y_{i}^{j+1}=1|X_i\in B_i\right]\\
&\ =\text{Pr}\left[Y_{i}^{j+1}=1|X_i\in B_i,Y_{i}^{j}=1\right]\text{Pr}\left[Y_{i}^{j}=1|X_i\in B_i\right]\\
&\qquad\qquad\qquad+\text{Pr}\left[Y_{i}^{j+1}=1|X_i\in B_i,Y_{i}^{j}=0\right]\text{Pr}\left[Y_{i}^{j}=0|X_i\in B_i\right]\\
&\ =\text{Pr}\left[U_{i}^{j+1}=0\right]\text{Pr}\left[Y_{i}^{j}=1|X_i\in B_i\right]+\text{Pr}\left[U_{i}^{j+1}=1\right]\text{Pr}\left[Y_{i}^{j}=0|X_i\in B_i\right]\\
&\ =\left(1-q_{j+1}\right)\left(1-z_j\right)+q_{j+1}z_j\\
&\ =1-z_{j+1}=\frac{e^{\epsilon_{j+1}}}{e^{\epsilon_{j+1}}+1}
\end{aligned}
\end{equation}
Similarly, we can prove that $\text{Pr}\left[Y_{i}^{j+1}=1|X_i\notin B_i\right]=z_{j+1}=\frac{1}{e^{\epsilon_{j+1}}+1}$. Hence, the proof is completed.

\section{Proof of Lemma~\ref{lemm5_1}}
\label{AppE} 
 
In order to recover $X$ from $Y$ and $U$, it is required that each input database $x\in\left[k\right]$ is mapped to $y$ with a different value of key $U$ for every output $y\in\left[k\right]$. Let $y=x\oplus u$ for all $x\in\left[k\right]$ and $u\in\left[k\right]$, where  $x\oplus y=\left[\left(x+u-2\right)\bmod k \right]+1$. Note that the set $\left[k\right]$ along with the operation $\oplus$ forms a group\footnote{It is exactly the group defined on integers $\lbrace 0,\ldots,k-1\rbrace$ with modulo-$k$ operation, but we subtract $-2$ before taking $\bmod k$ and adding one to fit modulo-$k$ operation with the set $\left[k\right]=\lbrace 1,\ldots,k\rbrace$}. The private mechanism $Q$ is defined as follows
\begin{equation}
Q\left(y|x\right)=q_{u},
\end{equation}
 for $y=x\oplus u$. Note that an input $x$ is mapped to each output $y$ with a different value of the key $U=\left(k-x+2\right)\oplus y$. Moreover, for a given output $y$, we can easily see that each input $x\in\left[k\right]$ is mapped to $y$ with a different value of the key $U$. Hence, it is possible to recover $X$ from $Y$ and $U$. Furthermore, for any two inputs $x,x^{\prime}\in\mathcal{X}$, we have
 \begin{equation}
 \sup_{y\in\left[k\right]}\frac{Q\left(y|x\right)}{Q\left(y|x^{\prime}\right)}\leq \frac{q_{\max}}{q_{\min}}\stackrel{\left(a\right)}{\leq}  e^{\epsilon},
 \end{equation}
 where $q_{\max}=\max\limits_{j\in\left[k\right]}q_j$ and $q_{\min}=\min\limits_{j\in\left[k\right]}q_j$. Step $\left(a\right)$ follows from the assumption that $\frac{q_{\max}}{q_{\min}}\leq e^{\epsilon}$. Thus, the mechanism $Q$ is an $\epsilon$-LDP-Rec mechanism.

\section{Proof of Lemma~\ref{lemm5_2}}
\label{AppF}

Before we present the proof of Lemma~\ref{lemm5_2}, we provide the following lemma whose proof is in Appendix~\ref{AppI}.
\begin{lemma}~\label{lemmf_1}Let $U\in\mathcal{U}=\lbrace u_1,\ldots,u_m\rbrace$ be a random variable with size $m$ having a distribution $\mathbf{q}=\left[q_1,\ldots,q_m\right]$, where $q_1\geq \cdots\geq q_m$. Then, the random variable $U^{\prime}\in\mathcal{U}^{\prime}=\lbrace u_1,\ldots,u_{m-1}\rbrace$ with distribution $\mathbf{q}^{\prime}=\left[q_1^{\prime},\ldots,q_{m-1}^{\prime}\right]$ has an entropy
\begin{equation}
H\left(U\right)\geq H\left(U^{\prime}\right),
\end{equation}
where $q_j^{\prime}=q_j/\left(1-q_m\right)$ for $j\in\lbrace 1,\ldots,m-1\rbrace$.
\end{lemma}
This lemma shows that if we trim the last symbol that has the lowest probability from a distribution, and normalize the remaining probabilities, then we get a distribution that has lower entropy.

The main idea of the proof of Lemma~\ref{lemm5_2} is that we do some reduction steps to get a new random key $U^{\prime}$ with a support size equal to the input size from the random key $U$. In addition, this new random key $U^{\prime}$ has lower entropy than the entropy of the original random key $U$. First, we give an example to illustrate the idea, and then we proceed to the general proof.

\begin{Example} Suppose that a random key $U\in\lbrace 1,2,\ldots,6\rbrace$ has a distribution $\mathbf{q}=\left[q_1,\ldots,q_6\right]$, where $q_1\geq\cdots\geq q_6$. The random key $U$ is used to design an $\epsilon$-LDP-Rec mechanism $Q$ with input $X\in\lbrace 1,2,3\rbrace$. Suppose that there exists an output $y$ such that $X=x$ is mapped to $y$ when $U\in\mathcal{U}_{yx}$, where $\mathcal{U}_{y1}=\lbrace 6\rbrace$, $\mathcal{U}_{y2}=\lbrace 2,3\rbrace$, and $\mathcal{U}_{y3}=\lbrace 1\rbrace$. Hence, $Q\left(y|X=1\right)=q_6$, $Q\left(y|X=2\right)=q_2+q_3$, and $Q\left(y|X=3\right)=q_1$. Let $\mathcal{U}_y=\bigcup_{x\in\left[3\right]}\mathcal{U}_{yx}=\lbrace 1,2,3,6\rbrace$, and $\overline{\mathcal{U}}_y=\mathcal{U}\setminus\mathcal{U}_y=\lbrace 4,5\rbrace$. Let $\tilde{\mathbf{q}}=\left[q_6,q_{2}+q_3,q_1,q_4,q_5\right]$, where the first three elements are $Q\left(y|X=i\right)$ for $i\in\left[3\right]$ and the remaining elements represent $q_u$ for $u\in\overline{\mathcal{U}}_y$. Then, we sort the distribution $\tilde{\mathbf{q}}$ in a descending order to get $\tilde{\mathbf{q}}^{\downarrow}=\left[q_{2}+q_{3},q_{1},q_4,q_5,q_6\right]$, where $\tilde{q}_{i}^{\downarrow}$ denotes the $i$th largest component in $\tilde{\mathbf{q}}$. Consider a random key $\tilde{U}\in\lbrace 1,2,3,4,5\rbrace$ having a distribution $\tilde{\mathbf{q}}^{\downarrow}$. Observe that $H\left(\tilde{U}\right)\leq H\left(U\right)$, since $\tilde{U}$ can be represented as a function of $U$. Furthermore, we have $\frac{q_{2}+q_{3}}{q_1}\leq\frac{q_{2}+q_{3}}{q_4}\leq \frac{q_{2}+q_{3}}{q_6}\leq e^{\epsilon}$, since $Q$ is an $\epsilon$-LDP mechanism, and $q_4\geq q_6$. Consider a random key $U^{\prime}$ having a distribution $\mathbf{q}^{\prime}=\left[\frac{q_2+q_3}{1-\left(q_5+q_6\right)},\frac{q_1}{1-\left(q_5+q_6\right)},\frac{q_4}{1-\left(q_5+q_6\right)}\right]$ obtained by trimming sequentially the last two symbols of the random key $\tilde{U}$. By applying Lemma~\ref{lemmf_1} twice on the distribution $\tilde{\mathbf{q}}^{\downarrow}$, we get that $H\left(U\right)\geq H\left(\tilde{U}\right)\geq H\left(U^{\prime}\right)$. Furthermore, we have $q_{\max}^{\prime}/q_{\min}^{\prime}\leq e^{\epsilon}$. Thus, from Lemma~\ref{lemm5_1}, we can construct an $\epsilon$-LDP-Rec mechanism with input $X\in\left[3\right]$ and an output $Y\in\left[3\right]$ using the random key $U^{\prime}$, where $H\left(U\right)\geq H\left(U^{\prime}\right)$.
\end{Example}

We now present the general proof. Let $U\in\mathcal{U}=\lbrace u_1,\ldots,u_m\rbrace$ be a random key with size $m>k$ having a distribution $\mathbf{q}=\left[q_1,\ldots,q_m\right]$. Without loss of generality, assume that $q_1\geq \cdots\geq q_m$. Let $Q$ be an $\epsilon$-LDP-Rec mechanism designed using a random key $U$ with input $X\in\left[k\right]$ and an output $Y\in\mathcal{Y}$. Let $\mathcal{U}_{yx}\subset\mathcal{U}$ be a subset of keys such that the input $X=x$ is mapped to $Y=y$ when $U\in\mathcal{U}_{yx}$ for all $x\in\left[k\right]$ and $y\in\mathcal{Y}$. As a result the private mechanism $Q$ can be represented by $Q\left(y|X=x\right)=\sum_{u\in\mathcal{U}_{yx}}q_u$.

Observe that for given $y$, we have $\mathcal{U}_{yx}\bigcap \mathcal{U}_{yx^{\prime}}=\phi$, otherwise we cannot recover $X$ from $Y$ and $U$, since there would be $x$ and $x^{\prime}$ mapped to $y$ with the same key value. Let $\mathcal{U}_y=\bigcup_{x\in\left[k\right]}\mathcal{U}_{yx}$, and hence, $\mathcal{U}_y\subseteq \mathcal{U}$. Furthermore, for given $y$, we have $Q\left(y|X=x\right)/Q\left(y|X=x^{\prime}\right)\leq e^{\epsilon}$, since $Q$ is an $\epsilon$-LDP mechanism.

 Consider an output $y\in\mathcal{Y}$ such that $u_1\in\mathcal{U}_y$. Let $\overline{\mathcal{U}}_y=\mathcal{U}\setminus\mathcal{U}_y$ be an indexed set with size $l=|\overline{\mathcal{U}}_y|$, where $\overline{\mathcal{U}}_y\left(j\right)$ denotes the $j$th element in $\overline{\mathcal{U}}_y$. Consider a distribution $\tilde{\mathbf{q}}=\left[\tilde{q}_1,\ldots,\tilde{q}_{l+k}\right]$ designed as follows $\tilde{q}_j=Q\left(y|X=j\right)$ for all $j\in\left[k\right]$ and $\tilde{q}_j=q_{\overline{\mathcal{U}}_y\left(j-k\right)}$ for all $i\in\lbrace k+1,\ldots,k+l\rbrace$. We can sort the distribution $\tilde{\mathbf{q}}$ in a descending order to get $\tilde{\mathbf{q}}^{\downarrow}=\left[\tilde{q}_{1}^{\downarrow},\ldots,\tilde{q}_{l+k}^{\downarrow}\right]$, where $\tilde{q}_{i}^{\downarrow}$ denotes the $i$th largest component in $\tilde{\mathbf{q}}$. Let $\tilde{U}$ be a random key drawn from a distribution $\tilde{\mathbf{q}}^{\downarrow}$. We have the following two properties on the distribution $\tilde{\mathbf{q}}^{\downarrow}$:
\begin{enumerate}
\item $H\left(U\right)\geq H\left(\tilde{U}\right)$.
\item $\frac{\tilde{q}_{1}^{\downarrow}}{\tilde{q}_{k}^{\downarrow}}\leq e^{\epsilon}$.
\end{enumerate}
The first property is straightforward, since the random key $\tilde{U}$ can be represented as a function of $U$. Observe that $u_1\in\mathcal{U}_y$, and $q_1\geq q_u$ for all $u\in\overline{\mathcal{U}}_y$. Hence, $\tilde{q}_{1}^{\downarrow}$ is one of the first $k$ elements in $\tilde{\mathbf{q}}$. Thus, we get
$$\frac{\tilde{q}_{1}^{\downarrow}}{\tilde{q}_{k}^{\downarrow}}\stackrel{\left(a\right)}{\leq} \frac{\tilde{q}_{\max}}{\tilde{q}_{\min}}\leq e^{\epsilon}$$
where $\tilde{q}_{\max}=\max_{j\in\left[k\right]}\tilde{q}_{j}=\tilde{q}_{1}^{\downarrow}$ and $\tilde{q}_{\min}=\min_{j\in\left[k\right]}\tilde{q}_{j}$. If $q_u$ for $u\in\overline{\mathcal{U}}_y$ is one of the first $k$ elements in $\tilde{\mathbf{q}}^{\downarrow}$, i.e, $q_u>\tilde{q}_{\min}$, then inequality $\left(a\right)$ is still valid.

Now, let $U^{\prime}\in\left[k\right]$ be a random key drawn from a distribution $\mathbf{q}^{\prime}=\left[q_1^{\prime},\ldots,q_{k}^{\prime}\right]$, where $q_{j}^{\prime}=\frac{\tilde{q}_{j}^{\downarrow}}{\sum_{j=1}^{k}\tilde{q}_{j}^{\downarrow}}$. Observe that $\mathbf{q}^{\prime}$ is obtained by applying Lemma~\ref{lemmf_1} $l$ times on $\tilde{\mathbf{q}}^{\downarrow}$ to trim sequentially the last $l$ symbols of $\tilde{U}$ that have the lowest $l$ probabilities. Thus, we get that $H\left(U\right)\geq H\left(\tilde{U}\right)\geq H\left(U^{\prime}\right)$. Furthermore, from the second property, we have $q_{\max}^{\prime}/q_{\min}^{\prime}=\frac{\tilde{q}_{1}^{\downarrow}}{\tilde{q}_{k}^{\downarrow}}\leq e^{\epsilon}$. Thus, from Lemma~\ref{lemm5_1}, we can construct an $\epsilon$-LDP-Rec mechanism with input $X\in\left[k\right]$ and an output $Y\in\left[k\right]$ using the random key $U^{\prime}$, and $H\left(U\right)\geq H\left(U^{\prime}\right)$. This completes the proof.

\section{Omitted Details from Section~\ref{Recov-A}}\label{AppF-2} 
First we prove the first necessary condition of Theorem~\ref{Th2_4}. As mentioned in Section~\ref{Recov-A}, we prove this in two parts:
First we show $|\mathcal{Y}|\geq |\mathcal{X}|$ using the recoverability constraint and
then $|\mathcal{U}|\geq |\mathcal{Y}|$ using the privacy constraint.

$|\mathcal{Y}|\geq |\mathcal{X}|$: Observe that the output $Y$ of the private mechanism $Q$ can be represented as a function of the input $X$ and the random key $U$, i.e., $Y=f\left(X,U\right)$. Fix the value of the random key $U=u$ for an arbitrary $u\in\mathcal{U}$. Then, for each value of $x\in\mathcal{X}$, the function $f\left(X,U\right)$ should generate a different output $Y$ in order to be able to recover $X$ from $Y$ and $U$. In other words, each input $x\in\mathcal{X}$ should be mapped to a different output $y\in\mathcal{Y}$ for the same value of the random key $u\in\mathcal{U}$. Otherwise, there exists two inputs mapped with the same key value to the same output. As a result, it is required that the output size is at least the same as the input size: $|\mathcal{Y}|\geq |\mathcal{X}|$.
 
$|\mathcal{U}|\geq |\mathcal{Y}|$: Let $\mathcal{Y}\left(x\right)\subseteq\mathcal{Y}$ be a subset of outputs such that input $X=x$ is mapped with non-zero probability to every $y\in\mathcal{Y}\left(x\right)$. We claim that $\mathcal{Y}\left(x\right)=\mathcal{Y}$ for all $x\in\mathcal{X}$ for any $\epsilon$-LDP-Rec mechanism. In other words, we claim that each input $x\in\mathcal{X}$ should be mapped with non-zero probability to every output $y\in\mathcal{Y}$. We prove our claim by contradiction. Suppose that there exist $x,x^{\prime}\in\mathcal{X}$ such that $\mathcal{Y}\left(x\right)\neq \mathcal{Y}\left(x^{\prime}\right)$. Thus, there exists $y\in\mathcal{Y}\left(x\right)\setminus \mathcal{Y}\left(x^{\prime}\right)$ or $y\in\mathcal{Y}\left(x^{\prime}\right)\setminus \mathcal{Y}\left(x\right)$. Hence, we have $\frac{Q\left(y|x\right)}{Q\left(y|x^{\prime}\right)}\to \infty$ or $\frac{Q\left(y|x^{\prime}\right)}{Q\left(y|x\right)}\to \infty$ which violates the privacy constraints. Therefore, $\mathcal{Y}\left(x\right)= \mathcal{Y}\left(x^{\prime}\right)=\mathcal{Y}$ for all $x,x^{\prime}\in\mathcal{X}$.
However, for a given $x\in\mathcal{X}$, we have $|\mathcal{Y}\left(x\right)|\leq |\mathcal{U}|$,  since each input $x\in\mathcal{X}$ can be mapped with non-zero probability to at most $|\mathcal{U}|$ outputs. Thus, we get that the random key size is at least the same as the output size: $|\mathcal{U}|\geq |\mathcal{Y}| \geq |\mathcal{X}|$. 

Hence, the first condition is necessary to design an $\epsilon$-LDP-Rec mechanism.
This completes the proof of the first necessary condition of Theorem~\ref{Th2_4}.

Now, assuming $q_1\leq q_2\leq \ldots\leq q_k$, we show $q_k/q_1\leq e^{\epsilon}$. This will be required to prove the second necessary condition to prove Theorem~\ref{Th2_4}.

$q_k/q_1\leq e^{\epsilon}$: We prove our claim by contradiction. Suppose that $q_k/q_1>e^{\epsilon}$. Consider a certain output $y\in\mathcal{Y}$ such that there exists $x\in\mathcal{X}$ mapped to $y$ when $U=u_{k}$ with probability $q_{k}$. Note that each sample $x\in\mathcal{X}$ should be mapped using a different value of the key to each output $y\in\mathcal{Y}$ in order to be able to recover the sample $X$ from $Y$ and $U$. In our case, there are $k-1$ remaining inputs to be mapped to $y$ with different values of keys; however, none of these $k-1$ inputs can be mapped to $y$ with $U=u_1$, since $q_{k}/q_{1}>e^{\epsilon}$, which violates the privacy constraint. Hence, we have $k-1$ inputs mapped to $y$ using at most $k-2$ values of keys. Thus, there would exist at least two inputs mapped to output $y$ with the same key value. Therefore, we cannot recover $X$ from $y$ given $U$. As a result, we should have $q_{k}/q_{1}\leq e^{\epsilon}$.

\section{Proof of Lemma~\ref{lemm6_1}}
\label{AppG} 
To simplify the proof, we assume that $\left[k\right]=\lbrace 0,\ldots,k-1\rbrace$. Let $\mathcal{X}^{T}=\left[k\right]^{T}$ denote the input dataset, and  $Y^{T}=\left(Y^{\left(1\right)},\ldots,Y^{\left(T\right)}\right)$ be the output of the private mechanism $Q$ that takes a value from a set $\mathcal{Y}^{T}=\left[k\right]^{T}$. In order to recover $X^{T}$ from $Y^{T}$ and $U$, it is required that each input database $\mathbf{x}\in\mathcal{X}^{T}$ is mapped to each output $\mathbf{y}\in\left[k\right]^{T}$ with a different value of key $U$. Let the random key $U$ be drawn from an $\epsilon$-DP distribution $\mathbf{q}$. Hence, there exists a bijective function $f:\mathcal{X}^{T}\to\left[k\right]^{T}$ such that 
\begin{equation}
\frac{q_{f\left(\mathbf{x}\right)}}{q_{f\left(\mathbf{x}^{\prime}\right)}}\leq e^{\epsilon}.
\end{equation}
for every neighboring databases $\mathbf{x},\mathbf{x}^{\prime}\in\left[k\right]^{T}$. Let $Q$ be a private mechanism defined as follows
\begin{equation}~\label{eqnG_2}
Q\left(\mathbf{y}|\mathbf{x}\right)=q_{f\left(\mathbf{x}\oplus\mathbf{y}\right)}.
\end{equation}
where $\mathbf{x}\oplus\mathbf{y}=\left(x^{\left(1\right)}\oplus y^{\left(1\right)},\ldots,x^{\left(T\right)}\oplus y^{\left(T\right)}\right)$\footnote{We apply elementwise operation $\oplus$ on the vectors $\mathbf{x}$ and $\mathbf{y}$.}, and $x^{\left(j\right)}\oplus y^{\left(j\right)}=\left[\left(x^{\left(j\right)}+y^{\left(j\right)}\right)\bmod k \right]$ which is an addition between $x^{\left(j\right)}$ and $y^{\left(j\right)}$ in a finite group of order $k$. For a fixed $\mathbf{y}\in\mathcal{Y}^{T}$, we can easily see that $f\left(\mathbf{x}\oplus\mathbf{y}\right)\neq f\left(\hat{\mathbf{x}}\oplus\mathbf{y}\right)$ for any $\mathbf{x}\neq \hat{\mathbf{x}}$ and $\mathbf{x},\hat{\mathbf{x}}\in\left[k\right]^{T}$, since $\mathbf{x}\oplus\mathbf{y}\neq
\hat{\mathbf{x}}\oplus\mathbf{y}$ and $f$ is a bijection. Hence, for every output $\mathbf{y}\in\left[k\right]^{T}$, each input database $\mathbf{x}\in\mathcal{X}^{T}$ is mapped to an output $\mathbf{y}$ with a different value of key $U$. Thus, we can recover $X^{T}$ from $Y^{T}$ and $U$. For a fixed $\mathbf{x}\in\left[k\right]^{T}$, we can see that $f\left(\mathbf{x}\oplus\mathbf{y}\right)\neq f\left(\mathbf{x}\oplus\hat{\mathbf{y}}\right)$ for any $\mathbf{y}\neq \hat{\mathbf{y}}$ and $\mathbf{y},\hat{\mathbf{y}}\in\left[k\right]^{T}$, since $\mathbf{x}\oplus\mathbf{y}\neq
\mathbf{x}\oplus\hat{\mathbf{y}}$ and $f$ is a bijection. Hence $Q\left(\mathbf{y|\mathbf{x}}\right)$ is a valid conditional distribution for each $\mathbf{x}\in\left[k\right]^{T}$. It remains to prove that the private mechanism $Q$ given in~\eqref{eqnG_2} is $\epsilon$-DP. In the following, we prove that for every output $\mathbf{y}$, and every neighboring databases $\mathbf{x},\tilde{\mathbf{x}}\in\left[k\right]^{T}$, we have
\begin{equation}~\label{eqnG_6}
\frac{Q\left(\mathbf{y}|\mathbf{x}\right)}{Q\left(\mathbf{y}|\tilde{\mathbf{x}}\right)}\leq e^{\epsilon}
\end{equation}   
Therefore, the private mechanism $Q$ is $\epsilon$-DP. The proof is by induction. For the basis step, observe that each input database $\mathbf{x}$ is mapped to $\mathbf{y}_0=\left[0,\ldots,0\right]$ with probability $q_{f\left(\mathbf{x}\right)}$ for $\mathbf{x}\in\left[k\right]^{T}$. Thus, for every neighboring databases $\mathbf{x},\tilde{\mathbf{x}}\in\left[k\right]^{T}$, we get
\begin{equation}
\frac{Q\left(\mathbf{y}_0|\mathbf{x}\right)}{Q\left(\mathbf{y}_0|\tilde{\mathbf{x}}\right)}=\frac{q_{f\left(\mathbf{x}\right)}}{q_{f\left(\tilde{\mathbf{x}}\right)}}\stackrel{\left(a\right)}{\leq} e^{\epsilon}
\end{equation}
where step $\left(a\right)$ follows from the assumption that the distribution $\mathbf{q}$ satisfies $\epsilon$-DP. For the induction step, suppose there exists an output $\mathbf{y}\in\left[k\right]^{T}$ that satisfies~\eqref{eqnG_6}. Let $\tilde{\mathbf{y}}$ be a neighboring output to $\mathbf{y}$, i.e., $\tilde{\mathbf{y}}$ and $\mathbf{y}$ are different in only one element. Without loss of generality, let $y^{\left(i\right)}\neq \tilde{y}^{\left(i\right)}$ while $y^{\left(j\right)}= \tilde{y}^{\left(j\right)}$ for $j\neq i$. Then, for every neighboring databases $\mathbf{x},\tilde{\mathbf{x}}\in\left[k\right]^{T}$, we get
\begin{equation}
\begin{aligned}
\frac{Q\left(\tilde{\mathbf{y}}|\mathbf{x}\right)}{Q\left(\tilde{\mathbf{y}}|\tilde{\mathbf{x}}\right)}&=\frac{q_{f\left(\mathbf{x}\oplus \tilde{\mathbf{y}}\right)}}{q_{f\left(\tilde{\mathbf{x}}\oplus \tilde{\mathbf{y}}\right)}}\\
&=\frac{q_{f\left(\underline{\mathbf{x}}\oplus \mathbf{y}\right)}}{q_{f\left(\underline{\tilde{\mathbf{x}}}\oplus \mathbf{y}\right)}}\\
&\stackrel{\left(a\right)}{\leq} e^{\epsilon}
\end{aligned}
\end{equation} 
where $\underline{\mathbf{x}}=\left(\underline{x}^{\left(1\right)},\ldots,\underline{x}^{\left(T\right)}\right)$ such that $\underline{x}^{\left(j\right)}=x^{\left(j\right)}$ for $j\neq i$ and $\underline{x}^{\left(i\right)}=\left[\left(k+x^{\left(i\right)}+y^{\left(i\right)}-\tilde{y}^{\left(i\right)}\right)\bmod k \right]$. Similarly, $\underline{\tilde{\mathbf{x}}}=\left(\underline{\tilde{x}}^{\left(1\right)},\ldots,\underline{\tilde{x}}^{\left(T\right)}\right)$ such that $\underline{\tilde{x}}^{\left(j\right)}=\tilde{x}^{\left(j\right)}$ for $j\neq i$ and $\underline{\tilde{x}}^{\left(i\right)}=\left[\left(k+\tilde{x}^{\left(i\right)}+y^{\left(i\right)}-\tilde{y}^{\left(i\right)}\right)\bmod k \right]$. Since $\mathbf{x}$ and $\tilde{\mathbf{x}}$ are neighboring databases, then $\underline{\mathbf{x}}$ and $\underline{\tilde{\mathbf{x}}}$ are also neighboring databases. Step $\left(a\right)$ follows from the assumption that $\mathbf{y}$ satisfy~\eqref{eqnG_6}. From the basic step along with the induction step, we conclude that the mechanism $Q$ given in~\eqref{eqnG_2} is $\epsilon$-DP-Rec mechanism. Hence, the proof is completed.

\subsection{Proof of The First Necessary Condition ($|\mathcal{U}|\geq|\mathcal{Y}^{T}|\geq |\mathcal{X}^{T}|$) of Theorem~\ref{Th2_6}}
We prove it in two parts: first we show $|\mathcal{Y}^{T}|\geq |\mathcal{X}^{T}|$, and then we show $|\mathcal{U}|\geq |\mathcal{Y}^{T}|$.

{\bf $|\mathcal{Y}^{T}|\geq |\mathcal{X}^{T}|$:} Note that the output is a deterministic function of the input and the random key, i.e., $Y^{T}=f(X^T,U)$ for some deterministic function $f$.
This implies that, for any fixed $u\in\mathcal{U}$, the function $f\left(\mathbf{x},u\right)$ should generate a different output $\mathbf{y}\in \mathcal{Y}^{T}$ for different values of $\mathbf{x}\in\mathcal{X}^{T}$, which implies that $|\mathcal{Y}^{T}|\geq |\mathcal{X}^{T}|$.

{\bf $|\mathcal{U}|\geq |\mathcal{Y}^{T}|$:} Let $\mathcal{Y}\left(\mathbf{x}\right)\subseteq\mathcal{Y}^{T}$ be a subset of outputs such that the input $X^{T}=\mathbf{x}$ is mapped with non-zero probability to every $\mathbf{y}\in\mathcal{Y}\left(\mathbf{x}\right)$. We claim that $\mathcal{Y}\left(\mathbf{x}\right)=\mathcal{Y}^{T}$ for all $\mathbf{x}\in\mathcal{X}^{T}$ for any $\epsilon$-DP-Rec mechanism. In other words, we claim that each input $\mathbf{x}\in\mathcal{X}^{T}$ should be mapped with non-zero probability to every output $\mathbf{y}\in\mathcal{Y}^{T}$. We prove our claim by contradiction. Suppose that there exist two neighboring $\mathbf{x},\mathbf{x}^{\prime}\in\mathcal{X}^{T}$ such that $\mathcal{Y}\left(\mathbf{x}\right)\neq \mathcal{Y}\left(\mathbf{x}^{\prime}\right)$. Thus, there exists $\mathbf{y}\in\mathcal{Y}\left(\mathbf{x}\right)\setminus \mathcal{Y}\left(\mathbf{x}^{\prime}\right)$ or $\mathbf{y}\in\mathcal{Y}\left(\mathbf{x}^{\prime}\right)\setminus \mathcal{Y}\left(\mathbf{x}\right)$. Hence, we have $\frac{Q\left(\mathbf{y}|\mathbf{x}\right)}{Q\left(\mathbf{y}|\mathbf{x}^{\prime}\right)}\to \infty$ or $\frac{Q\left(\mathbf{y}|\mathbf{x}^{\prime}\right)}{Q\left(\mathbf{y}|\mathbf{x}\right)}\to \infty$ which violates the privacy constraints. Therefore, $\mathcal{Y}\left(\mathbf{x}\right)= \mathcal{Y}\left(\mathbf{x}^{\prime}\right)=\mathcal{Y}^{T}$ for all $\mathbf{x},\mathbf{x}^{\prime}\in\mathcal{X}^{T}$.
Given $\mathbf{x}\in\mathcal{X}^{T}$, we have that $|\mathcal{Y}\left(\mathbf{x}\right)|\leq |\mathcal{U}|$, where $|\mathcal{U}|$ is the maximum number of possible keys. Thus, the random key size is at least the same as the output size: $|\mathcal{U}|\geq |\mathcal{Y}^{T}|$. Hence, the first condition of Theorem~\ref{Th2_6} is necessary to design an $\epsilon$-DP-Rec mechanism.

\section{Proof of Lemma~\ref{lemm6_2}}
\label{AppH}

Let $g_i=i\left(k\right)^{\left(T-1\right)}$ for $i\in\lbrace 0,\ldots,k\rbrace$. Observe that databases $\mathbf{x}_1,\ldots,\mathbf{x}_{g_1}$ have $x^{\left(1\right)}=1$ and the databases $\mathbf{x}_{g_1+1},\ldots,\mathbf{x}_{g_2}$ have $x^{\left(1\right)}=2$. Generally, the databases $\mathbf{x}_{g_{i-1}+1},\ldots,\mathbf{x}_{g_{i}}$ have $x^{\left(1\right)}=i$. Let $C_i=\sum_{a=g_{i-1}+1}^{g_{i}}P^{\mathbf{y}}_a$ for $i\in\left[k\right]$. Consider the following inequalities that we will prove next
\begin{align}
H\left(\mathbf{P}^{\mathbf{y}}\right)&=-\sum_{a=1}^{k^{T}}P^{\mathbf{y}}_a\log\left(P^{\mathbf{y}}_a\right)\nonumber\\
&=\sum_{i=1}^{k}C_i\left[-\sum_{a=g_{i-1}+1}^{g_{i}}\frac{P^{\mathbf{y}}_a}{C_i}\log\left(\frac{P^{\mathbf{y}}_a}{C_i}\right)\right]-\sum_{i=1}^{k}C_i\log\left(C_i\right)\\
&\geq\sum_{i=1}^{k}C_iH\left(U_{\min,T-1}\right)-\sum_{i=1}^{k}C_i\log\left(C_i\right))~\label{eqn6_2}\\
&\geq \sum_{i=1}^{k} C_i H\left(U_{\min,T-1}\right)+H\left(U_{\min,1}\right)~\label{eqn6_3}\\
&=H\left(U_{\min,T-1}\right)+H\left(U_{\min,1}\right)~\label{eqn6_4}
\end{align}

We begin with inequality~\eqref{eqn6_2}. Observe that the $k^{T-1}$ databases $\mathbf{x}_{g_{i-1}+1},\ldots,\mathbf{x}_{g_{i}}$ have the same value of the first sample $x^{\left(1\right)}=i$, and hence these $k^{T-1}$ databases cover all possible databases in $\mathcal{X}^{T-1}$. Consider a random variable $U^{T-1}$ drawn according to the distribution $\mathbf{P}_{T-1}=\left[\frac{P_{g_{i-1}+1}^{\mathbf{y}}}{C_{i}},\ldots,\frac{P_{g_i}^{\mathbf{y}}}{C_{i}}\right]$. This is a valid distribution with support size $k^{T-1}$. Furthermore, since the distribution $\mathbf{P}^{\mathbf{y}}$ is $\epsilon$-DP, then the distribution $\mathbf{P}_{T-1}$ is also $\epsilon$-DP. From Lemma~\ref{lemm6_1}, the random key $U^{T-1}$ can be used to construct an $\epsilon$-DP-Rec mechanism with the possibility to recover the databases $X^{T-1}=\left(x^{\left(2\right)},\ldots,x^{\left(T\right)}\right)$ from the output of the mechanism and the random key $U^{T-1}$. Hence, we get
\begin{equation}
H\left(U^{T-1}\right)\geq H\left(U_{\min,T-1}\right).
\end{equation}
This proves inequality~\eqref{eqn6_2}. Now, observe that databases $\mathbf{x}_{i},\mathbf{x}_{g_1+i},\ldots,\mathbf{x}_{g_{k-1}+i}$ are neighboring databases for each $i\in\left[k^{T-1}\right]$, since they are only different in the value of the first sample $x^{\left(1\right)}$. Since the mechanism $Q$ is $\epsilon$-DP-Rec, we have
\begin{equation}
e^{-\epsilon}\leq \frac{P_{g_a+i}^{\mathbf{y}}}{P_{g_j+i}^{\mathbf{y}}}\leq e^{\epsilon} \qquad \forall a,j\in\lbrace 0,\ldots,k-1\rbrace
\end{equation}
Thus, we get 
\begin{equation}
e^{-\epsilon}\leq \frac{\sum_{i=g_{a-1}+1}^{g_{a}}P^{\mathbf{y}}_i}{e^{\epsilon}\sum_{i=g_{a-1}+1}^{g_{a}}P^{\mathbf{y}}_i}\leq \frac{C_{a}}{C_{j}}=\frac{\sum_{i=g_{a-1}+1}^{g_{a}}P^{\mathbf{y}}_i}{\sum_{i=g_{j-1}+1}^{g_{j}}P^{\mathbf{y}}_i}\leq \frac{e^{\epsilon}\sum_{i=g_{j-1}+1}^{g_{j}}P^{\mathbf{y}}_i}{\sum_{i=g_{j-1}+1}^{g_{j}}P^{\mathbf{y}}_i}\leq e^{\epsilon} \qquad \forall a,j\in \left[k\right]
\end{equation}
Consider a random key $U^{1}$ that has a distribution $\mathbf{C}=\left[C_1,\ldots,C_k\right]$, where $C_a=\sum_{i=g_{a-1}+1}^{g_{a}}P^{\mathbf{y}}_i$. From Lemma~\ref{lemm5_1}, the random key $U^{1}$ can be used to construct an $\epsilon$-LDP-Rec mechanism with the possibility to recover the sample $X_1$ from the output of the mechanism and the random key $U^{1}$. Hence from Theorem~\ref{Th2_4}, we have
\begin{equation}
H\left(U^{1}\right)\geq H\left(U_{\min,1}\right).
\end{equation}
This proves inequality~\eqref{eqn6_3}, and completes the proof of Lemma~\ref{lemm6_2}.

\section{Proof of Lemma~\ref{lemmf_1}}
\label{AppI}

For the random variable $U^{\prime}$, the distribution $\mathbf{q}^{\prime}=\left[q_{1}^{\prime},\ldots,q_{m-1}^{\prime}\right]$ is given by
\begin{equation}
q^{\prime}_j=\frac{q_j}{1-q_{m}}.
\end{equation}    
Note that the distribution $\mathbf{q^{\prime}}$ is a valid distribution on $U^{\prime}$ since $\sum_{j=1}^{m-1}q^{\prime}_{j}=\sum_{j=1}^{m-1}\frac{q_j}{1-q_{m}}=1$. Now, we can bound the difference between $H\left(U\right)-H\left(U^{\prime}\right)$ as follows
\begin{align}
H\left(U\right)-H\left(U^{\prime}\right)&=\sum_{j=1}^{m-1}q^{\prime}_{j}\log\left(q^{\prime}_{j}\right)-\sum_{j=1}^{m}q_j\log\left(q_j\right)\nonumber\\
&=\sum_{j=1}^{m-1}\frac{q_j}{1-q_{m}}\log\left(\frac{q_j}{1-q_{m}}\right)-\sum_{j=1}^{m}q_j\log\left(q_j\right)\nonumber\\
&=\sum_{j=1}^{m-1}\frac{q_j}{1-q_{m}}\left[\log\left(\frac{q_j}{1-q_{m}}\right)-\log\left(q_j^{\left(1-q_m\right)}\right)\right]-q_m\log\left(q_m\right)\nonumber\\
&=\sum_{j=1}^{m-1}\frac{q_j}{1-q_{m}}\left[-\log\left(\frac{1-q_m}{q_j^{q_m}}\right)\right]-q_m\log\left(q_m\right)\nonumber\\
&> -\log\left(\sum_{j=1}^{m-1}q_j^{\left(1-q_m\right)}\right)-q_m\log\left(q_m\right)~\label{eqni_1}\\ 
&\geq -\left(1-q_m\right)\log\left(1-q_m\right)-q_m\log\left(m-1\right)-q_m\log\left(q_m\right)~\label{eqni_2}\\ 
&\geq \min\left(0,\log\left(\frac{m}{m-1}\right)\right)~\label{eqni_3}\\
&\geq 0 ~\label{eqni_4}
\end{align}
where~\eqref{eqni_1} follows from the fact that $-\log\left(.\right)$ is a strictly convex function and $q_j/1-q_m >0$ for $j\in\left[m-1\right]$. The inequality~\eqref{eqni_2} follows from solving the convex problem 
\begin{equation}~\label{opt_1}
\begin{aligned}
\max_{\lbrace q_j\rbrace_{j=1}^{m-1}}&\  \sum_{j=1}^{m-1}q_j^{\left(1-q_m\right)}\\
s.t.&\ \sum_{j=1}^{m-1}q_j=1-q_m\\
&\ q_j\geq q_m\ \forall j\in\left[m-1\right]
\end{aligned}
\end{equation}
Note that $x^{a}$ is a concave function on $x\in\mathbb{R}_{+}$ for $0\leq a\leq 1$. Therefore, the objective function in~\eqref{opt_1} is concave in $\lbrace q_j\rbrace$. By solving the optimization problem in~\eqref{opt_1}, we get $q_j^{*}=\frac{1-q_m}{m-1}\geq q_m$ for all $j\in\left[m-1\right]$ and $\sum_{j=1}^{m-1}q_j^{\left(1-q_m\right)}\leq \frac{\left(1-q_m\right)^{\left(1-q_m\right)}}{\left(m-1\right)^{\left(-q_m\right)}}$. Since $\log\left(x\right)$ is a monotonic function, we get $-\log\left(\sum_{j=1}^{m-1}q_j^{\left(1-q_m\right)}\right)\geq -\left(1-q_m\right)\log\left(1-q_m\right)-q_m\log\left(m-1\right)$. The inequality~\eqref{eqni_3} follows from the fact that $-\left(1-q_m\right)\log\left(1-q_m\right)-q_m\log\left(m-1\right)-q_m\log\left(q_m\right)=H\left(q_m\right)-q_m\log\left(m-1\right)$ is a concave function of $q_m$. The minimum of a concave function is one of the vertices, where $q_m\in\lbrace 0,\frac{1}{m}\rbrace$. Hence, the proof is completed.

%%%%%%%%%%%%%%%%%%%%%%%%%%%%%%%%%%%%%%%%%%%%%%%%%%%%%%%%%%%%%%%%%%%%%%%%%%%%%%%
%%%%%%%%%%%%%%%%%%%%%%%%%%%%%%%%%%%%%%%%%%%%%%%%%%%%%%%%%%%%%%%%%%%%%%%%%%%%%%%
\end{document}